\documentclass[%
 reprint,
 amsmath,amssymb,
 aps,
]{revtex4-2}

\usepackage{graphicx}%
\usepackage{dcolumn}%
\usepackage{bm}%

\usepackage{amsmath,amsthm,amssymb,amsfonts}
\usepackage{graphicx}
\usepackage[colorlinks=true, allcolors=blue]{hyperref}
\usepackage{natbib}
\usepackage{tikz}
\usetikzlibrary{arrows.meta, positioning, calc}
\usepackage{graphicx}
\usepackage{soul}
\usepackage{booktabs} 
\usepackage{multirow}
\usepackage{float} 
\newcommand{\ac}{{\check a}}
\newcommand{\zetac}{{\check z}}

\newcommand{\wc}{{\check w}}
\newcommand{\qc}{{\check q}}
\newcommand{\pc}{{\check p}}

\newtheorem{theorem}{Theorem}

\newtheorem{lemma}{Lemma}
\newtheorem{corollary}{Corollary}
\newtheorem{proposition}{Proposition}
\newtheorem{remark}{Remark}
\newtheorem{definition}{Definition}

\begin{document}

\preprint{APS/123-QED}

\title{Quantum Wiener architecture for quantum reservoir computing}%

 \author{Alessio Benavoli}
 \email{alessio.benavoli@tcd.ie}
\affiliation{School of Computer Science and Statistics, Trinity College Dublin.\vspace{1mm} \\
 Trinity Quantum Alliance, Unit 16,\\ 
 Trinity Technology and Enterprise Centre,\\
 Pearse Street, Dublin 2, Ireland
}%

\author{Felix Binder}%
 \email{felix.binder@tcd.ie}
 \affiliation{School of Physics, Trinity College Dublin.\vspace{1mm}\\
 Trinity Quantum Alliance, Unit 16, Trinity Technology and Enterprise Centre, Pearse Street, Dublin 2, Ireland
}%

\date{\today}%

\begin{abstract}
This work focuses on  quantum reservoir computing and, in particular, on quantum Wiener architectures (qWiener), consisting of quantum linear dynamic networks with weak continuous measurements and classical nonlinear static readouts. We provide the first rigorous proof that qWiener systems retain the fading-memory property and universality of classical Wiener architectures, despite quantum constraints on linear dynamics and measurement back-action. Furthermore, we develop a kernel-theoretic interpretation showing that qWiener reservoirs naturally induce deep kernels, providing a principled framework for analysing their expressiveness. We further characterise the simplest qWiener instantiation, consisting of concatenated quantum harmonic oscillators, and show the difference with respect to the classical case. Finally, we empirically evaluate the architecture on standard reservoir computing benchmarks, demonstrating systematic performance gains over prior classical and quantum reservoir computing models.

\end{abstract}

\maketitle

\section{Introduction}
Reservoir computing (RC) \citep{jaeger2001echo,maass2002real} is a computational framework in which temporal information is processed by driving a dynamical system with input signals while training only the output layer. Despite this strong restriction, RC architectures such as  liquid state machines \citep{jaeger2001echo} and echo-state networks \citep{maass2002real} are provably universal for approximating fading-memory systems, i.e., dynamical systems depending only on a finite past history of inputs. 

Recent efforts have extended RC into the quantum regime \citep{Fujii2017,PhysRevApplied.11.034021,kutvonen2020optimizing,chen2020temporal,nokkala2021gaussian,mujal2021,martinez2023information,hu2024overcoming,paparelle2025experimental,selimovic2025experimental}, motivated both by the high dimensionality of quantum systems and by the prospect of exploiting near-term noisy intermediate-scale quantum platforms. Quantum reservoir computing (QRC) aims to process temporal data streams through quantum dynamics, and extract classical information through suitable measurements. 
QRC have achieved promising results  in addressing temporal and classification tasks \citep{ghosh2019quantum,markovic2020quantum,bravo2022quantum,sakurai2022quantum,nokkala2023online,dudas2023quantum,hayashi2023impact}. 

In \citep{nokkala2021gaussian}, it was shown that even linear quantum networks restricted to Gaussian states can, in principle, provide universality for RC. This argument relies on the classical result  \citep{boyd2003fading} that any continuous fading-memory dynamical system can be represented by a single-input, multi-output linear time-invariant (LTI) system followed by a polynomial nonlinearity. This architecture is known as the (classical) Wiener architecture \citep{wiener1966nonlinear}, depicted in Figure \ref{fig:1}. By suitably combining the outputs of finitely many quantum linear networks with a polynomial readout, \cite{nokkala2021gaussian} claimed that universality can be attained by this QRC architecture. In this paper, we refer to it as the \textit{quantum Wiener architecture (qWiener)}. This claim is particularly interesting, as continuous-variable quantum optical systems provide efficient platforms for its realization \citep{nokkala2021gaussian,garcia2023scalable,nokkala2024retrieving}. Moreover, by employing weak measurements, this architecture mitigates the loss of past input memory \citep{nokkala2021gaussian,mujal2023time,yasuda2023quantum} that typically results from the measurement process.

However, there are fundamental differences  between a classical and a quantum architecture:  
\begin{itemize}
    \item Unlike classical LTI systems, quantum linear networks must preserve the canonical commutation relations, which impose stringent constraints (e.g., unitarity) on their parameters  and their dynamics.  
    \item Weak measurements also perturb the quantum system, introducing back-action that may compromise both fading-memory and universality properties.  
\end{itemize}  
These important aspects were not explicitly addressed in \cite{nokkala2021gaussian}.  In this work, our contributions are as follows:  
\begin{itemize}
    \item We rigorously prove the universality of the \textit{qWiener} architecture: a system of quantum linear networks with weak continuous measurements, whose outcomes are processed through classical polynomial readouts. We prove that, despite quantum constraints in their dynamics and measurement back-action, \textit{qWiener} systems retain the fading-memory property and universal approximation capabilities of classical Wiener architectures.  
    \item We develop a kernel-theoretic perspective on classical Wiener architectures for RC, linking their expressive power to kernel methods and Gaussian processes.  We show
that the reservoir naturally induces kernel functions
capturing correlations between input sequences.
This leads to a deep-kernel interpretation of classical Wiener architectures. Similarly, we show that, with some important differences (due to the need of preserving the canonical commutation relations), just as
randomised LTI dynamical systems with nonlinear readouts yield deep kernels in classical RC, qWiener architectures combining randomised quantum linear networks  with classical nonlinear readouts also generate deep kernels. This allows us to bridge QRC
with modern kernel methods in machine learning, providing a novel approach to assess their expressiveness instead, for instance, of using \textit{information processing capacity} \citep{dambre2012information,nokkala2021gaussian}.
    \item Leveraging this kernel-based viewpoint, we analyse the expressive power of the simplest \textit{qWiener} instantiation, where the quantum linear block is the concatenation of $n_c$ quantum harmonic oscillators, each linearly coupled to an external field. This allows us to highlight the difference between the classical and quantum case. We further show that the resulting system is completely passive and, therefore,  implementable using only passive optical components. Moreover, we prove that the weak measurement does not disturb its dynamics but merely adds classical Wiener noise to the output. This result exploits the quantum nature of the proposed QRC architecture, in particular the fact that an observed quantum linear network evolves according to the quantum Kalman filter equations.
\item We evaluate the simplest \textit{qWiener} architecture against the model proposed in \citep{nokkala2021gaussian}, as well as a classical echo-state network (used as a baseline), on three benchmark reservoir computing tasks: parity-check, NARMA10, and delay. Our results demonstrate that the proposed approach consistently outperforms the method in \citep{nokkala2021gaussian}. In particular, we show empirically (consistently with our theoretical derivation of the corresponding deep kernel) that the performance of our model systematically improves as the number of harmonic oscillators in the linear network increases -- a property that does not hold for \citep{nokkala2021gaussian}.

\item Finally, we exploit our universality result for the qWiener architecture to show that it can be designed to have a prescribed long-term memory capacity, and a computational basis related to the Legendre polynomials, using a result derived in \cite{voelker2019legendre} for classical reservoir computing.
\end{itemize}

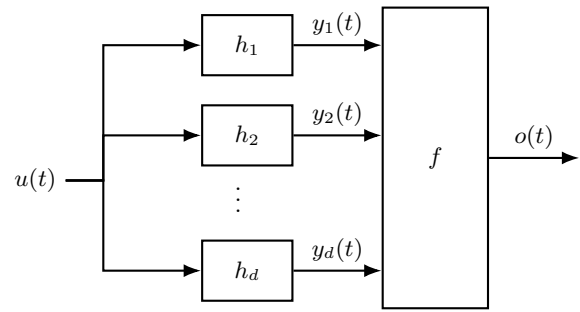
\begin{figure}
\centering
\begin{tikzpicture}[node distance=1.2cm and 1.8cm, thick, >=Latex]
\begin{scope}[xshift=-1cm]
\node (input) {$u(t)$};
\node[draw, rectangle, minimum width=1.2cm, minimum height=0.8cm, right=of input, yshift=1.8cm] (g1) {$h_1$};
\node[draw, rectangle, minimum width=1.2cm, minimum height=0.8cm, right=of input, yshift=0.6cm] (g2) {$h_2$};
\node[draw=none, right=of input, yshift=-0.15cm] (dots) {~~~$\vdots$};
\node[draw, rectangle, minimum width=1.2cm, minimum height=0.8cm, right=of input, yshift=-1.2cm] (gm) {$h_d$};
\draw[->] (input.east) -- ++(0.5,0) |- (g1.west);
\draw[->] (input.east) -- ++(0.5,0) |- (g2.west);
\draw[->] (input.east) -- ++(0.5,0) |- (gm.west);

\node[draw, rectangle, minimum width=1.4cm, minimum height=4.0cm, right=4.2cm of input, yshift=0.3cm] (combiner) {$f$};

\draw[->] (g1.east) -- (combiner.west |- g1.east) node[midway, above] {$y_1(t)$};
\draw[->] (g2.east) -- (combiner.west |- g2.east) node[midway, above] {$y_2(t)$};
\draw[->] (gm.east) -- (combiner.west |- gm.east) node[midway, above] {$y_d(t)$};

\draw[->] (combiner.east) -- ++(1.2,0) node[midway, above] {$o(t)$};
\end{scope}
\end{tikzpicture}
\caption{Wiener architecture: connection of $d$ parallel linear SISO dynamical systems ($h_i$) with a  memoryless nonlinear function $f$. $u(t),o(t)$ are the input and output signal respectively.}
\label{fig:1}
\end{figure}

\section{Background and Novel Preliminaries}
In this section, we introduce the notation and present both established and novel background results that form the foundation of our main contributions. While extensive, this material is essential, as these results will be directly employed in developing our core findings in the subsequent sections.

\subsection{Classical Wiener architecture and connection to kernel methods}
\label{sec:kernel_methods}
A classical Wiener architecture consists of multiple parallel classical Single Input Single Output (SISO)  LTI  systems whose outputs are fed into a static nonlinearity. When this architecture is used in RC, only the final (readout) layer is trained. Consequently, the parameters of the parallel SISO LTI systems must be selected beforehand.
There are two main approaches for choosing these parameters: (1) random generation from specific probability distributions; (2) deliberate design to achieve a desired dynamic behavior of the LTI system \citep{voelker2019legendre}.
In the first approach, the main challenge is to select distributions that yield a system with sufficient expressiveness. This expressiveness can be assessed by examining the statistical properties of the nonlinear basis functions obtained through combinations of randomly generated LTI systems. In this paper, we concentrate on this approach, although we will provide an example of the second approach in the numerical experiments section.

A unifying perspective for analysing the first approach  is offered by kernel theory and its Bayesian counterpart, Gaussian processes \citep{rasmussen2010gaussian}.  We recall that:

\begin{definition}[Gaussian Process]
A GP is a collection of random variables, any finite number of which have a joint Gaussian distribution. Formally, a real-valued function $f(\mathbf{x})$  with $\mathbf{x} \in \mathbb{R}^n$ is said to be GP distributed, denoted as $f \sim \mathcal{GP}(\mu, k)$, if for any finite set of inputs $\mathbf{x}_1, \dots, \mathbf{x}_n$, the corresponding function values $f(\mathbf{x}_1), \dots, f(\mathbf{x}_n)$ follow a multivariate normal distribution:
\[
f(\mathbf{x}) \sim N\left(
\begin{pmatrix} \mu_{\boldsymbol{\theta}}(\mathbf{x}_1) \\ \vdots \\ \mu_{\boldsymbol{\theta}}(\mathbf{x}_n) \end{pmatrix},
\begin{pmatrix}
K_{\boldsymbol{\theta}}(\mathbf{x}_1, \mathbf{x}_1) & \cdots & K_{\boldsymbol{\theta}}(\mathbf{x}_1, \mathbf{x}_n) \\
\vdots & \ddots & \vdots \\
K_{\boldsymbol{\theta}}(\mathbf{x}_n, \mathbf{x}_1) & \cdots & K_{\boldsymbol{\theta}}(\mathbf{x}_n, \mathbf{x}_n)
\end{pmatrix}
\right),
\]
where $\mu_{\boldsymbol{\theta}}(\mathbf{x}) = \mathbb{E}[f(\mathbf{x})]$ is the mean function and $K_{\boldsymbol{\theta}}(\mathbf{x}, \mathbf{x}') = \mathbb{E}[(f(\mathbf{x}) - \mu(\mathbf{x}))(f(\mathbf{x}') - \mu(\mathbf{x}'))]$ is the covariance function (or kernel) for any $\mathbf{x},\mathbf{x}' \in \mathbb{R}^n$ . The mean and covariance functions fully define a GP. They usually depend on a vector of hyparparameters ${\boldsymbol{\theta}} \in \mathbb{R}^{n_p}$.
\end{definition}

Generally, the mean function $\mu_{\boldsymbol{\theta}}$ of a GP is assumed to be zero, so its covariance function entirely characterises the properties and behavior of the process.

Now consider a deterministic SISO  LTI model:
\begin{align}
\label{eq:sys1}
d \mathbf{x}(t)&=A\mathbf{x}(t)dt+B u(t)dt,\\
\label{eq:sys2}
y(t)&=C\mathbf{x}(t)+Du(t),
\end{align}
where $t\geq 0$ is time,  $u(t) \in \mathbb{R}$ is the input signal, $y(t) \in \mathbb{R}$ is the output signal and $\mathbf{x}(t) \in \mathbb{R}^{n_x}$ is the state. The matrices are $A \in \mathbb{R}^{n_x \times n_x}$, $B \in \mathbb{R}^{n_x \times 1}$, $C \in \mathbb{R}^{1 \times n_x}$, $D \in \mathbb{R}$. We assume  the system is stable, that is the eigenvalues of $A$ have strictly negative real parts. The output of a SISO LTI system can be written as
\begin{equation}
y(t) = \int_{0}^t h(t-\tau)u(\tau) d\tau,
\end{equation}
where $h$ is the impulse response  (we have assumed that $u(t)=0$ for $t<0$).\footnote{Here, we considered a deterministic, classical, stable SISO LTI system. For a stochastic LTI system with classical Wiener noise, the output $y(t)$   becomes stochastic. In this case, the convolution $\int_{0}^t h(t-\tau)u(\tau) d\tau$ coincides with the mean of $y(t)$.} The impulse response of the above system depends on  the matrix parameters as follows: $h(t)=Ce^{At}B+D\delta(t)$, where $e^M$ denotes the \textit{matrix exponential} of $M$ and $\delta(t)$ is the Dirac's delta.
It is well known that the term $Ce^{At}B$ can be written as the linear combination of decaying sinusoidal functions:\footnote{We do not consider here the case where $A$ has eigenvalues with multiplicity greater than one.}
\begin{equation}
\label{eq:h}
h(t) = \sum_{i=1}^{n_c} e^{-\alpha_{i} t} (\beta_{i}\cos(\omega_{i} t)-\gamma_{i}\sin(\omega_{i} t))+D\delta(t),
\end{equation}
where  $\alpha_{i}>0$, $\omega_{i} \in \mathbb{R}$ are the decaying rate, and, respectively,  the angular frequency of the i-th damped
oscillation, and $\beta_{i},\gamma_{i} \in \mathbb{R}$ are scaling parameters. The number of components is $n_c=n_x/2$, where we have assumed that $n_x\geq 2$ is an even number.  For simplicity, in this section we  assume that $D=0$.

We now consider the case where the parameters $\alpha_i,\omega_i,\beta_i,\gamma_i$ are randomly generated from some specific probability distributions.
We report results in dynamic system identification \citep{zorzi2018harmonic}, which allow us to define distributions over these parameters that generate universal kernel functions.

\begin{lemma}[\citep{zorzi2018harmonic}]
\label{lem:1} For $i=1,\dots,n_c$, assume that the parameters $\alpha_{i},\omega_{i}$ of the impulse response \eqref{eq:h} are jointly sampled from a PDF $\phi_{\boldsymbol{\theta}}(\alpha,\omega)$ (depending on some hyperparameters ${\boldsymbol{\theta}}$),  and $\beta_{i},\gamma_{i} \sim N(0,1/\sqrt{n_{c}})$ independently.
Then, given the parameters $\{\alpha_{i},\omega_{i}\}_{i=1}^{n_c}$, the impulse response $h(t)$ is GP distributed with zero-mean function and covariance
 function
\begin{equation}
\label{eq:fintiecov}
 E[h(t)h(s)]   =\frac{1}{n_c}\sum_{i=1}^{n_c}  e^{-\alpha_{i} (t+s)} \cos(\omega_{i} (t-s)),
\end{equation}
for all time instants $t,s$.
Then, in the limit as $n_c \rightarrow \infty$, $E[h(t)h(s)]$ tends to:
\begin{equation}
\label{eq:covlimit}
\begin{aligned}
 K_{\boldsymbol{\theta}}(t,s)=\int_{0}^{\infty} \int_{-\infty}^{\infty} \phi_{\boldsymbol{\theta}}(\alpha,\omega)e^{-\alpha (t+s)} \cos(\omega (t-s))d\alpha d\omega.
 \end{aligned}
\end{equation}
Moreover, at the limit, $h(t) \sim \mathcal{GP}(0,K_{\boldsymbol{\theta}}(t,s))$.
\end{lemma}
Note that, $\phi_{\boldsymbol{\theta}}(\alpha,\omega)$ is a generalised power spectral density (that is, a nonnegative function  satisfying $\phi_{\boldsymbol{\theta}}(\alpha,\omega)=\phi_{\boldsymbol{\theta}}(\alpha,-\omega)$), which has been normalised to integrate to one. We can then interpret it as a PDF and, therefore, the limit
in \eqref{eq:covlimit} follows by the strong law of large numbers.\footnote{
  Lemma \ref{lem:1} relies on a different convergence result than that originally used in \citep{zorzi2018harmonic}. Specifically, we employ the strong law of large numbers of Monte Carlo integration, whereas \citep{zorzi2018harmonic} relies on the convergence of Riemann sums.}

\begin{proposition}[\citep{zorzi2018harmonic}]
\label{prop:1}
Assume that the impulse response $h(t)$ is GP distributed with zero mean function and covariance function \eqref{eq:covlimit}. Depending on the choice of the probability density function $\phi_{\boldsymbol{\theta}}(\alpha,\omega)$,  from \eqref{eq:covlimit} we can derive  that
\begin{enumerate}
\item if $\phi_{\boldsymbol{\theta}}(\alpha,\omega)=\delta(\alpha)\mathcal{N}(\omega; \omega_0, \frac{1}{\ell^2})$ then 
\begin{equation}
    K_{\boldsymbol{\theta}}^{(SE)}(t,s)=e^{-\frac{(t-s)^2}{2\ell^2}}\cos(\omega_0(t-s));
\end{equation}
\item if $\phi_{\boldsymbol{\theta}}(\alpha,\omega)=\frac{\alpha}{\pi(\omega^2+\alpha)}\text{Unif}(\alpha; \frac{\alpha_m}{2},\frac{\alpha_M}{2})$ then the kernel is 
\begin{equation}
\label{eq:TC}
K_{\boldsymbol{\theta}}^{(TC)}(t,s)=\frac{e^{-\alpha_m \max(t,s)}-e^{-\alpha_M \max(t,s)}}{2\max(t,s)(a_M-a_m)};
\end{equation}
\end{enumerate}
where $\delta(\alpha)$ is the Dirac's delta. The vector ${\boldsymbol{\theta}}$ includes the hyperparameters, such as, for instance,    $\ell>0$ (lengthscale), $\omega_0 \in [0,2\pi]$ and  $0<\alpha_m<\alpha_M$.
\end{proposition}
Observe that, for $\omega_0=0$, the first kernel function coincides with the
Squared-Exponential kernel, which is a standard kernel in machine learning due to its universal function approximation  property \citep{rasmussen2010gaussian}. For $\omega_0\neq0$, one obtains one component of a spectral mixture kernel \citep{wilson2013gaussian}, which is also universal. The second kernel is known as the integrated
tuned-correlated kernel \citep{chen2012estimation,pillonetto2016regularized}, whose universality holds for representing the impulse response of a linear dynamical system. Therefore, in the sequel, we will focus on this  kernel.

\begin{remark}
Proposition \ref{prop:1} shows that, for suitable choices of the distributions of the random parameters $\alpha_i$ and $\omega_i$, a reservoir composed of $n_c$ SISO LTI systems gives rise to a Gaussian joint distribution for its output signal, whose covariance coincides with standard kernels commonly used in machine learning. In particular, by appropriately selecting the distribution of the parameters, one can recover well-known kernels such as the squared exponential or the integrated tuned-correlated kernel (and many others), thereby providing a principled interpretation of reservoir features in terms of reproducing kernel Hilbert spaces \citep{paulsen2016introduction}.
\end{remark}

The GP prior on $h$ induces a GP prior on $y(t)=\int_{0}^t h(t-\tau)u(\tau)d\tau$ (because the integral is a linear operator and preserves gaussianity). In particular, we have
\begin{equation}
\label{eq:GPony}
   y(t) \sim \mathcal{GP}\left(0, \int_{0}^t \int_{0}^{t'} K_{\boldsymbol{\theta}}(t-\tau_1,t'-\tau_2)u(\tau_1)u(\tau_2)\right).
\end{equation}
While in the previous proposition we considered the limiting case $n_c \rightarrow \infty$, in the context of reservoir computing we focus on the finite case, where $n_c$ is finite.\footnote{The finite case is also of interest in machine learning as technique in kernel approximation for scalability of kernel methods \citep{rahimi2007random}.} In this regime, the impulse response can be \textit{physically realised} by a state-space model (i.e., by a physical system), since
$$
h(t) = \sum_{i=1}^{n_c} e^{-\alpha_i t} \left( \beta_i \cos(\omega_i t) - \gamma_i \sin(\omega_i t) \right) = C e^{At} B,
$$
where the matrices $A,B,C$ include the parameters $\{\alpha_i,\beta_i,\gamma_i,\omega_i\}_{i=1}^{n_c}$.

\begin{proposition}
\label{prop:linBlocks}
Assume the parameters $(\alpha_i,\omega_i)_{i=1}^{n_c}$ are sampled from $\phi_{\boldsymbol{\theta}}(\alpha,\beta)$. Define
\begin{equation}
\label{eq:h_i}
h_j(t)=\sum_{i=1}^{n_c} e^{-\alpha_i t} \left( \beta^{(j)}_i \cos(\omega_i t) - \gamma^{(j)}_i \sin(\omega_i t) \right),
\end{equation}
where $\beta^{(j)}_i,\gamma^{(j)}_i \sim N(0,1/\sqrt{n_c})$ independently for $i=1,\dots,n_c$ and $j=1,\dots,d$. Denote with $y_j(t)=\int_{0}^t h_j(t-\tau)u(\tau)d\tau$.
Define ${\bf y}(t)=\frac{1}{\sqrt{d}}(y_1(t),\dots,y_d(t))^\top$, then  we have that for a large $d$ and $n_c$,
\begin{equation}
\label{eq:yy}
{\bf y}^\top(t){\bf y}(t') \approx \int_{0}^t\int_{0}^{t'}  K_{{\boldsymbol{\theta}},n_c}(t-\tau_1,t'-\tau_2)u(\tau_1)u(\tau_2)d\tau_1d\tau_1,
\end{equation}
where $K_{{\boldsymbol{\theta}},n_c}(\cdot,\cdot)$ is the covariance defined in \eqref{eq:covlimit}.
\end{proposition}
The proof of this proposition is provided in Appendix \ref{app:proofs}, which also contains the proofs of all results stated without explicit references. Proposition \ref{prop:linBlocks}, which follows directly from Lemma \ref{lem:1} combined with \eqref{eq:GPony},  shows how we can use $d$ parallel LTI systems to produce a vector ${\bf y}(t)$, whose squared euclidean norm approximates the covariance function of the GP in \eqref{eq:GPony}.
{
\begin{remark}
The rate of convergence of the left hand side of \eqref{eq:yy} to the right hand side is $1/\sqrt{r}$ with $r=\min(d,n_c)$, which follows by the properties of Monte Carlo integration.
\end{remark}
}
Now we are ready to define the last component of the classical Wiener architecture, that is a memoryless nonlinearity. Assume we use the vector ${\bf y}(t)$  as input for the memoryless function $f$, as shown in Figure \ref{fig:1}. The output is $o(t)=f(y_1(t),\dots,y_d(t))$ at each time instant $t$. Since we do not know $f$ we put  a GP prior on it, that is $f \sim \mathcal{GP}(0,K^o_{\boldsymbol{\vartheta}}({\bf y}(t),{\bf y}(t')))$. Note that, contrary to the previous case, the kernel has the vector ${\bf y}(t)$ as input and it is memoryless, because the input only depends on the value of ${\bf y}(t)$ at time $t$.
\begin{proposition}
\label{prop:deepKernel}
Assume    $f \sim \mathcal{GP}(0,K^o_{\boldsymbol{\vartheta}}({\bf y}(t),{\bf y}(t')))$ with either a polynomial kernel \citep[Sec.4.2]{rasmussen2010gaussian}:
\begin{equation}
\label{eq:finitedeepkernelpoly}
K^o_{\boldsymbol{\vartheta}}(\mathbf{y}(t), \mathbf{y}(t')) = \sigma_f^2 (\mathbf{y}(t)^\top
          \mathbf{y}(t') + c)^{deg},
\end{equation}
with ${\boldsymbol{\vartheta}}=(\sigma_f,c)$ denoting the scale and off-set hyperparameter and $\text{deg}$ the degree, or a squared-exponential kernel \citep[Sec.4.2]{rasmussen2010gaussian}:
\begin{equation}
\label{eq:finitedeepkernel}
K^o_{\boldsymbol{\vartheta}}(\mathbf{y}(t), \mathbf{y}(t')) = \sigma_f^2 \exp\left( -\frac{1}{2\ell^2} \|\mathbf{y}(t) - \mathbf{y}(t')\|^2 \right),
\end{equation}
with ${\boldsymbol{\vartheta}}=(\sigma_f,\ell)$ denoting the scale and lengthscale hyperparameter. Then, for a large $d$ and $n_c$, we have that, for the polynomial kernel
\begin{equation}
\label{eq:deepkernelpoly}
\begin{aligned}
&K^o_{\boldsymbol{\vartheta}}(\mathbf{y}(t), \mathbf{y}(t'))\approx  \sigma_f^2 \Big(\\
&\int_{0}^{t}\int_{0}^{t'}   K_{{\boldsymbol{\theta}},n_c}(t-\tau_1,t'-\tau_2)u(\tau_1)u(\tau_2)d\tau_1d\tau_1 + c \Big)^{deg},
\end{aligned}
\end{equation}
or, for the squared-exponential kernel,
\begin{equation}
\label{eq:deepkernel}
\begin{aligned}
&K^o_{\boldsymbol{\vartheta}}(\mathbf{y}(t), \mathbf{y}(t')) \approx  \sigma_f^2 \exp\Bigg(\\
&-\frac{1}{2\ell^2} \int_{0}^t\int_{0}^t   K_{{\boldsymbol{\theta}},n_c}(t-\tau_1,t-\tau_2)u(\tau_1)u(\tau_2)d\tau_1d\tau_1\\
&-\frac{1}{2\ell^2} \int_{0}^{t'}\int_{0}^{t'}   K_{{\boldsymbol{\theta}},n_c}(t'-\tau_1,t'-\tau_2)u(\tau_1)u(\tau_2)d\tau_1d\tau_1\\
&+\frac{2}{2\ell^2} \int_{0}^{t}\int_{0}^{t'}   K_{{\boldsymbol{\theta}},n_c}(t-\tau_1,t'-\tau_2)u(\tau_1)u(\tau_2)d\tau_1d\tau_1\Bigg) .
\end{aligned}
\end{equation}
\end{proposition}
The proof provided in Appendix \ref{app:proofs} specialises the construction given in \citep{cho2009kernel} to a Wiener architecture.
\begin{remark}
  Proposition \ref{prop:deepKernel} shows that an input signal $u(t)$, when passed through $d$ stable, SISO LTI state-space models -- each differing only in their sampled parameters $\beta^{(j)}_i, \gamma^{(j)}_i \sim \mathcal{N}(0, 1/\sqrt{n_c})$, independently for $i = 1, \dots, n_x$ and $j = 1, \dots, d$ -- produces a feature vector ${\bf y}(t)=\frac{1}{\sqrt{d}}(y_1(t),\dots,y_d(t))^\top$. When this feature vector is used within either a polynomial or a squared-exponential kernel, the resulting covariance function converges to a deep kernel in the limit as $d \rightarrow \infty$, thereby providing a principled interpretation of reservoir features in terms of reproducing kernel Hilbert spaces.  
\end{remark}
Note that the polynomial kernel is relevant here because, in Wiener's original proof of universality, the final layer involves a polynomial nonlinearity. However, the polynomial kernel can also be replaced with other nonlinear kernels that exhibit better extrapolation properties  (such as the squared-exponential one), that is, kernels whose induced function spaces generalise more reliably outside the range of the training data, rather than being dominated by rapidly growing polynomial terms.

In summary, because the classical Wiener architecture induces a deep kernel, the model can be trained on a given dataset by applying standard GP regression to the final GP layer. The computations involved in GP regression are analogous to those in linear regression with nonlinear basis functions. This effectively implements a classical RC architecture, where only the last layer is trained.
In particular, the training dataset is:
\[
\mathcal{D}_m = \{ u(k \Delta), \, \mathbf{y}(k \Delta), \, o(k \Delta) \}_{k=1}^m,
\] 
where $\mathbf{y}(k \Delta) \in \mathbb{R}^d$ is a feature vector obtained by passing the input signal $u(t)$ through $d$ SISO LTI state-space models. The output $o(t)$ is modelled as
\[
o(t_i) = f({\mathbf{y}}(t_i)) + \varepsilon_i, \quad \varepsilon_i \sim \mathcal{N}(0, \sigma^2), \quad t_i = i \Delta,
\]
where $\Delta$ is the sampling interval and $\sigma^2$ is the variance of the measurement noise. Analytical derivations for the computation of the Gaussian posterior of $f$ and the GP predictions are provided in Appendix~\ref{app:gp_posterior}.

\begin{remark}
By mapping a classical Wiener architecture to a deep kernel, we can characterise its input–output behavior within a probabilistic framework, linking deep dynamical systems to compositional kernel models. This will allow us, in the next section, to understand the differences in expressivity between the classical and quantum case.
\end{remark}

\subsection{Quantum  systems}
In the quantum Wiener architecture we propose, classical SISO LTI systems are replaced with quantum linear dynamical systems. Before introducing these systems in the next section, we first establish some preliminary notation. Consider an n-mode quantum system with annihilation
and creation operators $\ac_j,\ac_j^*$ for $j=1,2,\dots,n$, obeying the
boson commutation relations
$[\ac_j,\ac_k^*] =\delta_{jk}$, $[\ac_j,\ac_k] =[\ac_j^*,\ac_k^*] =0$, 
where $[a,b]=ab-ba$ and $\delta_{jk}$ is the Kronecker delta. The operators $\ac_j,\ac_j^*$ can equivalently be described in terms of the real quadrature components (e.g., position $\qc_j$ and momentum $\pc_j$):
\begin{align}
\label{eq:aprel}
\ac_j&=\dfrac{\qc_j + \iota \pc_j }{\sqrt{2}}, ~~\ac_j^* =\dfrac{\qc_j - \iota \pc_j }{\sqrt{2}},
\end{align}   
where $\iota$ is the complex-unit, with equivalent commutation relations: $[\qc_j,\pc_k] =\iota \delta_{jk},[\qc_j,\qc_k] =[\pc_j,\pc_k] =0$.
We introduce the vector
\begin{equation}
\boldsymbol{\zetac}^{(r)}=\begin{pmatrix}
\qc_1 &
\qc_2 &
\cdots &
\qc_n &
\pc_1&
\pc_2 &
\cdots &
\pc_n
\end{pmatrix}^\top,
\end{equation}
where the superscript $r$ denotes the real quadrature form. 
We can then rewrite the commutation relations as
$[ \zetac_i^{(r)}, \zetac_j^{(r)}] =\iota (J_n^{(r)})_{ij}$, where the $2n \times 2n$ is
\begin{align}
\label{eq:J}
J_n^{(r)}=\begin{pmatrix}
0_{n \times n} & I_n\\
-I_n & 0_{n \times n}
    \end{pmatrix}.
\end{align}
Note that, in quantum mechanics, any  linear transformation of the vector 
$\boldsymbol{\zetac}^{(r)}$, that is $\boldsymbol{\zetac}^{(r)'}=S  \boldsymbol{\zetac}^{(r)}$, must preserve the commutation relations. This  means that: $S  J_n^{(r)} S^{\top} = J_n^{(r)}$, which implies that the matrix $S$ must be \textit{sympletic}. Finally, observe that 
 $(\boldsymbol{\zetac}^{(r)}  \boldsymbol{\zetac}^{(r)\top})_{ij}=\tfrac{1}{2}\{\zetac_i^{(r)},\zetac_j^{(r)}\}+\tfrac{\iota}{2}(J_n^{(r)})_{ij}$, where $\{\zetac_i^{(r)},\zetac_j^{(r)}\}=\zetac_i^{(r)}\zetac_j^{(r)}+\zetac_j^{(r)}\zetac_i^{(r)}$ is the anti-commutator. {
  We recall that the  expectation value of any observable $\check{o}$  with respect to a density operator $\check{\rho}$ is defined as $\mathbb{E}(\check{o})= Tr(\check{\rho}\check{o})$.
We can then define the expectation of the vector of observables   $\boldsymbol{\zetac}^{(r)}$ as $\mathbb{E}(\boldsymbol{\zetac}^{(r)})=Tr(\check{\rho}\boldsymbol{\zetac}^{(r)} )$, applied element-wise to  the elements of $\boldsymbol{\zetac}^{(r)}$. Similarly, the covariance matrix is  defined element-wise as $
\mathbb{E}(\boldsymbol{\zetac}^{(r)}-\mathbb{E}(\boldsymbol{\zetac}^{(r)}))(\boldsymbol{\zetac}^{(r)}-\mathbb{E}(\boldsymbol{\zetac}^{(r)}))^\top=\tfrac{1}{2}Tr(\check{\rho}\{\boldsymbol{\zetac}^{(r)}-\mathbb{E}(\boldsymbol{\zetac}^{(r)}) ,\boldsymbol{\zetac}^{(r)}-\mathbb{E}(\boldsymbol{\zetac}^{(r)}) \})+\tfrac{\iota}{2}Tr(\check{\rho} J_n^{(r)})=Q^{(r)}+\tfrac{\iota}{2}J_n^{(r)}$. Therefore, this implies that
a real symmetric positive-definite $2n \times 2n$ matrix $Q^{(r)}$ is a bona fide (quantum) covariance matrix if
and only if the Hermitian matrix  $Q^{(r)}+\tfrac{\iota}{2}J_n^{(r)}$ is positive semidefinite \citep{simon1994quantum}. For instance, for $n=1$, we have that
{\scriptsize
$$
Q^{(r)}=
\begin{pmatrix}
\mathbb{E}(\qc_1-\mathbb{E}(\qc_1))^2& \tfrac{1}{2} \mathbb{E}(\{\qc_1-\mathbb{E}(\qc_1),\pc_1-\mathbb{E}(\pc_1)\}) \\
 \tfrac{1}{2}\mathbb{E}(\{\qc_1-\mathbb{E}(\qc_1),\pc_1-\mathbb{E}(\pc_1)\}) & \mathbb{E}(\pc_1-\mathbb{E}(\pc_1))^2\\
\end{pmatrix}.
$$}
Then the matrix $Q^{(r)}+\tfrac{\iota}{2}J_n^{(r)}$ is positive semidefinite if 
{
$$
\begin{aligned}
&det(Q^{(r)}+\tfrac{\iota}{2}J_n^{(r)})\\
&=
\mathbb{E}(\qc_1-\mathbb{E}(\qc_1))^2 \mathbb{E}(\pc_1-\mathbb{E}(\pc_1))^2\\
&- (\tfrac{1}{2}\mathbb{E}(\{\qc_1-\mathbb{E}(\qc_1),\pc_1-\mathbb{E}(\pc_1)\}))^2-\tfrac{1}{4}\geq 0,
\end{aligned}
$$}
which corresponds to the uncertainty principle. This is a first key difference between classical and quantum systems: the Hermitian matrix $Q^{(r)} + \tfrac{\iota}{2} J_n^{(r)}$ must be positive semidefinite; the positive semidefiniteness of $Q^{(r)}$ alone is not sufficient.}

 \subsubsection{Linear quantum dynamical systems}
 \label{sec:revLQS}
Now we will focus on linear quantum dynamical systems. 
 The Hamiltonian of a quantum linear system is at most quadratic in terms:
\[
H = \frac{1}{2} 
\boldsymbol{\zetac}^{(c)\dagger}
\begin{pmatrix}
M_1 & M_2 \\
M_2^* & M_1^*
\end{pmatrix}
\boldsymbol{\zetac}^{(c)} =\frac{1}{2} \boldsymbol{\zetac}^{(c)\dagger}
M
\boldsymbol{\zetac}^{(c)},
\]
where $\boldsymbol{\zetac}^{(c)}=\begin{pmatrix}
\ac_1 &
\ac_2 &
\cdots &
\ac_n &
\ac_1^*&
\ac_2^* &
\cdots &
\ac_n^*
\end{pmatrix}^\top$, 
which specifies the dynamics of the harmonic oscillators, with $M_1, M_2 \in \mathbb{C}^{n \times n}$ and $M = M^\dagger$.  
{ We assume that the harmonic oscillators are coupled to $s$ distinct external quantum boson fields.} The coupling between the harmonic oscillators and the input bosonic quantum fields is linear, taking the form
$\begin{pmatrix} N_1 & N_2 \end{pmatrix} \boldsymbol{\zetac}^{(c)}, \quad N_1, N_2 \in \mathbb{C}^{s \times n}$. Prior to interacting with the quantum system, the input fields may be processed by static devices, such as beamsplitters or phase shifters. This is modeled by a unitary scattering matrix \(S \in \mathbb{C}^{s \times s}\). In the Heisenberg picture, the joint dynamics of the harmonic oscillators and the quantum fields are described by a system of \textit{Quantum Stochastic Differential Equations} (QSDEs)  \citep{nurdin2017linear}. These can be equivalently represented as a \textit{Quantum State-Space} (QSS) model, expressed in the real-quadrature coordinates $\boldsymbol{\zetac}^{(r)} =(\qc_1,\ldots,\qc_n,\pc_1,\ldots,\pc_n)^{\top}$ as
 \citep{zhang2022linear}:
\begin{align}
\label{eq:realqsys10}
d \boldsymbol{\zetac}^{(r)}(t)&= A^{(r)} \boldsymbol{\zetac}^{(r)}(t)dt+L^{(r)}(\boldsymbol{u}(t)dt+d\boldsymbol{\wc}^{(r)}(t))\\
\label{eq:realqsys20}
d \boldsymbol{\check y}^{(r)}(t)&= C^{(r)} \boldsymbol{\zetac}^{(r)}(t)dt+D^{(r)} (\boldsymbol{u}(t)dt+d\boldsymbol{\wc}^{(r)}(t))
\end{align}
where
\begin{align}
\label{eq:Dr}
D^{(r)}&= \begin{pmatrix}
    \Re(S) & -\Im(S)\\
    \Im(S) & \Re(S)
\end{pmatrix},\\
\label{eq:Cr}
C^{(r)}&= \begin{pmatrix}
    \Re(N_1+N_2) & \Im(-N_1+N_2)\\
    \Im(N_1+N_2) & \Re(N_1-N_2)
\end{pmatrix},\\
\label{eq:Lr}
L^{(r)}&=  -C^{(r)\sharp}D^{(r)},\\
\nonumber
A^{(r)}&=\begin{pmatrix}
    \Im(M_1+M_2) & \Re(M_1-M_2)\\
       -\Re(M_1+M_2) & -\Im(-M_1+M_2)\\ 
\end{pmatrix} \\
\label{eq:Ar}
&-\frac{1}{2} C^{(r)\sharp}C^{(r)},
\end{align}
where $\Re$ denotes the real-part, $\Im$ is the imaginary part, the $\sharp$-adjoint of a matrix $X\in \mathbb{R}^{2s \times 2n}$ is defined as $X^\sharp=- J^{(r)}_nX^\dagger J^{(r)}_s$ with $J_n^{(r)}$ defined in \eqref{eq:J}, and { $\boldsymbol{u}(t)dt+d\boldsymbol{\wc}^{(r)}(t)$ is the vector of  the $s$ external input fields with} $\boldsymbol{\wc}^{(r)}(t)$ being a quantum Wiener process satisfying the It\={o} relationship:
$d\boldsymbol{\wc}^{(r)}(t)d\boldsymbol{\wc}^{(r)\top}(t)=Q^{(r)}+\tfrac{\iota}{2} J^{(r)}_{s}$.
It can be proved that the  system matrices $A^{(r)},L^{(r)},C^{(r)}$
satisfy the following physical realisability conditions:
\begin{align}
\label{eq:physicalreal1}
A^{(r)}J_n^{(r)}+J_n^{(r)}A^{(r)\top}+L^{(r)}J_s^{(r)}L^{(r)\top}&=0,\\
\label{eq:physicalreal2}
J_n^{(r)}L^{(r)}+C^{(r)\top}J_s^{(r)} D^{(r)}&=0,\\
\label{eq:physicalreal3}
D^{(r)}J_s^{(r)}D^{(r)\top}-J_s^{(r)}&=0.
\end{align}
Contrarily, we can also prove that for any system matrices $A^{(r)},L^{(r)},C^{(r)},D^{(r)}$ satisfying the physical realisability conditions, originates into a valid QSDE, that is the QSS can be physically realised by a genuine quantum-mechanical system. This also ensures that the fundamental commutation relations are preserved during temporal evolution. 
\begin{remark}
This is a key difference between classical linear dynamical systems and quantum ones. In the latter, the system matrices $A^{(r)},L^{(r)},C^{(r)},D^{(r)}$  cannot be chosen arbitrarily, which consequently imposes a limit on universality and must be taken into account when designing QRC.
\end{remark}

\begin{remark}
In \eqref{eq:realqsys10}--\eqref{eq:realqsys20}, we could alternatively include the input { $\boldsymbol{u}(t)$} as a classical signal, for example a photocurrent. However, we do not do so here because we want to consider the more general case in which the input is quantum, for instance, the amplitude of an input laser. This ensures that the quantum linear network is fully quantum, while only the readout layer will be  classical in our qWiener architecture.
\end{remark}

Finally, we recall that:

\begin{definition}
{\citep{nurdin2017linear,zhang2022linear}}
\label{def:comppas}
A   QSS is said to be completely passive if $M_2=0$, $N_2=0$,   
 $\Re(M_1)= \Re(M_1)^\top$, $\Im(M_1)=-\Im(M_1)^\top$.
\end{definition}

In quantum reservoir computing, understanding the realisability of quantum linear systems is crucial. A \textit{completely passive} system has a quadratic Hamiltonian and linear couplings, generating no energy and requiring no active elements and can be implemented using only passive optical components.

\subsubsection{Quantum Kalman filter}
The temporal evolution of an open quantum system under
continuous (weak) measurement is described by the quantum Kalman filter. By denoting the density operator of the measured system with $\check{\rho}_m(t)$ and redefining $\pi_t(\boldsymbol{\zetac}^{(r)})=Tr(\check{\rho}_m(t)\boldsymbol{\zetac}^{(r)})$. Note that, before measurement, $\check{\rho}(0)=\check{\rho}_m(0)$. 
We consider a homodyne detection of the first quadrature.
We  can rewrite the stochastic evolution of the observed system \eqref{eq:realqsys10}--\eqref{eq:realqsys20}   as  the quantum Kalman filtering equation \citep{zhang2022linear}:

\begin{align}
 \label{eq:KB1}
d\pi_t(\boldsymbol{\zetac}^{(r)})&= A^{(r)} \pi_t(\boldsymbol{\zetac}^{(r)} )dt+ L^{(r)}  {\bf u}(t) dt + G_td\boldsymbol{\nu}_q(t),\\
\nonumber
d{\bf y}_{qm}(t) &= C_q^{(r)} \pi_t(\boldsymbol{\zetac}^{(r)}) dt + D_q^{(r)} {\bf u}(t) dt \\
\label{eq:KB2}
&+\left(D_q^{(r)}D_q^{(r)\top}\right)d\boldsymbol{\nu}_q(t),
\end{align}
where $C_q^{(r)}=\begin{pmatrix} I_s & 0_s\end{pmatrix}C^{(r)}$, $D_q^{(r)}=D^{(r)}\begin{pmatrix} I_s & 0_s\end{pmatrix}^\top$, $G_t=V_t C_q^{(r)\top} + L^{(r)}Q^{(r)}D_q^{(r)\top}$ and the covariance matrix $V_t$ satisfies the Riccati differential equation:
{\small
\begin{align}
\nonumber
\dot{V}_t &=A^{(r)}V_t+ V_tA^{(r)\top} + L^{(r)}Q^{(r)}L^{(r)\top}\\
\label{eq:KB3}
&-G_t \left(D_q^{(r)}D_q^{(r)\top}\right)^{-1}G_t^{\top},
\end{align} }
 and 
\begin{align*}
\boldsymbol{\nu}_q(t)= &\left(D_q^{(r)}D_q^{(r)\top}\right)^{-1}\Big(\\
&{\bf y}_{qm}(t) - \int_{0}^t (C_q^{(r)} \pi_\tau(\boldsymbol{\zetac}^{(r)})  + D_q^{(r)} {\bf u}(\tau)  ) d\tau \Big)
\end{align*}
is the so-called  {\em innovation process} of the quantum Kalman filter. Note that $\boldsymbol{\nu}_q(t)$ is a classical standard Wiener process, $\mathbb{E}[\boldsymbol{\nu}_q(t)\boldsymbol{\nu}_q(t')^{\top}]=\min\{t,t'\}I_{2s}$. 

We recall the following definitions. For the system \eqref{eq:sys1}-\eqref{eq:sys2}, the matrix pair $(A ,B)$ is said to be \textit{stabilizable} if there exists some matrix $R$ such that $A-BR$ has all eigenvalues with strictly negative real-parts.
The matrix pair $(A ,C)$ is said to be \textit{detectable} if there exists some matrix $R$ such that $A-RC$ has all eigenvalues with strictly negative real-parts.

Given \eqref{eq:KB1}--\eqref{eq:KB2}, if the pair $(A^{(r)},C_q^{(r)})$ is detectable and $(A^{(r)},L^{(r)}Q^{1/2(r)})$ is stabilizable,  the quantum Kalman filter  converges to a unique steady-state \citep[Sec.\ 8.5]{simon2006optimal}, whose covariance $V_{\infty}$ is given by the solution of 
{\small
\begin{align}
\nonumber
0&=A^{(r)}V_{\infty}+ V_{\infty}A^{(r)\top} + L^{(r)}\tilde{Q}^{(r)}L^{(r)\top}\\
\label{eq:KB3ss}
&-G_{\infty}\left(D_q^{(r)}D_q^{(r)\top}\right)^{-1}G_{\infty}^{\top}.
\end{align} }

\begin{remark}~
\begin{itemize}
 \item
 It is important to notice that, in \eqref{eq:KB1}-\eqref{eq:KB2}, the impulse response between the input ${\bf u}(t)$ and the output ${\bf y}_{qm}$ is simply the restriction to the first quadrature of the impulse response of the system \eqref{eq:realqsys10}--\eqref{eq:realqsys20}. Therefore, by filtering out the contribution of the noise $\boldsymbol{\nu}_q$ into ${\bf y}_{qm}$ (induced by the continuous weak measurement), we reduce to the deterministic part of the system \eqref{eq:realqsys10}--\eqref{eq:realqsys20}.
 \item
 It is important to notice the particular form of the Kalman gain $G_{\infty}=V_{\infty} C_q^{(r)\top} + L^{(r)}Q^{(r)} D_q^{(r)\top}$, where the second term is due to the common noise  in the two equations in \eqref{eq:realqsys10}--\eqref{eq:realqsys20} (see for instance \citep[Sec.\ 8.4]{simon2006optimal}).
This means that the Kalman gain can be zero even when the term $V_{\infty} C_q^{(r)\top}$ is not zero.
\end{itemize}
\end{remark}
We illustrate the second remark above by proving an important  case in which $G_{\infty}=0$.

\begin{theorem}
\label{th:1}
Assume $Q^{(r)}=\tfrac{1}{2}I_{2s}$,  which is the covariance of the vacuum, $D^{(r)}=I_{2s}$, $n=s$, the matrix $A^{(r)}$ has eigenvalues with strictly negative real-parts, the pair $(A^{(r)},C_q^{(r)})$ is detectable and $(A^{(r)},L^{(r)}Q^{1/2(r)})$ is stabilizable and assume that the QSS \eqref{eq:realqsys10}--\eqref{eq:realqsys20} is completely passive. Then  the asymptotic Kalman gain $G_{\infty}=0$.
\end{theorem}

This case is important because it implies that, for large $t$, there is no noise in the dynamic equation~\eqref{eq:KB1}, since $G_t \approx  G_{\infty} = 0$. In other words, the weak measurement does not disturb the dynamics of the QSS. We will exploit this result in the qWiener architecture in the next section.

\section{Main results}
We derive our QRC framework based on a quantum Wiener architecture.
This consists of a layer of $d$ QSS models, each continuously monitored via homodyne detection. The $d$ measured outputs are then fed into a classical GP model that constitutes the final memoryless nonlinear layer of the Wiener architecture.

First, we consider a network of $d$ one-dimensional quantum harmonic oscillators connected together via \textit{concatenation product} \citep{gough2009series}.   This will allow us to understand the difference between classical and quantum linear dynamcial systems in terms of impulse response.

Second, we will show that we can build a quantum linear dynamical system with any pre-defined impulse response, proving their universality.

\subsection{One-dimensional quantum harmonic oscillator}
\label{sec:harm1}
Consider the generic QSS model in \eqref{eq:realqsys10}--\eqref{eq:realqsys20}. Assume  that $n=1$, i.e., a single degree of freedom, so that the state reduces to $\boldsymbol{\zetac}^{(r)} =(\qc_1,\pc_1)^{\top}$. We can then prove the following result.

\begin{lemma}
\label{lem:h_harmonic}
The state space matrices of a $n=1$ degree-of-freedom linear    passive QSS  with  $s$ inputs are
 \begin{align}
A^{(r)}
        &=\begin{pmatrix}
         -\lambda &\Re(M_1)\\-\Re(M_1) &-\lambda
        \end{pmatrix},\\
        C^{(r)}
    &=\begin{pmatrix}
         \Re(N_1) & - \Im(N_1)\\\Im(N_1) &  \Re(N_1)
        \end{pmatrix},\\
L^{(r)}
        &=-\begin{pmatrix}
         \Re(N_1)^\top &  \Im(N_1)^\top\\-\Im(N_1)^\top &  \Re(N_1)^\top
        \end{pmatrix}\begin{pmatrix}
    \Re(S) & -\Im(S)\\
    \Im(S) & \Re(S)
\end{pmatrix},\\
D^{(r)}
        &=\begin{pmatrix}
    \Re(S) & -\Im(S)\\
    \Im(S) & \Re(S)
\end{pmatrix},
 \end{align}
 where $\lambda=\frac{1}{2}(\Re(N_1)^\top \Re(N_1)+\Im(N_1)^\top \Im(N_1) )$  and the impulse response is
 \begin{align}
{\bf h}(t)=\Big(& e^{-\lambda t}\begin{pmatrix}
         -\cos(\omega t) & -\sin(\omega t)\\ \sin(\omega t)&  -\cos(\omega t)
        \end{pmatrix} \otimes \Lambda\\
        &+I_{2s}\delta(t) \Big)\begin{pmatrix}
    \Re(S) & -\Im(S)\\
    \Im(S) & \Re(S)
\end{pmatrix},
 \end{align}
 with  $\Lambda=\Re(N_1) \Re(N_1)^\top+ \Im(N_1)\Im(N_1)^\top$, $\omega=\Re(M_1)$ and $\Re(S)^2+\Im(S)^2=1$.
\end{lemma}
The impulse response is a matrix ${\bf h}(t) \in \mathbb{R}^{2s \times 2s}$ as the dimension of the input and output is $2s$.
We  then specialise the result to the case where the dimension of the input  quantum field is $s=1$.

\begin{corollary}
\label{co:h_harmonic}
Assume that the dimension $s=1$, $\Re(M_1)=-\omega$,   $\Re(N_1)=\alpha$, $\Im(N_1)=\beta$ and $S=s_1+\iota s_2$, with $s_1^2+s_2^2=1$, where $\alpha,\beta,\omega,s_1,s_2 \in \mathbb{R}$, then
\begin{align}
\label{eq:Acpb}
A^{(r)} &=
\begin{pmatrix}
-\frac{\alpha^2+\beta^2}{2} & -\omega\\
\omega & -\frac{\alpha^2+\beta^2}{2} & \\
\end{pmatrix},\\
\label{eq:Lcpb}
L^{(r)} &=
\begin{pmatrix}
-\alpha & \beta  \vspace{1mm}\\
 -\beta  & -\alpha\\
\end{pmatrix}\begin{pmatrix}
s_1 & -s_2  \vspace{1mm}\\
 s_2 & s_1\\
\end{pmatrix},\\
\label{eq:Ccpb}
C^{(r)} &= \begin{pmatrix}
\alpha & -\beta  \vspace{1mm}\\
 \beta  & \alpha\\
\end{pmatrix},\\
\label{eq:Dcpb}
D^{(r)} &= \begin{pmatrix}
s_1 & -s_2  \vspace{1mm}\\
 s_2 & s_1\\
\end{pmatrix},
\end{align}
and
\begin{align}
\nonumber
{\bf h}(t) =  \Big(&\left(\alpha ^2+\beta ^2\right) e^{-\frac{1}{2} t
   \left(\alpha ^2+\beta ^2\right)} \left(
\begin{array}{cc}
-\cos ( \omega t ) &
\sin ( \omega t ) \\
-\sin ( \omega t ) &
 -\cos ( \omega t )\end{array}
\right) \\
\label{eq:hcpb}
&+I_2\delta(t) \Big)\begin{pmatrix}
s_1 & -s_2  \vspace{1mm}\\
 s_2 & s_1\\
\end{pmatrix}.
\end{align}
\end{corollary}
This is an impulse response for a system with a two-dimensional input $\mathbf{u}(t)$ and output $\mathbf{y}(t)$.
Note that, without loss of generality, we can choose $\beta=0$.

To realise our quantum Wiener architecture, composed of $n_c$ one-dimensional harmonic oscillators, we construct the system by concatenating the harmonic oscillators with state matrices defined in \eqref{eq:Acpb}--\eqref{eq:Dcpb}, using the concatenation product \citep{gough2009series}. { The concatenation product simply assembles the different components together, without making any connections between them. Consider $n_c$ harmonic oscillators with parameters
 $\{\alpha_i,\omega_i,s_{i1},s_{i2}\}_{i=1}^{n_c}$, their concatenation product is obtained by setting $M_1=\text{diag}(\omega_1,\dots,\omega_{n_c})$, $N_1=\text{diag}(\alpha_1,\dots,\alpha_{n_c})$ and\\ $S=\text{diag}\left(s_{11}+\iota s_{12},\dots,s_{n_c1}+\iota s_{n_c2}\right)$.}

\begin{proposition}
\label{prop:concatenation}
 Consider the QSS obtained by the  concatenation product of $n_c$ harmonic oscillators with parameters
 $\{\alpha_i,\omega_i,s_{i1},s_{i2}\}_{i=1}^{n_c}$. Assume that we measure the first quadratures $\qc_1, \ldots, \qc_{n_c}$ of the state, and then sum the measured outputs. Denote this sum by $\tilde{y}_{qm}(t)$.
 Moreover, assume that the $2n_c$-dimensional input ${\bf u}(t)=\begin{pmatrix} \mathbf{1}_{n_c} & 0_{n_c}\end{pmatrix}^\top u(t)$ for some scalar input $u(t)$, then the SISO input-output impulse response (between $u(t)$ and $\tilde{y}_{qm}(t)$)  is
\begin{equation}
\label{eq:hquantum}
\tilde{h}(t)=\sum_{i=1}^{n_c} \alpha_i^2 e^{-\frac{\alpha_i^2t}{2} } \left(-
s_{i1}\cos ( \omega_i t ) +s_{i2} \sin ( \omega_i t )\right) +s_{i1}\delta(t).
\end{equation}
\end{proposition}

\begin{remark}
By comparing \eqref{eq:hquantum} with \eqref{eq:h}, we can notice two major differences in the quantum case: (i) the multiplying scalar $\alpha_i^2$; (ii) the variables $s_{i1},s_{i2}$ are constrained to have unit-norm. These differences follows from the physical realisability constraints \eqref{eq:physicalreal1}--\eqref{eq:physicalreal3}. These constraints limit the expressivity of quantum linear networks for reservoir computing, as their parameters cannot be freely chosen.
\end{remark}
{ The quantum system in Proposition \ref{prop:concatenation} is completely passive and, therefore, it can be realised using only passive optical components like phase shifters, beam splitters
and mirrors.}  Since it is completely passive, from Theorem \ref{th:1} it follows that the dynamics of the measured system is deterministic and the noise only affects the measurement equation \eqref{eq:KB2}, which implies that:
\begin{equation}
\label{eq:ynoise}
\tilde{y}_{qm}(t)=\int_{0}^{t} \tilde{h}(t-\tau) u(\tau)d\tau +\mathbf{1}_{n_c}^\top D^{(r)}D^{(r)\top} \mathbf{1}_{n_c}\boldsymbol{\nu}(t),
\end{equation}
where $\tilde{h}(t)$ is defined \eqref{eq:hquantum}.
This is important because it implies that the weak measurement does not disturb the dynamics of the QSS. Filtering out the static noise in \eqref{eq:ynoise} is quite straightforward.

\subsubsection{Limit Kernel}
\label{sec:limitkernelquantum}
The  results in Section \ref{sec:kernel_methods} exploit the Gaussianity and independence of the coefficients $\{\beta_i,\gamma_i\}_{i=1}^{n_c}$ (given $\{\alpha_i,\omega_i\}_{i=1}^{n_c}$) to show that $\tilde{h}(t)$ is GP distributed.

However, note that, due to the unit-norm constraint $s_{i1}^2+s_{i2}^2=1$, the coefficients $\{s_{i1},s_{i2}\}_{i=1}^{n_c}$ in \eqref{eq:hquantum} cannot be  Gaussian distributed. For instance, if we assume that $s_{i1},s_{i2}$ are independently Gaussian distributed, then conditioned on $s_{i1}^2+s_{i2}^2=1$, their distribution is Von-Mises $(s_{i1},s_{i2})^\top \sim \text{VM}\left(\tfrac{\mathbf{1}_2}{\sqrt{2}},(1/\kappa) I_2\right)$ (see Appendix \ref{app:VM} for details about the Von-Mises distribution).

We will now consider another way to derive  results similar to those in in Section \ref{sec:kernel_methods}. By exploiting the central limit theorem (similarly to what was done in \citep{williams1998computation,neal2012bayesian} for neural networks), we will show that the corresponding stochastic process converges to a GP in the
limit as $n_c \rightarrow \infty$.
\begin{theorem}
\label{th:clt}
For $i=1,\dots,n_c$, consider the impulse response \eqref{eq:hquantum} and assume that $\alpha^2_i,\omega_{i} \sim \phi_{\boldsymbol{\theta}}(\alpha^2,\omega)$ and $(s_{i1},s_{i2})^\top \sim \text{VM}\left(\tfrac{\mathbf{1}_2}{\sqrt{2}},(1/\kappa) I_2\right)$ independently for $i=1,\dots,n_c$.
Then, as $n_c$ tends to infinity, we have that
\begin{equation}
\sqrt{n_c}\left(\tilde{h}(t)-\mu_h(t)\right) \stackrel{d}{\longrightarrow} \mathcal{GP}\left(0,K_{\boldsymbol{\theta}}(t,t')\right),
\end{equation}
where
\begin{align}
\label{eq:h_alpha_mu}
&\mu_h(t)=  \frac{r_1}{\sqrt{2}}\int_{-\infty}^{\infty}\int_{0}^{\infty} \alpha^2 e^{-\frac{\alpha^2t}{2} } \cos ( \omega t ) \phi_{\boldsymbol{\theta}}(\alpha^2,\omega) d\alpha^2 d\omega,\\
\nonumber
&K_{\boldsymbol{\theta}}(t,t')\\
\nonumber
&= \frac{1 }{2} \int_{-\infty}^{\infty}\int_{0}^{\infty} \alpha^4 e^{-\frac{\alpha^2(t+t')}{2} }\cos ( \omega (t-t') )  \phi_{\boldsymbol{\theta}}(\alpha^2,\omega) d\alpha^2 d\omega\\
\label{eq:h_alpha_K}
&-\mu_h(t)\mu_h(t') ,
\end{align}
and $r_1=\mathcal{I}_1({\displaystyle \kappa })/\mathcal{I}_0({\displaystyle \kappa })$ with $\mathcal{I}_n({\displaystyle \kappa })$ being the modified Bessel function of the first kind of order $n$. Assuming that $\phi_{\boldsymbol{\theta}}(\alpha^2,\omega)=\frac{\tfrac{\alpha^2}{2}}{\pi(\omega^2+\tfrac{\alpha^4}{4})}\text{Unif}(\alpha^2; \alpha_m,\alpha_M)$, then we have
\begin{align}
\label{eq:h_alpha_muint}
&\mu_h(t)= - \frac{r_1}{\sqrt{2}}\frac{(1+ a_m t) e^{- a_m t
   }-e^{- a_M t }
   (1+ a_M t)}{t ^2 (a_M-a_m)},\\ &\\
\nonumber
&K_{\boldsymbol{\theta}}(t,t')=\\
\nonumber
& \frac{1 }{2}  \frac{ (1+
   a_m \tau +\frac{  a_m^2 \tau^2}{2}
   ) e^{- a_m
   \tau}-e^{- a_M
   \tau} (1+ a_M \tau +\frac{a_M^2 \tau^2}{2} )}{2 \tau^3 (a_M-a_m)}\\
  \label{eq:h_alpha_Kint}
   &-\mu_h(t)\mu_h(t') ,
\end{align}
where $\tau=\max(t,t')$.
\end{theorem}

\begin{remark}
Theorem \ref{th:clt} extends Proposition \ref{prop:1} to the quantum case.
It is important to emphasize the difference between the kernel in \eqref{eq:h_alpha_Kint} for the quantum case and the standard integrated tuned-correlated kernel in \eqref{eq:TC}. The former includes the fractional polynomial terms 
$(1 + a_m \tau + \tfrac{a_m^2 \tau^2}{2})/\tau^3,
$ (similarly for $a_M$), while the latter only contains the term $1/\tau$. Again this difference follows from the physical realisability constraints \eqref{eq:physicalreal1}--\eqref{eq:physicalreal3}, which introduce the term $\alpha^2$ in \eqref{eq:hquantum} and, consequently, $\alpha^4$ in \eqref{eq:h_alpha_K}.
For large values of $\tau$, the two kernels converge. For small $\tau$, the  polynomial term $1 + a_m \tau + \tfrac{a_m^2 \tau^2}{2}$ behaves similarly to the polynomial term that appears in the impulse response of a system in Jordan form (specifically, in the case of a LTI system with a complex-conjugate pair of eigenvalues of multiplicity three). Therefore, paradoxically, the physical realisability constraints \eqref{eq:physicalreal1}--\eqref{eq:physicalreal3} can increase the expressivity of the RC when the distribution of randomised parameters is suitably chosen.
\end{remark}
The difference between the kernel  \eqref{eq:TC}  and \eqref{eq:h_alpha_Kint} is the main difference between the classical and the quantum Wiener architectures. 
Indeed, as we perform a weak measurement on the first quadrature component of each one of the one-dimensional harmonic oscillators, we obtain the signal:
\begin{equation}
\tilde{y}_{qm}(t) = \int_{0}^{t} \tilde{h}(t-\tau)u(\tau)d\tau  
+ \mathbf{1}_{n_c}^\top D^{(r)}D^{(r)\top}\mathbf{1}_{n_c}\,\boldsymbol{\nu}(t).
\end{equation}
We can filter out the noise to recover the deterministic component 
$\int_{0}^{t} \tilde{h}(t-\tau)u(\tau)d\tau$. 
After the weak measurement, the qWiener architecture becomes classical and, therefore, we can apply the same results from Propositions~\ref{prop:linBlocks}--\ref{prop:deepKernel} 
to derive the architecture (shown in Figure~\ref{fig:2}) and the corresponding deep kernel for the qWiener model. 
The deep kernel has the same structure as in Proposition~\ref{prop:deepKernel}, but 
$K_{{\boldsymbol{\theta}},n_c}(t-\tau_1,t'-\tau_2)$ corresponds to the kernel in~\eqref{eq:h_alpha_Kint} 
instead of that in~\eqref{eq:TC}.

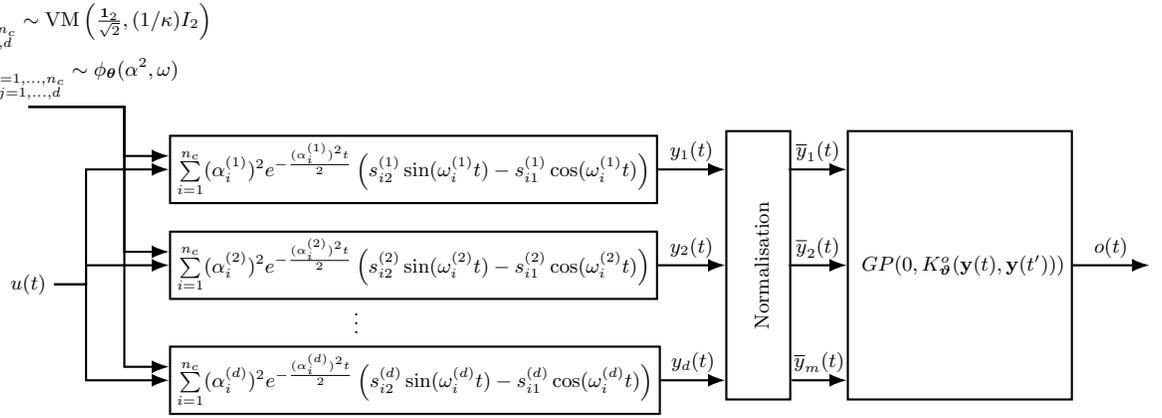
\begin{figure*}
\begin{tikzpicture}[node distance=1.2cm and 1.8cm, thick, >=Latex,scale=0.85, transform shape]
\begin{scope}[xshift=-10cm]
\node (input) at (0,0) {$u(t)$};
\node[above=2.5cm of input, align=center] (freqinput)
    {$\left\{ s^{(j)}_{i1},s^{(j)}_{12} \right\}_{\substack{i=1,\dots,n_c \\ j=1,\dots,d}} \sim \text{VM}\left(\tfrac{\mathbf{1}_2}{\sqrt{2}},(1/\kappa) I_2\right)$ \\
    $\left\{ \alpha_{i}^{2(j)}, \omega_{i}^{(j)}\right\}_{\substack{i=1,\dots,n_c \\ j=1,\dots,d}}
 \sim \phi_{\boldsymbol{\theta}}(\alpha^2,\omega)$};
\node[draw, rectangle, minimum width=4.6cm, minimum height=0.9cm, right=of input, yshift=1.8cm, align=center] (g1)
    {$\sum\limits_{i=1}^{n_c} (\alpha_i^{(1)})^2e^{-\frac{(\alpha_i^{(1)})^2t}{2} } \left(  s^{(1)}_{i2} \sin(\omega^{(1)}_i t) -  s^{(1)}_{i1} \cos(\omega^{(1)}_i t) \right)$};
\node[draw, rectangle, minimum width=4.6cm, minimum height=0.9cm, right=of input, yshift=0.3cm, align=center] (g2)
    {$\sum\limits_{i=1}^{n_c} (\alpha_i^{(2)})^2e^{-\frac{(\alpha_i^{(2)})^2t}{2} } \left(  s^{(2)}_{i2} \sin(\omega^{(2)}_i t) -  s^{(2)}_{i1} \cos(\omega^{(2)}_i t) \right)$};
\node[draw=none, right=of input, yshift=-0.5cm] (dots) {~~\hspace{2cm}~~~~~$\vdots$};
\node[draw, rectangle, minimum width=4.6cm, minimum height=0.9cm, right=of input, yshift=-1.5cm, align=center] (gm)
    {$\sum\limits_{i=1}^{n_c} (\alpha_i^{(d)})^2e^{-\frac{(\alpha_i^{(d)})^2t}{2} } \left(  s^{(d)}_{i2} \sin(\omega^{(d)}_i t) -  s^{(d)}_{i1} \cos(\omega^{(d)}_i t) \right)$};
\draw[->] (input.east) -- ++(0.5,0) |- (g1.west);
\draw[->] (input.east) -- ++(0.5,0) |- (g2.west);
\draw[->] (input.east) -- ++(0.5,0) |- (gm.west);
\draw[->] (freqinput.south) -- ++(1.5,0) |- ([yshift=6pt]g1.west);
\draw[->] (freqinput.south) -- ++(1.5,0) |- ([yshift=6pt]g2.west);
\draw[->] (freqinput.south) -- ++(1.5,0) |- ([yshift=6pt]gm.west);
\node[draw, rectangle, minimum width=1.0cm, minimum height=4.2cm, right=10.5cm of input, yshift=0.3cm] (normalis) {{ \rotatebox{90}{Normalisation}}};
\node[draw, rectangle, minimum width=1.4cm, minimum height=4.2cm, right=12.4cm of input, yshift=0.3cm] (combiner) {{ $GP(0,K^o_{\boldsymbol{\vartheta}}(\mathbf{y}(t), \mathbf{y}(t')) )$}};
\draw[->] (g1.east) -- (normalis.west |- g1.east) node[midway, above] {$y_1(t)$};
\draw[->] (g2.east) -- (normalis.west |- g2.east) node[midway, above] {$y_2(t)$};
\draw[->] (gm.east) -- (normalis.west |- gm.east) node[midway, above] {$y_d(t)$};

  \coordinate (out1) at ($(g1.east) + (2.05,0)$);
  \coordinate (out2) at ($(g2.east) + (2.05,0)$);
  \coordinate (out3) at ($(gm.east) + (2.05,0)$);

\draw[->] (out1.east) -- (combiner.west |- out1.east) node[midway, above] {$\overline{y}_1(t)$};

\draw[->] (out2.east) -- (combiner.west |- out2.east) node[midway, above] {$\overline{y}_2(t)$};
\draw[->] (out3.east) -- (combiner.west |- out3.east) node[midway, above] {$\overline{y}_m(t)$};

\draw[->] (combiner.east) -- ++(1.2,0) node[midway, above] {$o(t)$};
\end{scope}
\end{tikzpicture}
\label{fig:2}
\caption{Quantum Wiener architecture obtained by concatenating $d$ quantum harmonic oscillators.}
\end{figure*}
In Figure \ref{fig:2}, the normalisation layer performs the operation $\overline{y}_i(t)=\tfrac{\sqrt{n_c}}{\sqrt{d}}(y_i(t)-\hat{\mu}_y(t))$, where the mean $\hat{\mu}_y(t)=\frac{1}{d}\sum_{i=1}^d y_i(t)$ is approximatively  equal to 
 \begin{equation}
 \label{eq:muy}
   \mu_y(t)  = \int_{0}^t \mu_h(t-\tau)u(\tau)d\tau.
 \end{equation}

\subsection{Universality of quantum Wiener architectures}
\label{sec:univ}
To prove that \textit{qWiener} architectures can realise general SISO LTI blocks, we need to show that a QSS can implement any impulse response of the form $h(t) = Ce^{At}B + D\delta(t)$,
for arbitrary matrices $A, B, C, D$. Note that, without loss of generality, we can assume $D=0$ since, in an RC architecture, we can always provide $u(t)$ as an input in the last layer.

It is important to observe that we cannot directly set $ A^{(r)} = A $, $L^{(r)}=B$, and $C^{(r)} = C$, because physical realisability constraints \eqref{eq:physicalreal1}--\eqref{eq:physicalreal2} do not allow arbitrary choices of the matrices $A^{(r)},L^{(r)},C^{(r)}$. However, these constraints only impose conditions on the matrix elements relative to the state variables  $\qc_i,\pc_i$, and not, for example, between $\qc_i$ and $\qc_j$.
Therefore, in the following, we will show that one way to realise any desired impulse response is to embed its realisation in the subsystem associated with the components of the state $\qc_1,\dots,\qc_n$ (or $\pc_1,\dots,\pc_n$).

\begin{lemma}
\label{lem:qABCD}
 Consider the quantum linear system:
\begin{align}
\label{eq:Aruniv}
A^{(r)}&= \begin{pmatrix}
         A & 0_{n \times n}\\
         0_{n \times n} & Z
        \end{pmatrix}\\
        \label{eq:Cruniv}
C^{(r)}&= \begin{pmatrix}
         C^{(r)}_1 & 0_{1 \times n}\\
         0_{1 \times n} & C^{(r)}_2\\
        \end{pmatrix}\\
\label{eq:Lruniv}
L^{(r)}&=J_{n}^{(r)} C^{(r)\top} J_{s}^{(r)} \\
\label{eq:Druniv}
D^{(r)}&=I_{2s}
 \end{align}
 If  $Z=-A^\top-C^{(r)\top}_1 C^{(r)}_2$, then the system is physically realisable, that is, it satisfies \eqref{eq:physicalreal1}--\eqref{eq:physicalreal3}, and the impulse response is
 \begin{align*}
&{\bf h}(t)=\begin{pmatrix}
         -C^{(r)}_1  e^{At} C^{(r)\top}_2 & 0\\
         0 & -C^{(r)}_2  e^{Zt} C^{(r)\top}_1\\
        \end{pmatrix}+I_{2s}\delta(t).
 \end{align*}
\end{lemma}
From the above lemma, it can be noticed that, by setting $C^{(r)}_1 =-C$ and $C^{(r)}_2=B^\top$, the restriction of the impulse response to the first component of the input and output leads to the impulse response $Ce^{At}B+\delta(t)$. Therefore, we have shown that any SISO transfer function can be realised. As noted earlier, in RC the feedthrough term $\delta(t)$ can be disregarded, since the input $u(t)$ can always be fed into or removed from the input of the readout layer.

Observe that, in RC, the matrix $A$ is assumed to have all eigenvalues with strictly negative real parts. However, since the matrix $Z = -A^\top + C^\top B^\top$ includes the term $ -A^\top$, the matrix $Z$ is not, in general, guaranteed to have eigenvalues with strictly negative real parts. This implies that the resulting QSS could be unstable in general. This is problematic, as the corresponding QRC could then require an unbounded amount of energy. In what follows, we show that, by suitably choosing the matrices $C^{(r)}_1,C^{(r)}_2$, we can ensure that this does not occur.

\begin{theorem}
\label{th:qABCD}
 Assume that $C^{(r)}_1=\begin{pmatrix}R^\top & -C^\top\end{pmatrix}^\top,C^{(r)}_2=\begin{pmatrix}B & 0_{n \times 1}\end{pmatrix}^\top$
 then  $Z=-A^\top-R^\top B^\top=(-A-BR)^\top$, where $R \in \mathbb{R}^{1 \times n}$. The  impulse response is:
 \begin{align*}
&{\bf h}(t)=\begin{pmatrix}
         -R e^{At} B &0 & 0& 0\\
         C e^{At} B &0 & 0& 0\\
         0 & 0 & -B^\top  e^{Zt} R^\top & B^\top  e^{Zt} C^\top \\
         0 & 0 &0 &0\end{pmatrix}+I_{4}\delta(t).
 \end{align*}
 Assuming that  the matrix $\begin{pmatrix}B,AB,A^2B,\dots,A^{n-1}B\end{pmatrix}$ has rank $n$, then we can always find $R$ such that
$Z$ has eigenvalues with strictly negative real-part, so that the QSS \eqref{eq:Aruniv}--\eqref{eq:Druniv} is stable. 
\end{theorem}
The restriction of the impulse response to the second component of the output and the first component of the input is $Ce^{At}B+\delta(t)$.
In summary, we have shown that a QSS can realise any transfer function. 
The rank assumption in Theorem~\ref{th:qABCD} is standard in control theory. The matrix $R$ can for instance be determined by solving a Linear-Quadratic Gaussian control (LQG) \citep{green2012linear} problem with arbitrary weighting matrices.

It is important to note the following: { 
\begin{corollary}
 \label{co:active}
 The QSS described by~\eqref{eq:Aruniv}--\eqref{eq:Druniv} is not completely passive.
\end{corollary}
This means that this system requires an external driving (e.g., an external pump beam) to implement.} Since the QSS is active, when a weak measurement is performed on it, the system evolves according to~\eqref{eq:KB1}--\eqref{eq:KB2}. 
In this case, the measurement back-action introduces process noise into its dynamics, which needs to be properly accounted for.

To deal with the noise in the measured system \eqref{eq:KB1}--\eqref{eq:KB2}, we can employ variational inference in conjunction with stochastic optimization for the readout layer. At each optimisation iteration, the input 
$u(t)$ is propagated through the quantum linear networks under weak measurements, generating noisy inputs for the readout layer. The objective function optimised in the variational approximation (known as the Evidence Lower Bound (ELBO)), is estimated from noisy samples, but this estimate remains unbiased with respect to the true ELBO. Consequently, a stochastic optimizer such as Adam \citep{kingma2015adam} will converge to the true ELBO (or to a local minimum of it due to the nonlinearity of the optimisation).

\section{Connection to previous work}
\label{sec:prev}
The work \citep{nokkala2021gaussian} proposes a network of interacting quantum harmonic oscillators acting as the reservoir for RC, with
spring-like interaction strengths $g_{ij}$. The Hamiltonian
of such a system is described in
terms of the Laplacian matrix $V=V^\top$ having elements $V_{ij}=
\delta_{ij}\sum_k g_{ik} - (1 - \delta_{ij})g_{ij}$, where $\delta_{ij}$ is the Kronecker delta. The resulting Hamiltonian is
$$
\frac{\boldsymbol{\qc}^\top (G_\omega+V) \boldsymbol{\qc} }{2}+\frac{\boldsymbol{\pc}^\top \boldsymbol{\pc}}{2},
$$
where $\boldsymbol{\pc}=(\pc_1,\dots,\pc_n)^\top$,  $\boldsymbol{\qc}=(\qc_1,\dots,\qc_n)^\top$ and $G_\omega=\text{diag}(\omega_1,\dots,\omega_n)$ with $\omega_i>0$.
This leads to the differential equations
\begin{align}
\label{eq:gau1}
 \frac{d \boldsymbol{\qc}}{d t} &=\boldsymbol{\pc},\\
 \label{eq:gau2}
 \frac{d \boldsymbol{\pc}}{d t}&=-(G_\omega+V) \boldsymbol{\qc}.
\end{align}
In \citep{nokkala2021gaussian}, the input is injected into the network by resetting at each timestep the state of one of the oscillators. Instead of modelling the input by ancilla resets as in \citep{nokkala2021gaussian}, we use an equivalent linear QSS representation where the input couples through an external field as discussed in \ref{sec:revLQS}.
Assume without loss of generality, we use the first harmonic oscillator $\qc_1,\pc_1$ to inject the input.
By defining the matrices $M_1,M_2,N_1$ as follows,
\begin{align}
 M_1&=\frac{1}{2}(G_\omega+I_n+V),\\
M_2&=\frac{1}{2}(G_\omega-I_n+V),\\
N_1&=\left(a_1,0,\dots,0\right),
\end{align}
and $N_2=0$, we obtain the QSS matrices:
\begin{align}
 A^{(r)}&= \left(
\begin{array}{cc}
 -\frac{a_1^2}{2}\mathbf{e}_1\mathbf{e}^\top_1 & I_n\\
-G_\omega-V &  -\frac{a_1^2}{2}\mathbf{e}_1\mathbf{e}^\top_1
\end{array}
\right)\\
C^{(r)}&= a_1\left(
\begin{array}{cccccc}
 \mathbf{e}_1 & \mathbf{0}_n  \\
\mathbf{0}_n &  \mathbf{e}_1\\
\end{array}
\right)^\top\\
 D^{(r)}&=\left(
\begin{array}{cccccc}
 s_1\mathbf{e}_1 &  -s_2\mathbf{e}_1  \\
s_2\mathbf{e}_1&   s_1\mathbf{e}_1 \\
\end{array}
\right)\\
 L^{(r)}&= C^{(r)\sharp} D^{(r)}
\end{align}
where $\mathbf{e}_1$ is the first element of the canonical basis of $\mathbb{R}^n$ and $s_1^2+s_2^2=1$. Note that, for $a_1=0$ (no input), the matrix $ A^{(r)}$ coincides with  the state matrix defined by \eqref{eq:gau1}--\eqref{eq:gau2}. { It is interesting to note that,  this proposed quantum linear system has a very particular structure for its dynamics and we do not know if this structure can represent any linear dynamics, which is necessary for universality of the corresponding qWiener architecture. Conversely, our approach achieves  universality as proven for the system in Section \ref{sec:univ}. Additionally, in the numerical experiments, \cite{nokkala2021gaussian} sets $\omega = 0.25$ and selects random interaction strengths $g$ uniformly distributed in the interval $[0.01, 0.19]$ and argues that this breaks symmetries and favours a spectral radius $\rho(A)<1$. \cite{nokkala2021gaussian} does not provide an analysis of the statistical properties of the resulting reservoir to justify the choice of the uniform distribution. In the present work, we go beyond such a heuristic for the case of the qWiener model, e.g. by explicitly expressing the limit kernel  in Section \ref{sec:limitkernelquantum} and, therefore, justifying the choice of the distribution $\phi_{\boldsymbol{\theta}}(\alpha^2,\omega)=(\tfrac{\alpha^2}{2}/(\pi(\omega^2+\tfrac{\alpha^4}{4}))\text{Unif}(\alpha^2; \alpha_m,\alpha_M)$.

Finally, the work \citep{nokkala2021gaussian} discusses additional nonlinear capabilities obtained by varying the input encoding of the coherent ancillary state, from amplitude to phase. This type of model -- comprising a nonlinear transformation of the input, a bank of SISO LTI systems, and a static nonlinearity in the readout layer -- is known as a Hammerstein–Wiener model. In the following numerical experiment, we assume that the amplitude of the ancillary state is equal to the input itself, that is, no nonlinear transformation is applied at this stage. Note that,  nonlinear transformations of the input can be incorporated both in the model of \citep{nokkala2021gaussian} and in our proposed models. However, for  the same choice of nonlinearities at the input and readout layers, only our architecture in Section \ref{sec:univ} is proven to be universal.}

\begin{figure*}
\includegraphics[width=13cm]{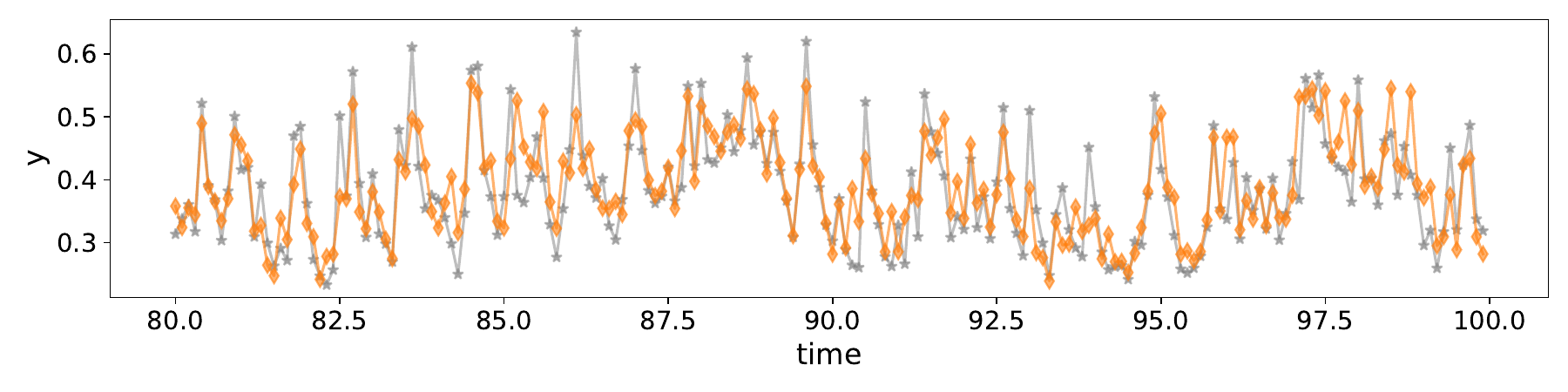}
\includegraphics[width=13cm]{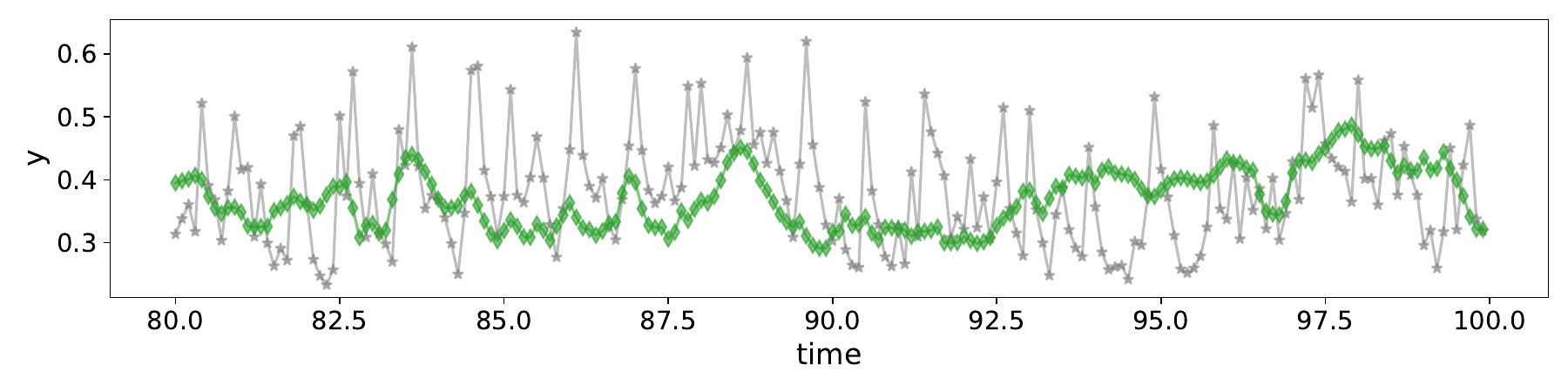}
\includegraphics[width=13cm]{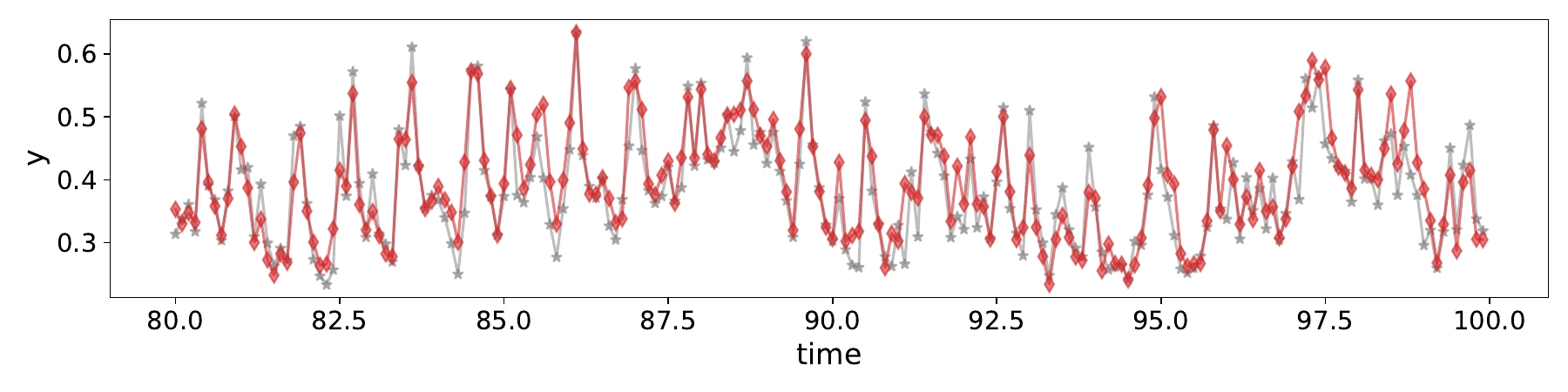}
\includegraphics[width=13cm]{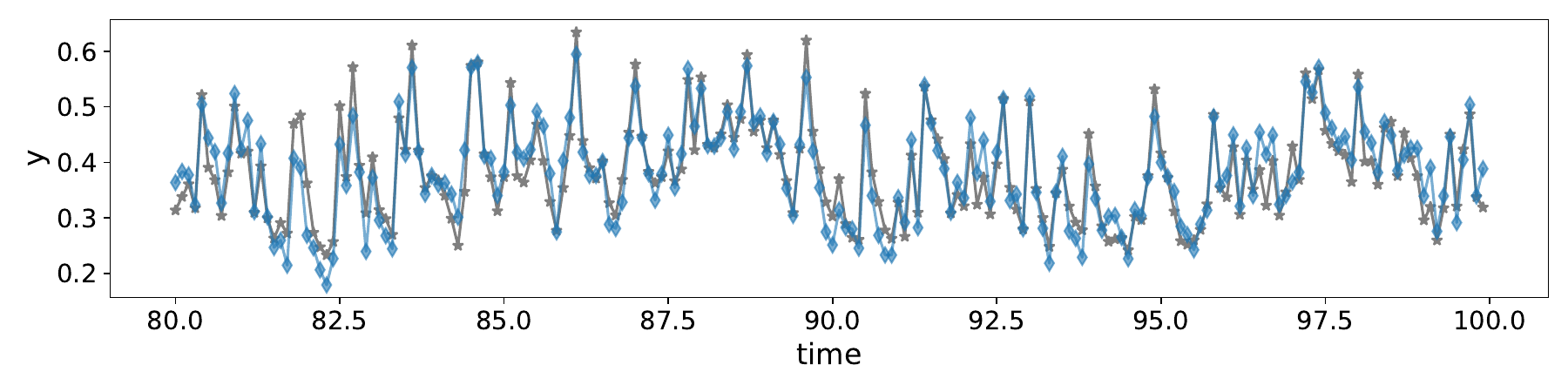}
\caption{Prediction on a test set for the NARMA10 benchmark: true-signal in gray;  HqW($24$) in orange; LqW($24$) in green; PadeqW($14$) in red; ESN($200$) in blue.}
\label{fig:traj}
\end{figure*}

\section{Numerical experiments}
In this section, we compare the Laplacian-based architecture proposed by \citep{nokkala2021gaussian} and illustrated in Section \ref{sec:prev}, which we refer to as \textit{LqW}, with the qWiener architecture obtained by concatenating harmonic oscillators, as presented in Section \ref{sec:harm1}, which we refer to as \textit{HqW}. Differently from \cite{nokkala2021gaussian}, for both LqW and HqW we use a squared-exponential kernel in the readout layer (instead of a polynomial nonlinearity). This choice eliminates the need to select the order of the polynomial and also provides better extrapolation properties ({ the results obtained using a polynomial nonlinearity in the readout layer are reported in Appendix \ref{app:addexp} for comparison).}
For LqW, in all the experiments, we select the parameters of the quantum linear network as $\omega=0.25$, $g_{ij} \sim \text{Unif}(0.01,0.19)$ and $a_1=5$. We tried different combination of the parameters but the performance did not change significantly. For HqW, in all experiments, we selected $a_m=0.01$ and $a_M=20$.
For comparison, as a baseline, we also include  a classical echo-state network (\textit{ESN}) \cite{jaeger2001echo} consisting of $200$ neurons (spectral radius $0.95$), that is a classical RC based on a recurrent neural
network. As benchmark problems, we consider the following three dynamical systems \citep{wringe2025reservoir}.

\textit{Parity check:} The parity check task is a nonlinear temporal benchmark used to assess a model's ability to retain and combine past input symbols in a nonlinear way.
Given a binary input sequence $u_t \in \{0, 1\}$, the target output is
$$
y_t = \mod\left( \sum_{l=0}^{\tau} u_{t-l},2\right),
$$
that is, the parity (even or odd) of the last $\tau$ input bits. As input, we use a pseudo-random binary sequences (PRBS) (with probability $p=0.5$). Note that, the parity task  was also considered in \citep{nokkala2021gaussian}. 

\textit{NARMA$10$:} The nonlinear autoregressive moving average  task is widely used to test nonlinear dynamical modeling and the balance between memory and nonlinearity. For NARMA10, the target output is:
$$
y_{t} = 0.3 y_{t-1} + 0.05 y_{t-1} \sum_{i=0}^{9} y_{t-1-i} + 1.5 u_{t-10} u_{t-1} + 0.1,
$$
where $u_t$ is  drawn from a uniform distribution in $[0, 0.5]$.
The task requires both long memory and nonlinear interactions between past inputs and outputs.

\textit{Delay:}
The delay task quantifies how much information about past inputs a model can retain.
Given an input $u_t$, the target output is a delayed version of it:
$$
y_t = u_{t-\tau}.
$$
As input, we use a PRRBS sequence (with probability $p=0.5$).

In all experiments, we selected a sampling step $\Delta=0.01$. Furthermore, we fixed $Q^{(r)}=\tfrac{1}{2}I_{2s}$,  which is the covariance of the vacuum, and, therefore, to maintain a high signal-to-noise ratio, we amplified the input signal by a factor of $500$.

Table \ref{tab:1} reports the median of the average Root Mean Square Error (RMSE) computed over 10 Monte Carlo repetitions. In each Monte Carlo repetition, a time series of length $1000$ is generated for the benchmarks described above, using different realisations of the random input sequence $u_t$. The first $800$ points are used for training, while the remaining $200$ points are used to compute the average RMSE. The table shows the median RMSE across the 10 repetitions for three models: (i) the classical RC ESN; (ii) the QRC LqW; (iii) our QRC HqW. In the top-row, the number between brackets after the name of the model is relative to the number of harmonic oscillators in the quantum linear network.

It can be seen that our method, HqW, outperforms LqW across all benchmarks when using the same number of harmonic oscillators. Moreover, the RMSE of HqW is typically half that of LqW, or even lower.
Contrary to what one would expect from a RC architecture, the performance of LqW does not generally improve as the number of harmonic oscillators increases; in fact, it often deteriorates, because the additional oscillators alter the system dynamics in a way that degrades performance. In contrast, a key advantage of our method is that we have proven that its performance improves as the number of oscillators increases, since the covariance converges more closely to the limiting kernel, thus increasing expressivity.
Our method also generally outperforms the ESN (with 200 neurons). This is not surprising, as a standard ESN cannot reliably solve tasks that require large values of $\tau$ and, in particular, struggles with the parity task \citep{barbosa2021symmetry}.

\begin{table*}[!tbp]
{\small
\begin{tabular}{|c|c||c|c||c|c||c|c|}
\hline
\textbf{RMSE}& ESN  & \textit{LqW}($8$)& \textit{HqW}($8$) & \textit{LqW}($16$)& \textit{HqW}($16$) & \textit{LqW}($24$) & \textit{HqW}($24$)\\
\hline   
parity($\tau=2$) & 0.49 & 0.005  & \textbf{0.0031} & 0.011& \textbf{0.0027}& 0.013 &\textbf{0.0020}\\
parity($\tau=4$) & 0.49 & 0.033   &  {\bf0.014}& 0.083 & {\bf0.013}& 0.33 & {\bf0.012}\\
NARMA10 & 0.041  & 0.105 & \textbf{0.069} &  0.113 &  \textbf{0.062} & 0.117  &\textbf{0.052}\\
delay($\tau=2$)&  0.026  & 0.013 & \textbf{0.0055} & 0.013 &  \textbf{0.0038} & 0.014 & \textbf{0.0038}\\
delay($\tau=4$) & 0.095  & 0.49  & \textbf{0.016} & 0.49 & \textbf{0.014} &  0.49 &\textbf{0.012}\\
\hline
\end{tabular}
}
\caption{Average RMSE for (i) the classical RC ESN; (ii) the QRC LqW; (iii) our QRC HqW. The number between brackets after the name of the model is relative to the number of harmonic oscillators in the quantum linear network.}
\label{tab:1}
\end{table*}

 Figure \ref{fig:traj}   shows the prediction  on one of the test sets of the Monte Carlo simulation for the NARMA10 benchmark: the true-signal is in grayc  HqW($24$) in orange, LqW($24$) in green, and  ESN($200$) in blue.

\subsection{Exploiting universality}
In this section, we demonstrate how to exploit the universality of the qWiener architecture (proven in Section~\ref{sec:univ}) to implement a QRC with long-term memory.
  In RC, a key determinant of performance is the \emph{memory capacity}, 
that is, the ability of the system to retain information about past inputs.  
Consider the delay task
$$y(t) = u(t-\tau),
$$
which requires perfect recall of the input \(\tau\) time units in the past.  
This can be written in convolution form as
$$
y(t) = \int_{-\infty}^{\infty} \delta(t - s - \tau)\, u(s)\, ds,
$$
where the impulse response of an ideal delay is the shifted Dirac delta  
\(h(t) = \delta(t-\tau)\).  

\begin{figure*}
\includegraphics[width=13cm]{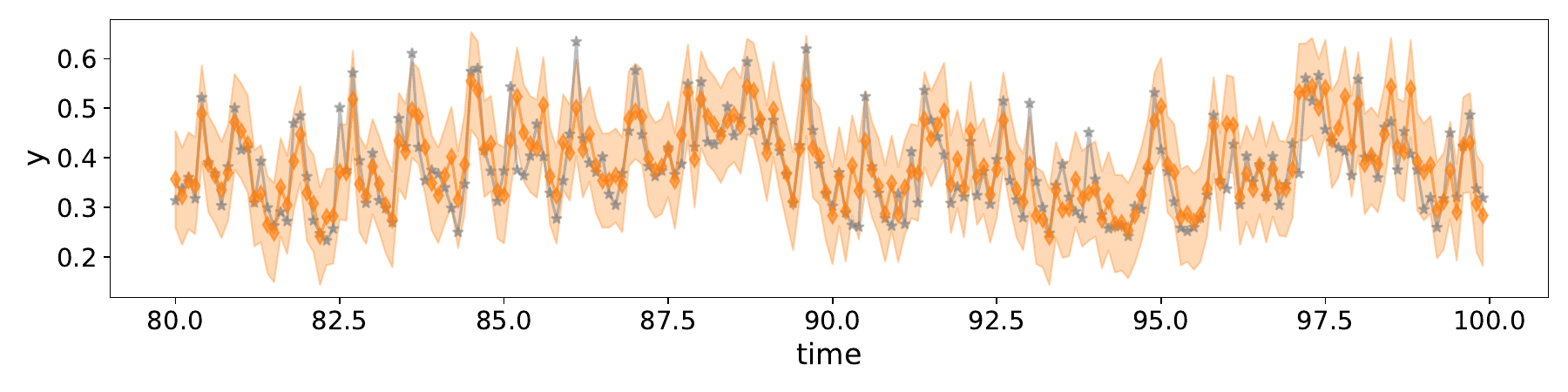}
\caption{Prediction on a test set for the NARMA10 benchmark for HqW($24$): the mean is the orange line and the orange-shaded region is the 95\% credible interval.}
\label{fig:traj1}
\end{figure*}

The Laplace transform of this impulse response is $H(s) = e^{-s\tau}$. Its transfer function \(e^{-s\tau}\) is \emph{not} a rational 
function.  
Thus, an ideal delay is an infinite-dimensional system, and cannot be realised by any 
finite-dimensional LTI  dynamical system. The solution proposed in \citep{voelker2019legendre} is to employ a Pad{\'e} approximation to find a \emph{rational} approximation of 
the delay operator by replacing
$$
e^{-s\tau} \approx \frac{P_m(s)}{Q_n(s)},
$$
where $P_m$ and $Q_n$ are polynomials of degrees $m$ and $n$, respectively (with $m\leq n$).  
This rational approximation can be implemented exactly by an LTI system. Moreover, as a RC, this design realises a computational basis corresponding to orthogonal Legendre polynomials \citep{voelker2019legendre}. This is important because using orthogonal basis functions provides a principled alternative for constructing RCs, compared to the traditional approach of using randomised basis functions (which we used for our HqW).

In this section, we fix $n = 8$ and $m = 7$, and use the controllable 
canonical form to derive the matrices $A$, $B$, and $C$ of a classical LTI system. We then apply Lemma \ref{lem:qABCD} and Theorem \ref{th:qABCD} to derive the QSS that implements the impulse response  Pad{\'e}  approximation of the delay.

We repeat these steps for different delays $\tau=k\Delta$ with $k=1,2,\dots,14$. This allows us to obtain a qWiener architecture with 14 linear networks, which we denote as PadeqW($14$). The performance of this QRC are reported in the below table together with the performance of HqW($24$) from the previous table for comparison:
\begin{table}[h]
{\small
\begin{tabular}{|c|c|c|}
\hline
\textbf{RMSE}& \textit{HqW}($24$)&  \textit{PadeqW}($14$)\\
\hline   
parity($\tau=4$) & 0.012&  \textbf{0.0078}\\
NARMA10 & 0.052 & \textbf{0.044}  \\
delay($\tau=4$) &0.012 & \textbf{0.0025} \\
\hline
\end{tabular}
}
\caption{Average RMSE for  our QRC HqW(24) versus  our QRC PadeqW(14).}
\end{table}
This shows that the system achieves practically zero error on the parity ($\tau=4$) and delay ($\tau=4$) tasks, as expected since its memory is $\tau=14$. It also attains a very low RMSE on the NARMA10 task.  
 Figure \ref{fig:traj1}   shows the prediction  on  the same test set of Figure \ref{fig:traj1} for PadeqW($14$), showing the improved performance.

\subsection{Implementation}
The source code to run the experiments in the manuscript was implemented in Python and  is available at \url{https://github.com/benavoli/qWiener}.

\section{Conclusions}
In this work, we have derived a quantum version of the Wiener architecture (qWiener) for Reservoir Computing (RC). 
This architecture consists of a quantum linear network followed by a weak measurement and a classical static nonlinear readout layer. 
We have proved that, although the quantum dynamics impose constraints on the parameters of the quantum linear network, the universality of the architecture can still be guaranteed. 
 { Additionally, we have shown that the qWiener architecture can be obtained by concatenating simple harmonic oscillators and using only passive optical components such as phase shifters, beam splitters, and mirrors. To prove universality of the architecture, however, we use an alternative structure which requires active optical components.}
Finally, through numerical experiments, we have verified that our quantum-RC achieves very high performance on standard benchmarks for RC.

Future work -- { In this paper, we have established the universality of the qWiener architecture. However, it remains an open question what advantages it might offer (beyond the inherent high dimensionality of quantum systems) over the classical Wiener architecture, which is also universal. One possible research direction is to explore quantum-RC architectures consisting of additional layers (deep), but where the weak measurement occurs only before the readout layer. 
This would allow a deeper exploitation of the quantum features of the system and extend the classical deep Wiener architectures proposed in \citep{forgione2021dynonet} to the quantum case.}
Finally, we have not explicitly emphasised that our proposed architecture is Bayesian, due to the Gaussian Processes used in the final readout layer. This provides, at no additional implementation cost, a principled quantification of uncertainty, as shown in Figure \ref{fig:traj1}. We believe this aspect is particularly important for reservoir computing architectures, which can often exhibit larger prediction errors due to their random initialisation. As future work, we plan to compare our uncertainty quantification with that obtained from fully trained classical networks, such as those considered in \citep{de2021use,benavoli2025dynogp}.

\appendix

\section{Computation of the Gaussian Process posterior}
\label{app:gp_posterior}
Assuming the observation noises $\varepsilon_i$ are independent, the likelihood of the observed data given the latent function values is:
\begin{equation}
\label{eq:GPlike}
p(\mathcal{D}_m \mid f(Y)) = \mathcal{N}(\mathbf{o}; f(Y), \sigma^2 I_m),
\end{equation}
where $\mathbf{o} = [o(t_1), \dots, o(t_m)]^\top$, $Y = [\mathbf{y}^\top(t_1), \dots, \mathbf{y}^\top(t_m)]^\top$, $t_i=i\Delta$ and $I_m$ is the identity matrix of size $m$.

\textit{Prior:} We place a Gaussian process prior on the function $f \sim \mathcal{GP}(0, K^o_{\boldsymbol{\vartheta}})$, where $K^o_{\boldsymbol{\vartheta}}$ is a kernel function defined over the feature vectors $\mathbf{y}(t)$. This implies:
\begin{equation}
\label{eq:GPprior}
p(f(Y) \mid \boldsymbol{\vartheta}) = \mathcal{N}(f(Y); 0, K^o_{\boldsymbol{\vartheta}}(Y,Y)),
\end{equation}
where $(K^o_{\boldsymbol{\vartheta}}(Y,Y))_{ij} = K^o_{\boldsymbol{\vartheta}}(\mathbf{y}(t_i), \mathbf{y}(t_j))$.

\textit{Posterior:} Combining the likelihood \eqref{eq:GPlike} with the prior \eqref{eq:GPprior}, the posterior distribution over the function values is:
\[
p(f(Y) \mid \mathcal{D}_m, \boldsymbol{\theta}) = \mathcal{N}(f(Y); \mu_p(Y), K_p(Y,Y)),
\]
with posterior mean and covariance given by:
\begin{align}
\mu_p(Y) &= K^o_{\boldsymbol{\vartheta}}(Y,Y) \left( K^o_{\boldsymbol{\vartheta}}(Y,Y) + \sigma^2 I_m \right)^{-1} \mathbf{o}, \\
\nonumber
K_p(Y,Y) &= K^o_{\boldsymbol{\vartheta}}(Y,Y) \\
&- K^o_{\boldsymbol{\vartheta}}(Y,Y) \left( K^o_{\boldsymbol{\vartheta}}(Y,Y) + \sigma^2 I_T \right)^{-1} K^o_{\boldsymbol{\vartheta}}(Y,Y).
\end{align}
This operation is equivalent to linear regression (with nonlinear basis functions), but performed in the dual space of the kernel \citep[Chap\.2]{rasmussen2010gaussian}.

\textit{Predictive posterior:} For a new input $\mathbf{y}(t_r)$, the predictive distribution is also Gaussian:
\[
p(f(\mathbf{y}(t_r)) \mid \mathcal{D}_m, \boldsymbol{\vartheta}) = \mathcal{N}(\mu_p(\mathbf{y}(t_r)), K_p(\mathbf{y}(t_r), \mathbf{y}(t'_r))),
\]
with
\begin{align}
&\mu_p(\mathbf{y}(t_r)) = K^o_{\boldsymbol{\vartheta}}(\mathbf{y}(t_r), Y) \left( K^o_{\boldsymbol{\vartheta}}(Y,Y) + \sigma^2 I_m \right)^{-1} \mathbf{o}, \\
\nonumber
 &K_p(\mathbf{y}(t_r), \mathbf{y}(t_r))= K^o_{\boldsymbol{\vartheta}}(\mathbf{y}(t_r), \mathbf{y}(t_r)) \\
&- K^o_{\boldsymbol{\vartheta}}(\mathbf{y}(t_r), Y) \left( K^o_{\boldsymbol{\vartheta}}(Y,Y) + \sigma^2 I_m \right)^{-1} K^o_{\boldsymbol{\vartheta}}(Y, \mathbf{y}(t_r)).
\end{align}

\textit{Hyperparameter selection:} The hyperparameters of the kernel $\boldsymbol{\vartheta}$ and noise variance $\sigma^2$ are the only unknowns. They are
typically estimated by maximising the marginal likelihood:
\[
p(\mathcal{D}_m \mid \boldsymbol{\vartheta}, \sigma^2) = \mathcal{N}(\mathbf{o}; 0, K^o_{\boldsymbol{\vartheta}}(Y,Y) + \sigma^2 I_m).
\]

\section{Von Mises--Fisher distribution}
\label{app:VM}
Assume $s_{1},s_{2} \sim N({\bf m},1/\kappa I_2)$, with ${\bf m} \in \mathbb{R}^2$  and $||{\bf m} ||=1$. Then, conditioned on $\sqrt{s^2_{1}+s^2_{2}}=1$, the vector $\mathbf {s}=(s_{1},s_{2})^\top$ is  von Mises-Fisher distributed with probability density function
\begin{equation}
 {\displaystyle f(\mathbf {s} ;{\boldsymbol {m}},\kappa )=C_{2}(\kappa )\exp \left({\kappa {\boldsymbol {m}}^\top\mathbf {s} }\right),}
\end{equation}
where
\begin{equation}
{\displaystyle C_{2}(\kappa )={\frac {1}{2\pi \mathcal{I}_{0}(\kappa )}},}
\end{equation}
with  $\mathcal{I}_n({\displaystyle \kappa })$ being the modified Bessel function of the first kind of order $n$. It then holds
\begin{align}
\label{eq:VMmean}
 E(\mathbf {s})&=r_1\mathbf {m},\\
 \label{eq:mean2nd}
  E(\mathbf {s}\mathbf {s}^\top)&=\frac{1 - r_{2}}{2}\, I_{2} + r_{2}\,\boldsymbol{m}\boldsymbol{m}^\top, \\
  \label{eq:mean2ndcenter}
Cov(\mathbf {s})&=\frac{1 - r_{2}}{2}\, I_{2} + \bigl(r_{2} - r_{1}^{2}\bigr)\boldsymbol{m}\boldsymbol{m}^\top,
\end{align}
where $r_n=\mathcal{I}_{n}(\kappa)/\mathcal{I}_{0}(\kappa)$.

\section{Proofs}
\label{app:proofs}
This section presents the proofs of the main results of the paper, along with the intermediate lemmas required for their derivation.

\paragraph{Proof of Proposition \ref{prop:linBlocks}}
The output of one LTI block is
$$
\begin{aligned}
y_j(t)&=\int_{0}^t h_j(t-\tau)u(\tau)d\tau\\
&=\int_{0}^t \sum_{i=1}^{n_c} e^{-\alpha_i  (t-\tau)} \big( \beta^{(j)}_i \cos(\omega_i (t-\tau)) \\
&~~~~~~~~- \gamma^{(j)}_i \sin(\omega_i (t-\tau)) \big)u(\tau)d\tau,
\end{aligned}
$$
and, therefore,
{\scriptsize
$$
\begin{aligned}
&{\bf y}^\top(t){\bf y}(t')=\frac{1}{d}\sum_{j=1}^{d} \int_{0}^t \int_{0}^{t'}  \\
& \Big(\sum_{i=1}^{n_c}  e^{-\alpha_i  (t-\tau)} \big( \beta^{(j)}_i \cos(\omega_i (t-\tau)) - \gamma^{(j)}_i \sin(\omega_i (t-\tau)) \big)u(\tau)d\tau\Big)\\
&\Big(\sum_{i=1}^{n_c}  e^{-\alpha_i  (t'-\tau')} \big( \beta^{(j)}_i \cos(\omega_i (t'-\tau')) - \gamma^{(j)}_i \sin(\omega_i (t'-\tau')) \big)u(\tau')d\tau'\Big)
\end{aligned}
$$}
We recall that $\beta^{(j)}_i,\gamma^{(j)}_i \sim N(0,1/\sqrt{n_c})$ independently for $i=1,\dots,n_c$ and $j=1,\dots,d$. Therefore, for large $d$ and $n_c$, we have that
{\footnotesize
$$
\begin{aligned}
&{\bf y}^\top(t){\bf y}(t')\approx \frac{1}{d}\sum_{j=1}^{d} \int_{0}^t \int_{0}^{t'}  \\
& \Big(\sum_{i=1}^{n_c}  e^{-\alpha_i  (t-\tau)-\alpha_i  (t'-\tau')}  \big( \beta^{2(j)}_i \cos(\omega_i (t-\tau)) \cos(\omega_i (t'-\tau'))\\
&+ \gamma^{2(j)}_i \sin(\omega_i (t-\tau)) \sin(\omega_i (t'-\tau'))\big)u(\tau)u(\tau')d\tau d\tau'\Big)\\
&\approx \frac{1}{n_c}\sum_{i=1}^{n_c} \int_{0}^t \int_{0}^{t'}  \\
&\Big(e^{-\alpha_i  (t-\tau)-\alpha_i  (t'-\tau')}  \big(   \cos(\omega_i (t-\tau)) \cos(\omega_i (t'-\tau'))\\
&+  \sin(\omega_i (t-\tau)) \sin(\omega_i (t'-\tau'))\big)u(\tau)u(\tau')d\tau d\tau'\Big)\\
&\approx \frac{1}{n_c}\sum_{i=1}^{n_c}  \int_{0}^t \int_{0}^{t'}  \\
&\Big(e^{-\alpha_i  (t-\tau)-\alpha_i  (t'-\tau')}    \cos(\omega_i (t-\tau-(t'-\tau')))\\
&u(\tau)u(\tau')d\tau d\tau'\Big)\\
&\approx \int_{0}^t \int_{0}^{t'} \int_{0}^{\infty} \int_{-\infty}^{\infty}   \\
&\Big(e^{-\alpha  (t-\tau - (t'-\tau'))}   cos(\omega(t-\tau-(t'-\tau')))\\
&\phi_{\boldsymbol{\theta}}(\alpha,\omega)d\alpha d\omega u(\tau)u(\tau')d\tau d\tau'\Big)\\
&= \int_{0}^t \int_{0}^{t'} K_{{\boldsymbol{\theta}},n_c}(t-\tau,t'-\tau') u(\tau)u(\tau')d\tau d\tau'\Big).\\
\end{aligned}
$$}
\hfill\qed

\paragraph{Proof of Proposition \ref{prop:deepKernel}}
The result is derived from Proposition \ref{prop:linBlocks}.
For the polynomial kernel, we take advantage of the fact that, for large $d$ and $n_c$, the term $\mathbf{y}(t)^\top \mathbf{y}(t') $ is approximately equal to \eqref{eq:yy}.
For the square-exponential kernel, note that
$\|\mathbf{y}(t) - \mathbf{y}(t')\|^2=\mathbf{y}^\top(t)\mathbf{y}(t) - 2\mathbf{y}^\top(t)\mathbf{y}(t')+\mathbf{y}^\top(t')\mathbf{y}(t')$ and the rest follows similarly to the derivation for the polynomial kernel.
\hfill\qed

\paragraph{Proof of Theorem \ref{th:1}}
Note that $Q^{(r)}=\tfrac{1}{2}I_{2s}$
and assume that $V_{\infty} =Q^{(r)}$, then
$$
\begin{aligned}
G_{\infty}&=V_{\infty} C_q^{(r)\top} + L^{(r)}Q^{(r)} D_q^{(r)\top}\\
&=\frac{1}{2}(C_q^{(r)\top} + L^{(r)} D_q^{(r)\top})\\
&=\frac{1}{2}\begin{pmatrix} I_s & 0_s\end{pmatrix}(C^{(r)\top} + L^{(r)} D^{(r)\top})\begin{pmatrix} I_s & 0_s\end{pmatrix}^\top\\
&=\frac{1}{2}\begin{pmatrix} I_s & 0_s\end{pmatrix}(C^{(r)\top} + L^{(r)})\begin{pmatrix} I_s & 0_s\end{pmatrix}^\top\\
&=0
\end{aligned}
$$
where the last equality follows by the physical realisability constraint \eqref{eq:physicalreal2}. Note in fact that, under the assumption $s=n$, we have that
$$
\begin{aligned}
J_n^{(r)}L^{(r)}+C^{(r)\top}J_s^{(r)} &=J_n^{(r)}(J_n^{(r)}C^{(r)\top}J_n^{(r)})+C^{(r)\top}J_s^{(r)}\\
&=-C^{(r)\top}J_n^{(r)}+C^{(r)\top}J_s^{(r)}=0.
\end{aligned}
$$
It remains to show that for $V_{\infty}=\tfrac{1}{2}I_{2s}$ the following equality holds $0=A^{(r)}V_{\infty}+ V_{\infty}A^{(r)\top} + L^{(r)}\tilde{Q}^{(r)}L^{(r)\top}$.
This results follows by
\begin{lemma}[\citep{koga2012dissipation}]
Assume $Q^{(r)}=\tfrac{1}{2}I_{2s}$, that is, $Q^{(r)}$ is the covariance of the vacuum. If a completely passive  QSS is strictly
stable, then it must evolve into a vacuum state, that is $V_{\infty}=\tfrac{1}{2}I_n$.
\end{lemma}
\hfill\qed

\paragraph{Useful lemma}
\begin{lemma}
\label{lem:mat}
Consider the matrix
$$
A = \begin{pmatrix}
-\frac{\alpha^2+\beta^2}{2} & -\omega\\
\omega & -\frac{\alpha^2+\beta^2}{2}  \\
\end{pmatrix},
$$
its matrix exponential is
\begin{equation}
\label{eq:matexpexpr}
e^{At} = e^{-(\alpha^2+\beta^2)t}\begin{pmatrix}
        \cos(\omega t) & - \sin(\omega t) \\
         \sin(\omega t) & \sin(\omega t)
        \end{pmatrix}.
\end{equation}
\end{lemma}
\begin{proof}
Rewrite
$$
\begin{aligned}
&\begin{pmatrix}
-\frac{\alpha^2+\beta^2}{2} & -\omega\\
\omega & -\frac{\alpha^2+\beta^2}{2}  \\
\end{pmatrix}=\begin{pmatrix}
-\frac{\alpha^2+\beta^2}{2} & 0\\
0 & -\frac{\alpha^2+\beta^2}{2}  \\
\end{pmatrix}+\begin{pmatrix}
0 & -\omega\\
\omega& 0  \\
\end{pmatrix}\\
&=X+Y
        \end{aligned}
$$

Since the two matrices are such that $XY=YX$ (they commute), from the properties of the matrix exponential, we have that $e^{(X+Y)t}=e^{Xt}e^{Yt}$, where
$$
\begin{aligned}
e^{Xt}&=e^{-\frac{\alpha^2+\beta^2}{2}t}I_2\\
e^{Yt}&=\Omega_2^\dagger \begin{pmatrix}
         e^{\iota \omega t} &0\\
         0 &  e^{-\iota \omega t}
        \end{pmatrix}\Omega_2\\
        &=\Omega_2^\dagger \begin{pmatrix}
        \cos(\omega t)+\iota \sin(\omega t) &0\\
         0 &  \cos(\omega t)-\iota \sin(\omega t)
        \end{pmatrix}\Omega_2
\\
&=\begin{pmatrix}
        \cos(\omega t) & - \sin(\omega t) \\
         \sin(\omega t) & \sin(\omega t)
        \end{pmatrix},
\end{aligned}
$$
where
\begin{equation}
\Omega_n = \dfrac{1}{\sqrt{2}}\begin{pmatrix}
    I_n & ~~\iota I_n \\
    I_n & -\iota I_n \\
\end{pmatrix},
\end{equation}
\end{proof}

\paragraph{Proof of Lemma \ref{lem:h_harmonic}}

 Since $n=1$, and the system is completely passive then $M_2=0$, $N_2=0$,  $\Re(M_1)= \Re(M_1)^\top$, $\Im(M_1)=-\Im(M_1)^\top=0$. Furthermore, since $n=1$, we have that $M_1,S \in \mathbb{C}$ and $N_1 \in \mathbb{C}^{1 \times s}$. Therefore, from \eqref{eq:Ar}, we can derive that
  {\small
 $$
 \begin{aligned}
\begin{pmatrix}
    \Im(M_1+M_2) & \Re(M_1-M_2)\\
       -\Re(M_1+M_2) & -\Im(-M_1+M_2)\\ 
\end{pmatrix}&=\begin{pmatrix}
         0 &\Re(M_1) \\-\Re(M_1) & 0
        \end{pmatrix}
 \end{aligned}
 $$}
 and
 {\scriptsize
  $$
 \begin{aligned}
&-\frac{1}{2} C^{(r)\sharp}C^{(r)}=\\
&\frac{1}{2}\begin{pmatrix}
         0 & 1\\-1 & 0
        \end{pmatrix}\begin{pmatrix}
         \Re(N_1)^\top &  \Im(N_1)^\top\\- \Im(N_1)^\top &  \Re(N_1)^\top
        \end{pmatrix} J_s^{(r)}\begin{pmatrix}
         \Re(N_1) & - \Im(N_1)\\\Im(N_1) &  \Re(N_1)
        \end{pmatrix}\\
&=\frac{1}{2}\begin{pmatrix}
         0 & 1\\-1 & 0
        \end{pmatrix}\begin{pmatrix}
         \Re(N_1)^\top &  \Im(N_1)^\top\\- \Im(N_1)^\top &  \Re(N_1)^\top
        \end{pmatrix}\begin{pmatrix}
         0 & I_s\\-I_s &0
        \end{pmatrix}\begin{pmatrix}
         \Re(N_1) & - \Im(N_1)\\\Im(N_1) &  \Re(N_1)
        \end{pmatrix}\\
        &=\frac{1}{2}\begin{pmatrix}
         0 & 1\\-1 & 0
        \end{pmatrix}\begin{pmatrix}
         \Re(N_1)^\top &  \Im(N_1)^\top\\- \Im(N_1)^\top &  \Re(N_1)^\top
        \end{pmatrix}\begin{pmatrix}
         \Im(N_1) & \Re(N_1)\\- \Re(N_1) &  \Im(N_1)
        \end{pmatrix}\\
        &=\begin{pmatrix}
        -\lambda &0\\0 &-\lambda
        \end{pmatrix}
 \end{aligned}
 $$}
 where $\lambda=\frac{1}{2} \Im(N_1)^\top \Im(N_1) +\frac{1}{2}\Re(N_1)^\top \Re(N_1)$.  Therefore, we have
   $$
 \begin{aligned}
{A}^{(r)}
        &=\begin{pmatrix}
         -\lambda &\Re(M_1)\\-\Re(M_1) & -\lambda 
        \end{pmatrix}
 \end{aligned}
$$
The matrices ${C}^{(r)},L^{(r)}$ are
   $$
 \begin{aligned}
{C}^{(r)}
        &=\begin{pmatrix}
         \Re(N_1) & - \Im(N_1)\\\Im(N_1) &  \Re(N_1)
        \end{pmatrix}\\
L^{(r)}
        &=\begin{pmatrix}
         -\Re(N_1)^\top & - \Im(N_1)^\top\\\Im(N_1)^\top &  -\Re(N_1)^\top
        \end{pmatrix} \begin{pmatrix}
    \Re(S) & -\Im(S)\\
    \Im(S) & \Re(S)
\end{pmatrix}
 \end{aligned}
$$
Therefore, by exploiting \eqref{eq:matexpexpr}, the impulse response  is
 {\scriptsize
   $$
 \begin{aligned}
&{C}^{(r)}e^{{A}^{(r)}t}{L}^{(r)}\\
        &=-e^{-\lambda t}\begin{pmatrix}
         \Re(N_1) & - \Im(N_1)\\\Im(N_1) &  \Re(N_1)
        \end{pmatrix} \begin{pmatrix}
         \cos(\omega t) & \sin(\omega t)\\- \sin(\omega t)&  \cos(\omega t)
        \end{pmatrix}\\
        &
        \begin{pmatrix}
         \Re(N_1)^\top &  \Im(N_1)^\top\\-\Im(N_1)^\top &  \Re(N_1)^\top
        \end{pmatrix} \begin{pmatrix}
    \Re(S) & -\Im(S)\\
    \Im(S) & \Re(S)
\end{pmatrix}\\
        &=-e^{-\lambda t} \begin{pmatrix}
          \Re(N_1)\cos(\omega t)+ \Im(N_1)\sin(\omega t) & \Re(N_1)\sin(\omega t)- \Im(N_1)\cos(\omega t)\\  \Im(N_1)\cos(\omega t)- \Re(N_1)\sin(\omega t)&  \Re(N_1)\cos(\omega t)+ \Im(N_1)\sin(\omega t)
        \end{pmatrix}\\
        &
        \begin{pmatrix}
         \Re(N_1)^\top &  \Im(N_1)^\top\\-\Im(N_1)^\top &  \Re(N_1)^\top
        \end{pmatrix} \begin{pmatrix}
    \Re(S) & -\Im(S)\\
    \Im(S) & \Re(S)
\end{pmatrix}\\
        &=-e^{-\lambda t} \begin{pmatrix}
          f_1& f_2\\  f_3 & f_4
        \end{pmatrix}\begin{pmatrix}
    \Re(S) & -\Im(S)\\
    \Im(S) & \Re(S)
\end{pmatrix}\\
 \end{aligned}
$$}
where  $\omega=\Re(M_1)$ and
{\scriptsize
   $$
 \begin{aligned}
f_1 &=\Re(N_1) \Re(N_1)^\top\cos(\omega t)+ \Im(N_1)\Re(N_1)^\top\sin(\omega t) \\
&- \Re(N_1)\Im(N_1)^\top\sin(\omega t)+ \Im(N_1)\Im(N_1)^\top\cos(\omega t)\\
&=(\Re(N_1) \Re(N_1)^\top+ \Im(N_1)\Im(N_1)^\top)\cos(\omega t)\\
f_2 &=\Re(N_1) \Im(N_1)^\top\cos(\omega t)+ \Im(N_1)\Im(N_1)^\top\sin(\omega t)\\
&+ \Re(N_1)\Re(N_1)^\top\sin(\omega t)- \Im(N_1)\Re(N_1)^\top\cos(\omega t)\\
&= (\Re(N_1)\Re(N_1)^\top+ \Im(N_1)\Im(N_1)^\top)\sin(\omega t)\\
f_3 &=\Im(N_1)\Re(N_1)^\top\cos(\omega t)- \Re(N_1)\Re(N_1)^\top\sin(\omega t)  \\&-\Re(N_1)\Im(N_1)^\top\cos(\omega t)- \Im(N_1)\Im(N_1)^\top\sin(\omega t)\\
&=-(\Re(N_1)\Re(N_1)^\top+ \Im(N_1)\Im(N_1)^\top)\sin(\omega t)\\
f_4 &=\Im(N_1)\Im(N_1)^\top\cos(\omega t)- \Re(N_1)\Im(N_1)^\top\sin(\omega t)\\
&+  \Re(N_1)\Re(N_1)^\top\cos(\omega t)+\Im(N_1)\Re(N_1)^\top\sin(\omega t)\\
&=(\Re(N_1) \Re(N_1)^\top+ \Im(N_1)\Im(N_1)^\top)\cos(\omega t)\\
 \end{aligned}
$$}
Finally, we have that
{\scriptsize
   $$
 \begin{aligned}
&{C}^{(r)}e^{{A}^{(r)}t}{L}^{(r)}
        &=-e^{-\lambda t} \left( \begin{pmatrix}
         \cos(\omega t) & \sin(\omega t)\\- \sin(\omega t)&  \cos(\omega t)
        \end{pmatrix} \otimes \Lambda\right)\begin{pmatrix}
    \Re(S) & -\Im(S)\\
    \Im(S) & \Re(S)
\end{pmatrix}\\
 \end{aligned}
$$}
where $\Lambda=(\Re(N_1) \Re(N_1)^\top+ \Im(N_1)\Im(N_1)^\top)$.
\hfill\qed

\paragraph{Proof of Corollary \ref{co:h_harmonic}}
The result follows from Lemma \ref{lem:h_harmonic} by direct substitution of the parameters.
\hfill\qed

\paragraph{Proof of Proposition \ref{prop:concatenation}}
W.l.o.g.\ we can prove the results for $n_c=2$ due to the symmetries in the concatenation product.
\begin{align}
\nonumber
  C^{(r)}
    &=\begin{pmatrix}
    \Re(N_1+N_2) & \Im(-N_1+N_2)\\
    \Im(N_1+N_2) & \Re(N_1-N_2)
\end{pmatrix}=\begin{pmatrix}
   \alpha_1 & 0 & 0 & 0\\
   0& \alpha_2 & 0 & 0\\
   0 & 0 &    \alpha_1 & 0 \\
   0 & 0 &   0 &  \alpha_2 \\
\end{pmatrix}
 \end{align}
and, therefore,
{\scriptsize
\begin{align}
\nonumber
&C^{(r)\sharp}C^{(r)}\\
\nonumber
    &=-\begin{pmatrix}
   0_2 &I_2\\
   -I_2 & 0_2
\end{pmatrix}\begin{pmatrix}
   \alpha_1 & 0 & 0 & 0\\
   0& \alpha_2 & 0 & 0\\
   0 & 0 &    \alpha_1 & 0 \\
   0 & 0 &   0 &  \alpha_2 \\
\end{pmatrix}\begin{pmatrix}
   0_2 &I_2\\
   -I_2 & 0_2
\end{pmatrix}\begin{pmatrix}
   \alpha_1 & 0 & 0 & 0\\
   0& \alpha_2 & 0 & 0\\
   0 & 0 &    \alpha_1 & 0 \\
   0 & 0 &   0 &  \alpha_2 \\
\end{pmatrix}\\
\nonumber
&=-\begin{pmatrix}
   0 & 0 & \alpha_1 & 0\\
   0& 0 &0 & \alpha_2\\
        -\alpha_1 & 0 &0 & 0 \\
   0 &  -\alpha_2 &0 & 0 \\
\end{pmatrix}\begin{pmatrix}
   0 & 0 & \alpha_1 & 0\\
   0& 0 &0 & \alpha_2\\
        -\alpha_1 & 0 &0 & 0 \\
   0 &  -\alpha_2 &0 & 0 \\
\end{pmatrix}
\nonumber
\\&=-\begin{pmatrix}
   -\alpha^2_1 & 0 & 0 & 0\\
   0& -\alpha^2_2 & 0 & 0\\
   0 & 0 &    -\alpha^2_1 & 0 \\
   0 & 0 &   0 &  \alpha^2_2 \\
\end{pmatrix}
 \end{align}}
The matrix $A^{(r)}$ is
{\scriptsize
\begin{align}
\nonumber
A^{(r)}&=\begin{pmatrix}
    \Im(M_1+M_2) & \Re(M_1-M_2)\\
       -\Re(M_1+M_2) & -\Im(-M_1+M_2)\\ 
\end{pmatrix} -\frac{1}{2} C^{(r)\sharp}C^{(r)},\\
\nonumber
&=\begin{pmatrix}
   -\frac{\alpha^2_1}{2} & 0 & -\omega_1  & 0\\
   0& -\frac{\alpha^2_2}{2} & 0 & -\omega_2\\
   \omega_1  & 0 &    -\frac{\alpha^2_1}{2} & 0 \\
   0 & \omega_2 &   0 &  -\frac{\alpha^2_2}{2} \\
\end{pmatrix}
\end{align}}
and
\begin{align}
\nonumber
L^{(r)}&=-  C^{(r)\sharp}D^{(r)}=-\begin{pmatrix}
   \alpha_1 & 0 & 0 & 0\\
   0& \alpha_2 & 0 & 0\\
   0 & 0 &    \alpha_1 & 0 \\
   0 & 0 &   0 &  \alpha_2 \\
\end{pmatrix}D^{(r)},\\
\nonumber
D^{(r)}&= \begin{pmatrix}
    s_{11} &0 & -s_{12} & 0\\
     0& s_{21} &0 &  -s_{22}\\
     s_{12} & 0 &  s_{11} &0\\
      0&  s_{22} & 0& s_{21}
\end{pmatrix}.
\end{align}
Therefore, by using a derivation similar to the one in Lemma \ref{lem:mat}, we have that
{\scriptsize
$$
\begin{aligned}
&e^{A^{(r)}t}+\\
&\left(
\begin{smallmatrix}
 e^{-\frac{1}{2} \alpha _1^2 t} \cos (\omega_1t) & 0 &
   -e^{-\frac{1}{2} \alpha _1^2 t} \sin (\omega_1t) &
   0 \\
 0 & e^{-\frac{1}{2} \alpha _2^2 t} \cos (\omega_2t) &
   0 & -e^{-\frac{1}{2} \alpha _2^2 t} \sin (t \text{$\omega
   $2}) \\
 e^{-\frac{1}{2} \alpha _1^2 t} \sin (\omega_1t) & 0 &
   e^{-\frac{1}{2} \alpha _1^2 t} \cos (\omega_1t) & 0
   \\
 0 & e^{-\frac{1}{2} \alpha _2^2 t} \sin (\omega_2t) &
   0 & e^{-\frac{1}{2} \alpha _2^2 t} \cos (\omega_2t)
   \\
\end{smallmatrix}
\right).
\end{aligned}
$$}
The concatenation product simply assembles the different components together, without making any connections between them.  Therefore, if we  observe the first quadrature for each component, and assume the input
${\bf u}(t)=\begin{pmatrix}1 &1  & 0  & 0\end{pmatrix}^\top u(t)$,  then the impulse response becomes:
{\scriptsize
\begin{align}
\nonumber
&{\bf h}(t)\\
\nonumber
&=  \left(-\begin{pmatrix}
\alpha_1 & 0 & 0 &0\\
0 & \alpha_2 &0 &0\\
\end{pmatrix}e^{A^{(r)}t}\begin{pmatrix}
   \alpha_1 & 0 & 0 & 0\\
   0& \alpha_2 & 0 & 0\\
   0 & 0 &    \alpha_1 & 0 \\
   0 & 0 &   0 &  \alpha_2 \\
\end{pmatrix}+\begin{pmatrix}
1 & 0 & 0 &0\\
0 & 1 &0 &0\\
\end{pmatrix}\delta(t)\right)\\
&~~~~~\begin{pmatrix}
    s_{11} &0 & -s_{12} & 0\\
     0& s_{21} &0 &  -s_{22}\\
     s_{12} & 0 &  s_{11} &0\\
      0&  s_{22} & 0& s_{21}
\end{pmatrix}\begin{pmatrix}
1 \\ 1 \\ 0 \\0\\
\end{pmatrix}\\
\nonumber
&= \Bigg(-\begin{pmatrix}
\alpha_1 & 0 & 0 &0\\
0 & \alpha_2 &0 &0\\
\end{pmatrix}e^{A^{(r)}t}\begin{pmatrix}
   \alpha_1 & 0 & 0 & 0\\
   0& \alpha_2 & 0 & 0\\
   0 & 0 &    \alpha_1 & 0 \\
   0 & 0 &   0 &  \alpha_2 \\
\end{pmatrix}\\
&~~~~~+\begin{pmatrix}
1 & 0 & 0 &0\\
0 & 1 &0 &0\\
\end{pmatrix}\delta(t)\Bigg)\begin{pmatrix}
    s_{11}\\
     s_{21} \\
     s_{12}\\
      s_{22}
\end{pmatrix}\\
\nonumber
&= \Bigg(-\left(\begin{smallmatrix}
 \alpha^2_1e^{-\frac{1}{2} \alpha _1^2 t} & 0\\
 0 & \alpha^2_2 e^{-\frac{1}{2} \alpha _2^2 t}  \\
\end{smallmatrix}\right)\left(
\begin{smallmatrix}
  \cos (\omega_1t) & 0 &
   - \sin (\omega_1t) &
   0 \\
 0 &  \cos (\omega_2t) &
   0 & - \sin (\omega_2t) \\
\end{smallmatrix}
\right)\\
\nonumber
&~~~~~+\begin{pmatrix}
1 & 0 & 0 &0\\
0 & 1 &0 &0\\
\end{pmatrix}\delta(t)\Bigg)\begin{pmatrix}
    s_{11}\\
     s_{21} \\
     s_{12}\\
      s_{22}
\end{pmatrix}\\
\nonumber
&=\left(
\begin{array}{c}
 -s_{11}\alpha^2_1e^{-\frac{1}{2} \alpha _1^2 t} \cos (\omega_1t) +s_{12}\alpha^2_1e^{-\frac{1}{2} \alpha _1^2 t} \sin (\omega_1t)\\
 -s_{21} \alpha^2_2 e^{-\frac{1}{2} \alpha _2^2 t} \cos (\omega_2t)+s_{22}\alpha^2_2 e^{-\frac{1}{2} \alpha _2^2 t} \sin (\omega_2t) \\
\end{array}
\right)+\left(
\begin{array}{c}
s_{11}\\
s_{21}
\end{array}
\right)\delta(t)
\end{align}}
By summing the components, we obtain \eqref{eq:hquantum}.
\hfill\qed

\paragraph{Proof of Theorem \ref{th:clt}}
Consider $\tilde{h}(t)$ in \eqref{eq:hquantum}, then,
from \eqref{eq:VMmean},  we have that
\begin{align}
\nonumber
\mu_h(t)&=-E\left(\alpha^2 e^{-\frac{\alpha^2t}{2} } \left(
s_{1}\cos ( \omega t ) -s_{2} \sin ( \omega t )\right)\right)\\
\nonumber
&=-\frac{r_1}{\sqrt{2}}\int \alpha^2 e^{-\frac{\alpha^2t}{2} } \left(\cos ( \omega t ) - \sin ( \omega t )\right)\phi_{\boldsymbol{\theta}}(\alpha,\omega) d\alpha d\omega,\\
\label{eq:proofEh}
&=-\frac{r_1}{\sqrt{2}}\int \alpha^2 e^{-\frac{\alpha^2t}{2} } \cos ( \omega t ) \phi_{\boldsymbol{\theta}}(\alpha,\omega) d\alpha d\omega,
\end{align}
where we have exploited that $\phi_{\boldsymbol{\theta}}(\alpha,\omega) =\phi_{\boldsymbol{\theta}}(\alpha,-\omega)$ and $\sin(x)=-\sin(-x)$.
For the second moments, we first compute the expectation w.r.t.\   ${\bf s}=(s_{1},s_{2})^\top$:
{\scriptsize
 \begin{align}
\nonumber
&E\left(\alpha^2 e^{-\frac{\alpha^2t}{2} } \left(
s_{1}\cos ( \omega t ) -s_{2} \sin ( \omega t )\right)\right)\hspace{-0.9mm}\left(\alpha^2 e^{-\frac{\alpha^2t'}{2} } \left(
s_{1}\cos ( \omega t' ) -s_{2} \sin ( \omega t' )\right)\right)\\
\nonumber
&= \alpha^4 e^{-\frac{\alpha^2(t+t')}{2} } E\left({\bf s}^\top{\bf c}_{\omega t}{\bf c}_{\omega t'}^\top{\bf s}\right),\\
\nonumber
&=\alpha^4 e^{-\frac{\alpha^2(t+t')}{2} }Tr\left({\bf c}_{\omega t}{\bf c}_{\omega t'}^\top E({\bf s}{\bf s}^\top)\right),\\
\nonumber
&=\alpha^4 e^{-\frac{\alpha^2(t+t')}{2} }Tr\left({\bf c}_{\omega t}{\bf c}_{\omega t'}^\top \left(\frac{1 - r_{2}}{2}\, I_{2} + \frac{r_{2}}{2}\,{\bf 1}_2{\bf 1}_2^\top\right)\right)\\
\nonumber
&=\alpha^4 e^{-\frac{\alpha^2(t+t')}{2} }\Big(\tfrac{1 - r_{2}}{2}\,\left(\cos ( \omega t )\cos ( \omega t' )+ \sin ( \omega t )\sin ( \omega t' )\right) \\
&+ \tfrac{r_{2}}{2}\,\left(\cos ( \omega t ) - \sin ( \omega t )\right)\left(\cos ( \omega t' ) - \sin ( \omega t' )\right)\Big)\\
\nonumber
&=\alpha^4 e^{-\frac{\alpha^2(t+t')}{2} }\left(\tfrac{1 - r_{2}}{2}\,\cos ( \omega (t-t') )  + \tfrac{r_{2}}{2}\left(\cos ( \omega (t-t') ) -\sin ( \omega (t-t') ) \right) \right)\\
\nonumber
&=\alpha^4 e^{-\frac{\alpha^2(t+t')}{2} }\frac{1 }{2}\left(\cos ( \omega (t-t') )  -r_2\sin ( \omega (t-t') ) \right)
\end{align}}
where ${\bf c}_{\omega t}=\left(\cos ( \omega t ), - \sin ( \omega t )\right)^\top$.
By computing the expectation w.r.t.
{\footnotesize
 \begin{align}
 \nonumber
&\int \alpha^4 e^{-\frac{\alpha^2(t+t')}{2} }\frac{1 }{2}\left(\cos ( \omega (t-t') )  -r_2\sin ( \omega (t-t') ) \right)\phi_{\boldsymbol{\theta}}(\alpha,\omega) d\alpha d\omega\\
&=\int \alpha^4 e^{-\frac{\alpha^2(t+t')}{2} }\frac{1 }{2}\cos ( \omega (t-t') )  \phi_{\boldsymbol{\theta}}(\alpha,\omega) d\alpha d\omega
\end{align}}
where we have exploited that $\phi_{\boldsymbol{\theta}}(\alpha,\omega) =\phi_{\boldsymbol{\theta}}(\alpha,-\omega)$ and $\sin(x)=-\sin(-x)$.
For the covariance function, we have that
{\scriptsize
$$
\begin{aligned}
 K_{\boldsymbol{\theta}}(t,t')
 &=\int \alpha^4 e^{-\frac{\alpha^2(t+t')}{2} }\frac{1 }{2}\cos ( \omega (t-t') )  \phi_{\boldsymbol{\theta}}(\alpha,\omega) d\alpha d\omega -\mu_h(t)\mu_h(t') \\
\end{aligned}
$$}
Assuming that
$\phi_{\boldsymbol{\theta}}(\alpha^2,\omega)=\frac{\tfrac{\alpha^2}{2}}{\pi(\omega^2+\tfrac{\alpha^2
4}{4})}\text{Unif}(\alpha^2; \alpha_m,\alpha_M)$ and exploiting the result:
\begin{equation}
\int \cos(\omega \tau) \frac{1}{\pi} \frac{\beta}{\omega^2+\beta^{2}} d\omega=e^{-\beta| \tau| }
\end{equation}
then, for $\beta=\alpha^2/2$,
\begin{equation}
\begin{aligned}
\mu_h(t)&= \int   \alpha^2 e^{-\tfrac{\alpha^2}{2} t}  e^{-\tfrac{\alpha^2}{2}| t| } \text{Unif}(\alpha^2; \alpha_m,\alpha_M) d\alpha^2 \\
&=\frac{(1+ a_m t) e^{- a_m t
   }-e^{- a_M t }
   (1+ a_M t)}{t ^2 (a_M-a_m)}
   \end{aligned}
\end{equation}
and
\begin{equation}
\begin{aligned}
 &\int   \alpha^4 e^{-\tfrac{\alpha^2}{2} (t+s)}  e^{-\tfrac{\alpha^2}{2}| t-s| } \text{Unif}(\alpha^2; \alpha_m,\alpha_M)  d\alpha^2 \\
 &=  \frac{ (1+
   a_m \tau +\frac{1}{2}
   a_m^2 \tau^2  ) e^{- a_m
   \tau}-e^{- a_M
   \tau} (1+ a_M \tau +\frac{1}{2} a_M^2 \tau^2)}{2 \tau^3 (a_M-a_m)}
   \end{aligned}
\end{equation}
where $\tau=\max(s,t)$.
\hfill\qed

\paragraph{Proof of Lemma \ref{lem:qABCD}}
We need to compute $C^{(r)}e^{A^{(r)}t}L^{(r)}$. Note that
\begin{align*}
e^{A^{(r)}t}&= \begin{pmatrix}
         e^{At} & 0\\
         0 & e^{Zt}
        \end{pmatrix}\\
        \nonumber
L^{(r)}&=J_{n}^{(r)} C^{(r)\top} J_{s}^{(r)}\\
&=\begin{pmatrix}
         0 & I_n\\
         -I_n & 0
        \end{pmatrix}  \begin{pmatrix}
         C^{(r)}_1 & 0_{1 \times n}\\
         0_{1 \times n} & C^{(r)}_2\\
        \end{pmatrix}^\top\begin{pmatrix}
         0 & I_s\\
         -I_s & 0
         \end{pmatrix}\\
         &=\begin{pmatrix}
         -C^{(r)\top}_2 & 0_{n \times 1}\\
         0_{n \times 1} & -C^{(r)\top}_1\\
        \end{pmatrix}
 \end{align*}
 Therefore,  we have that
\begin{align*}
&C^{(r)}e^{A^{(r)}t}L^{(r)}=\\
&\begin{pmatrix}
         C^{(r)}_1 & 0_{1 \times n}\\
        0_{1 \times n} & C^{(r)}_2\\
        \end{pmatrix}\begin{pmatrix}
         e^{At} & 0\\
         0 & e^{Zt}
        \end{pmatrix}\begin{pmatrix}
         -C^{(r)\top}_2 & 0_{n \times 1}\\
         0_{n \times 1} & -C^{(r)\top}_1\\
        \end{pmatrix}\\
&=\begin{pmatrix}
         -C^{(r)}_1  e^{At} C^{(r)\top}_2 & 0\\
         0 & -C^{(r)}_2  e^{Zt} C^{(r)\top}_1\\
        \end{pmatrix}
 \end{align*}
It remains to prove that the QSS is physically realisable. The physical realisability conditions is:

\begin{align*}
&0=A^{(r)}J_n^{(r)}+J_n^{(r)}A^{(r)\top}+L^{(r)}J_s^{(r)}L^{(r)\top}\\
&=\begin{pmatrix}
         A & 0\\
         0 &Z
        \end{pmatrix}\begin{pmatrix}
         0 & I_n\\
         -I_n & 0
        \end{pmatrix}\\
        &+\begin{pmatrix}
         0 & I_n\\
         -I_n & 0
        \end{pmatrix}\begin{pmatrix}
         A^\top & 0\\
         0 & Z^\top
        \end{pmatrix}\\
        &+L^{(r)} \begin{pmatrix}
         0 & I_2\\
         -I_2 & 0
         \end{pmatrix}L^{(r)\top}\\
&=
         \begin{pmatrix}
          0 & A+Z^\top\\
        -Z-A^\top&  0
        \end{pmatrix}\\
        &+\begin{pmatrix}
        0 & -C^{(r)\top}_2 \\
         C^{(r)\top}_1 & 0\\
        \end{pmatrix} \begin{pmatrix}
         -C^{(r)\top}_2 & 0\\
         0 & -C^{(r)\top}_1\\
        \end{pmatrix}^\top\\
&= \begin{pmatrix}
          0 &   A+Z^\top+C^{(r)\top}_2 C^{(r)}_1  \\
         -Z-A^\top-C^{(r)\top}_1 C^{(r)}_2  &  0
        \end{pmatrix}
\end{align*}
which implies that $Z=-A^\top-C^{(r)\top}_1 C^{(r)}_2$. 

The  condition $J_n^{(r)}L^{(r)}+C^{(r)\top}J_s^{(r)} D^{(r)}=0$ is satisfied as  $D^{(r)}=I_4$
and $L^{(r)}=J_{n}^{(r)} C^{(r)\top} J_{s}^{(r)}$ and so
$$
J_n^{(r)}L^{(r)}+C^{(r)\top}J_s^{(r)} D^{(r)} =-C^{(r)\top} J_{s}^{(r)}+C^{(r)\top}J_s^{(r)}=0.
$$
Finally, the following condition is satisfied as well: $D^{(r)}J_s^{(r)}D^{(r)\top}-J_s^{(r)}=0$.

\hfill\qed

\paragraph{Proof of Theorem \ref{th:qABCD}}
The first part of the proof follows directly from Lemma \ref{lem:qABCD}. About the second part, it follows from a standard result in control theory for LTI systems, which states that if the  matrix $\begin{pmatrix}B,AB,A^2B,\dots,A^{n-1}B\end{pmatrix}$ has rank $n$, there exists a matrix $R$ such that the closed-loop matrix $A-BR$ is stable \citep[Ch.5]{Astrom2004}. A similar result holds for $-A-BR$, since $\begin{pmatrix}B,-AB,(-A)^2B,\dots,(-A)^{n-1}B\end{pmatrix}$ has also rank $n$.
\hfill\qed

\paragraph{Proof of Corollary \ref{co:active}}
We note that the $C^{(r)}$ matrix is
$$
C^{(r)}=\begin{pmatrix}
         R & 0_{1 \times n}\\
         -C & 0_{1 \times n}\\
        0_{1 \times n} & B^\top\\
        0_{1 \times n} & 0_{1 \times n}\\
        \end{pmatrix}
$$
We aim to write it as 
$$
C^{(r)}= \begin{pmatrix}
    \Re(N_1+N_2) & \Im(-N_1+N_2)\\
    \Im(N_1+N_2) & \Re(N_1-N_2)
\end{pmatrix}
$$
with $N_1,N_2 \in \mathbb{C}^{2 \times n}$. It can easily be seen that this requires that $\Re(N_2)\neq 0_{2 \times n}$. Therefore, from Definition \ref{def:comppas}, it follows that the QSS described by~\eqref{eq:Aruniv}--\eqref{eq:Druniv} is not completely passive.

\section{Additional numerical experiments}
\label{app:addexp}
The following table reports the comparison between $LqW(24)$ and
$HqW(24)$ using  a polynomial nonlinearity (with degree $3$) in the readout layer instead of the square-exponential kernel (which we instead used  in the experiments in the last two columns of Table \ref{tab:1}).
\begin{table}[h]
{\small
\begin{tabular}{|c||c|c|}
\hline
\textbf{RMSE}&  \textit{LqW}($24$)  & \textit{HqW}($24$) \\
\hline   
parity($\tau=2$) & 0.065 &\textbf{0.0017}\\
parity($\tau=4$) &0.49 & {\bf0.15}\\
NARMA10 & 0.087  &\textbf{0.065}\\
delay($\tau=2$)& 0.33 & \textbf{0.04}\\
delay($\tau=4$) &  0.49 &\textbf{0.13}\\
\hline
\end{tabular}
}
\caption{Average RMSE for (i) the QRC LqW; (ii) our QRC HqW; using  a polynomial nonlinearity (with degree $3$) in the readout layer. The number between brackets after the name of the model is relative to the number of harmonic oscillators in the quantum linear network.}
\label{tab:1bis}
\end{table}
It can be noticed that 
$HqW(24)$ outperforms $LqW(24)$  also in this case. The performance  with the  polynomial nonlinearity is generally worse than the one obtained with the squared-exponential kernel in Table \ref{tab:1}.
\newpage

\bibliography{biblio}

\begin{thebibliography}{55}%
\makeatletter
\providecommand \@ifxundefined [1]{%
 \@ifx{#1\undefined}
}%
\providecommand \@ifnum [1]{%
 \ifnum #1\expandafter \@firstoftwo
 \else \expandafter \@secondoftwo
 \fi
}%
\providecommand \@ifx [1]{%
 \ifx #1\expandafter \@firstoftwo
 \else \expandafter \@secondoftwo
 \fi
}%
\providecommand \natexlab [1]{#1}%
\providecommand \enquote  [1]{``#1''}%
\providecommand \bibnamefont  [1]{#1}%
\providecommand \bibfnamefont [1]{#1}%
\providecommand \citenamefont [1]{#1}%
\providecommand \href@noop [0]{\@secondoftwo}%
\providecommand \href [0]{\begingroup \@sanitize@url \@href}%
\providecommand \@href[1]{\@@startlink{#1}\@@href}%
\providecommand \@@href[1]{\endgroup#1\@@endlink}%
\providecommand \@sanitize@url [0]{\catcode `\\12\catcode `\$12\catcode
  `\&12\catcode `\#12\catcode `\^12\catcode `\_12\catcode `\%12\relax}%
\providecommand \@@startlink[1]{}%
\providecommand \@@endlink[0]{}%
\providecommand \url  [0]{\begingroup\@sanitize@url \@url }%
\providecommand \@url [1]{\endgroup\@href {#1}{\urlprefix }}%
\providecommand \urlprefix  [0]{URL }%
\providecommand \Eprint [0]{\href }%
\providecommand \doibase [0]{https://doi.org/}%
\providecommand \selectlanguage [0]{\@gobble}%
\providecommand \bibinfo  [0]{\@secondoftwo}%
\providecommand \bibfield  [0]{\@secondoftwo}%
\providecommand \translation [1]{[#1]}%
\providecommand \BibitemOpen [0]{}%
\providecommand \bibitemStop [0]{}%
\providecommand \bibitemNoStop [0]{.\EOS\space}%
\providecommand \EOS [0]{\spacefactor3000\relax}%
\providecommand \BibitemShut  [1]{\csname bibitem#1\endcsname}%
\let\auto@bib@innerbib\@empty
\bibitem [{\citenamefont {Jaeger}(2001)}]{jaeger2001echo}%
  \BibitemOpen
  \bibfield  {author} {\bibinfo {author} {\bibfnamefont {H.}~\bibnamefont
  {Jaeger}},\ }\bibfield  {title} {\bibinfo {title} {The “echo state”
  approach to analysing and training recurrent neural networks-with an erratum
  note},\ }\href@noop {} {\bibfield  {journal} {\bibinfo  {journal} {Bonn,
  Germany: German national research center for information technology gmd
  technical report}\ }\textbf {\bibinfo {volume} {148}},\ \bibinfo {pages} {13}
  (\bibinfo {year} {2001})}\BibitemShut {NoStop}%
\bibitem [{\citenamefont {Maass}\ \emph {et~al.}(2002)\citenamefont {Maass},
  \citenamefont {Natschl{\"a}ger},\ and\ \citenamefont
  {Markram}}]{maass2002real}%
  \BibitemOpen
  \bibfield  {author} {\bibinfo {author} {\bibfnamefont {W.}~\bibnamefont
  {Maass}}, \bibinfo {author} {\bibfnamefont {T.}~\bibnamefont
  {Natschl{\"a}ger}},\ and\ \bibinfo {author} {\bibfnamefont {H.}~\bibnamefont
  {Markram}},\ }\bibfield  {title} {\bibinfo {title} {Real-time computing
  without stable states: A new framework for neural computation based on
  perturbations},\ }\href@noop {} {\bibfield  {journal} {\bibinfo  {journal}
  {Neural computation}\ }\textbf {\bibinfo {volume} {14}},\ \bibinfo {pages}
  {2531} (\bibinfo {year} {2002})}\BibitemShut {NoStop}%
\bibitem [{\citenamefont {Fujii}\ and\ \citenamefont
  {Nakajima}(2017)}]{Fujii2017}%
  \BibitemOpen
  \bibfield  {author} {\bibinfo {author} {\bibfnamefont {K.}~\bibnamefont
  {Fujii}}\ and\ \bibinfo {author} {\bibfnamefont {K.}~\bibnamefont
  {Nakajima}},\ }\bibfield  {title} {\bibinfo {title} {Harnessing
  disordered-ensemble quantum dynamics for machine learning},\ }\href
  {https://doi.org/10.1103/PhysRevApplied.8.024030} {\bibfield  {journal}
  {\bibinfo  {journal} {Phys. Rev. Appl.}\ }\textbf {\bibinfo {volume} {8}},\
  \bibinfo {pages} {024030} (\bibinfo {year} {2017})}\BibitemShut {NoStop}%
\bibitem [{\citenamefont {Nakajima}\ \emph {et~al.}(2019)\citenamefont
  {Nakajima}, \citenamefont {Fujii}, \citenamefont {Negoro}, \citenamefont
  {Mitarai},\ and\ \citenamefont {Kitagawa}}]{PhysRevApplied.11.034021}%
  \BibitemOpen
  \bibfield  {author} {\bibinfo {author} {\bibfnamefont {K.}~\bibnamefont
  {Nakajima}}, \bibinfo {author} {\bibfnamefont {K.}~\bibnamefont {Fujii}},
  \bibinfo {author} {\bibfnamefont {M.}~\bibnamefont {Negoro}}, \bibinfo
  {author} {\bibfnamefont {K.}~\bibnamefont {Mitarai}},\ and\ \bibinfo {author}
  {\bibfnamefont {M.}~\bibnamefont {Kitagawa}},\ }\bibfield  {title} {\bibinfo
  {title} {Boosting computational power through spatial multiplexing in quantum
  reservoir computing},\ }\href
  {https://doi.org/10.1103/PhysRevApplied.11.034021} {\bibfield  {journal}
  {\bibinfo  {journal} {Phys. Rev. Appl.}\ }\textbf {\bibinfo {volume} {11}},\
  \bibinfo {pages} {034021} (\bibinfo {year} {2019})}\BibitemShut {NoStop}%
\bibitem [{\citenamefont {Kutvonen}\ \emph {et~al.}(2020)\citenamefont
  {Kutvonen}, \citenamefont {Fujii},\ and\ \citenamefont
  {Sagawa}}]{kutvonen2020optimizing}%
  \BibitemOpen
  \bibfield  {author} {\bibinfo {author} {\bibfnamefont {A.}~\bibnamefont
  {Kutvonen}}, \bibinfo {author} {\bibfnamefont {K.}~\bibnamefont {Fujii}},\
  and\ \bibinfo {author} {\bibfnamefont {T.}~\bibnamefont {Sagawa}},\
  }\bibfield  {title} {\bibinfo {title} {Optimizing a quantum reservoir
  computer for time series prediction},\ }\href@noop {} {\bibfield  {journal}
  {\bibinfo  {journal} {Scientific reports}\ }\textbf {\bibinfo {volume}
  {10}},\ \bibinfo {pages} {14687} (\bibinfo {year} {2020})}\BibitemShut
  {NoStop}%
\bibitem [{\citenamefont {Chen}\ \emph {et~al.}(2020)\citenamefont {Chen},
  \citenamefont {Nurdin},\ and\ \citenamefont {Yamamoto}}]{chen2020temporal}%
  \BibitemOpen
  \bibfield  {author} {\bibinfo {author} {\bibfnamefont {J.}~\bibnamefont
  {Chen}}, \bibinfo {author} {\bibfnamefont {H.~I.}\ \bibnamefont {Nurdin}},\
  and\ \bibinfo {author} {\bibfnamefont {N.}~\bibnamefont {Yamamoto}},\
  }\bibfield  {title} {\bibinfo {title} {Temporal information processing on
  noisy quantum computers},\ }\href@noop {} {\bibfield  {journal} {\bibinfo
  {journal} {Physical Review Applied}\ }\textbf {\bibinfo {volume} {14}},\
  \bibinfo {pages} {024065} (\bibinfo {year} {2020})}\BibitemShut {NoStop}%
\bibitem [{\citenamefont {Nokkala}\ \emph {et~al.}(2021)\citenamefont
  {Nokkala}, \citenamefont {Mart{\'\i}nez-Pe{\~n}a}, \citenamefont {Giorgi},
  \citenamefont {Parigi}, \citenamefont {Soriano},\ and\ \citenamefont
  {Zambrini}}]{nokkala2021gaussian}%
  \BibitemOpen
  \bibfield  {author} {\bibinfo {author} {\bibfnamefont {J.}~\bibnamefont
  {Nokkala}}, \bibinfo {author} {\bibfnamefont {R.}~\bibnamefont
  {Mart{\'\i}nez-Pe{\~n}a}}, \bibinfo {author} {\bibfnamefont {G.~L.}\
  \bibnamefont {Giorgi}}, \bibinfo {author} {\bibfnamefont {V.}~\bibnamefont
  {Parigi}}, \bibinfo {author} {\bibfnamefont {M.~C.}\ \bibnamefont
  {Soriano}},\ and\ \bibinfo {author} {\bibfnamefont {R.}~\bibnamefont
  {Zambrini}},\ }\bibfield  {title} {\bibinfo {title} {Gaussian states of
  continuous-variable quantum systems provide universal and versatile reservoir
  computing},\ }\href@noop {} {\bibfield  {journal} {\bibinfo  {journal}
  {Communications Physics}\ }\textbf {\bibinfo {volume} {4}},\ \bibinfo {pages}
  {53} (\bibinfo {year} {2021})}\BibitemShut {NoStop}%
\bibitem [{\citenamefont {Mujal}\ \emph {et~al.}(2021)\citenamefont {Mujal},
  \citenamefont {Martínez-Peña}, \citenamefont {Nokkala}, \citenamefont
  {García-Beni}, \citenamefont {Giorgi}, \citenamefont {Soriano},\ and\
  \citenamefont {Zambrini}}]{mujal2021}%
  \BibitemOpen
  \bibfield  {author} {\bibinfo {author} {\bibfnamefont {P.}~\bibnamefont
  {Mujal}}, \bibinfo {author} {\bibfnamefont {R.}~\bibnamefont
  {Martínez-Peña}}, \bibinfo {author} {\bibfnamefont {J.}~\bibnamefont
  {Nokkala}}, \bibinfo {author} {\bibfnamefont {J.}~\bibnamefont
  {García-Beni}}, \bibinfo {author} {\bibfnamefont {G.~L.}\ \bibnamefont
  {Giorgi}}, \bibinfo {author} {\bibfnamefont {M.~C.}\ \bibnamefont
  {Soriano}},\ and\ \bibinfo {author} {\bibfnamefont {R.}~\bibnamefont
  {Zambrini}},\ }\bibfield  {title} {\bibinfo {title} {Opportunities in quantum
  reservoir computing and extreme learning machines},\ }\href
  {https://doi.org/https://doi.org/10.1002/qute.202100027} {\bibfield
  {journal} {\bibinfo  {journal} {Advanced Quantum Technologies}\ }\textbf
  {\bibinfo {volume} {4}},\ \bibinfo {pages} {2100027} (\bibinfo {year}
  {2021})}\BibitemShut {NoStop}%
\bibitem [{\citenamefont {Mart{\'\i}nez-Pe{\~n}a}\ \emph
  {et~al.}(2023)\citenamefont {Mart{\'\i}nez-Pe{\~n}a}, \citenamefont
  {Nokkala}, \citenamefont {Giorgi}, \citenamefont {Zambrini},\ and\
  \citenamefont {Soriano}}]{martinez2023information}%
  \BibitemOpen
  \bibfield  {author} {\bibinfo {author} {\bibfnamefont {R.}~\bibnamefont
  {Mart{\'\i}nez-Pe{\~n}a}}, \bibinfo {author} {\bibfnamefont {J.}~\bibnamefont
  {Nokkala}}, \bibinfo {author} {\bibfnamefont {G.~L.}\ \bibnamefont {Giorgi}},
  \bibinfo {author} {\bibfnamefont {R.}~\bibnamefont {Zambrini}},\ and\
  \bibinfo {author} {\bibfnamefont {M.~C.}\ \bibnamefont {Soriano}},\
  }\bibfield  {title} {\bibinfo {title} {Information processing capacity of
  spin-based quantum reservoir computing systems},\ }\href@noop {} {\bibfield
  {journal} {\bibinfo  {journal} {Cognitive Computation}\ }\textbf {\bibinfo
  {volume} {15}},\ \bibinfo {pages} {1440} (\bibinfo {year}
  {2023})}\BibitemShut {NoStop}%
\bibitem [{\citenamefont {Hu}\ \emph {et~al.}(2024)\citenamefont {Hu},
  \citenamefont {Khan}, \citenamefont {Bronn}, \citenamefont {Angelatos},
  \citenamefont {Rowlands}, \citenamefont {Ribeill},\ and\ \citenamefont
  {T{\"u}reci}}]{hu2024overcoming}%
  \BibitemOpen
  \bibfield  {author} {\bibinfo {author} {\bibfnamefont {F.}~\bibnamefont
  {Hu}}, \bibinfo {author} {\bibfnamefont {S.~A.}\ \bibnamefont {Khan}},
  \bibinfo {author} {\bibfnamefont {N.~T.}\ \bibnamefont {Bronn}}, \bibinfo
  {author} {\bibfnamefont {G.}~\bibnamefont {Angelatos}}, \bibinfo {author}
  {\bibfnamefont {G.~E.}\ \bibnamefont {Rowlands}}, \bibinfo {author}
  {\bibfnamefont {G.~J.}\ \bibnamefont {Ribeill}},\ and\ \bibinfo {author}
  {\bibfnamefont {H.~E.}\ \bibnamefont {T{\"u}reci}},\ }\bibfield  {title}
  {\bibinfo {title} {Overcoming the coherence time barrier in quantum machine
  learning on temporal data},\ }\href@noop {} {\bibfield  {journal} {\bibinfo
  {journal} {Nature communications}\ }\textbf {\bibinfo {volume} {15}},\
  \bibinfo {pages} {7491} (\bibinfo {year} {2024})}\BibitemShut {NoStop}%
\bibitem [{\citenamefont {Paparelle}\ \emph {et~al.}(2025)\citenamefont
  {Paparelle}, \citenamefont {Henaff}, \citenamefont {Garcia-Beni},
  \citenamefont {Gillet}, \citenamefont {Montesinos}, \citenamefont {Giorgi},
  \citenamefont {Soriano}, \citenamefont {Zambrini},\ and\ \citenamefont
  {Parigi}}]{paparelle2025experimental}%
  \BibitemOpen
  \bibfield  {author} {\bibinfo {author} {\bibfnamefont {I.}~\bibnamefont
  {Paparelle}}, \bibinfo {author} {\bibfnamefont {J.}~\bibnamefont {Henaff}},
  \bibinfo {author} {\bibfnamefont {J.}~\bibnamefont {Garcia-Beni}}, \bibinfo
  {author} {\bibfnamefont {E.}~\bibnamefont {Gillet}}, \bibinfo {author}
  {\bibfnamefont {D.}~\bibnamefont {Montesinos}}, \bibinfo {author}
  {\bibfnamefont {G.~L.}\ \bibnamefont {Giorgi}}, \bibinfo {author}
  {\bibfnamefont {M.~C.}\ \bibnamefont {Soriano}}, \bibinfo {author}
  {\bibfnamefont {R.}~\bibnamefont {Zambrini}},\ and\ \bibinfo {author}
  {\bibfnamefont {V.}~\bibnamefont {Parigi}},\ }\bibfield  {title} {\bibinfo
  {title} {Experimental memory control in continuous variable optical quantum
  reservoir computing},\ }\href@noop {} {\bibfield  {journal} {\bibinfo
  {journal} {arXiv preprint arXiv:2506.07279}\ } (\bibinfo {year}
  {2025})}\BibitemShut {NoStop}%
\bibitem [{\citenamefont {Selimovi{\'c}}\ \emph {et~al.}(2025)\citenamefont
  {Selimovi{\'c}}, \citenamefont {Agresti}, \citenamefont {Siemaszko},
  \citenamefont {Morris}, \citenamefont {Daki{\'c}}, \citenamefont {Albiero},
  \citenamefont {Crespi}, \citenamefont {Ceccarelli}, \citenamefont {Osellame},
  \citenamefont {Stobi{\'n}ska} \emph {et~al.}}]{selimovic2025experimental}%
  \BibitemOpen
  \bibfield  {author} {\bibinfo {author} {\bibfnamefont {M.}~\bibnamefont
  {Selimovi{\'c}}}, \bibinfo {author} {\bibfnamefont {I.}~\bibnamefont
  {Agresti}}, \bibinfo {author} {\bibfnamefont {M.}~\bibnamefont {Siemaszko}},
  \bibinfo {author} {\bibfnamefont {J.}~\bibnamefont {Morris}}, \bibinfo
  {author} {\bibfnamefont {B.}~\bibnamefont {Daki{\'c}}}, \bibinfo {author}
  {\bibfnamefont {R.}~\bibnamefont {Albiero}}, \bibinfo {author} {\bibfnamefont
  {A.}~\bibnamefont {Crespi}}, \bibinfo {author} {\bibfnamefont
  {F.}~\bibnamefont {Ceccarelli}}, \bibinfo {author} {\bibfnamefont
  {R.}~\bibnamefont {Osellame}}, \bibinfo {author} {\bibfnamefont
  {M.}~\bibnamefont {Stobi{\'n}ska}}, \emph {et~al.},\ }\bibfield  {title}
  {\bibinfo {title} {Experimental neuromorphic computing based on quantum
  memristor},\ }\href@noop {} {\bibfield  {journal} {\bibinfo  {journal} {arXiv
  preprint arXiv:2504.18694}\ } (\bibinfo {year} {2025})}\BibitemShut {NoStop}%
\bibitem [{\citenamefont {Ghosh}\ \emph {et~al.}(2019)\citenamefont {Ghosh},
  \citenamefont {Opala}, \citenamefont {Matuszewski}, \citenamefont {Paterek},\
  and\ \citenamefont {Liew}}]{ghosh2019quantum}%
  \BibitemOpen
  \bibfield  {author} {\bibinfo {author} {\bibfnamefont {S.}~\bibnamefont
  {Ghosh}}, \bibinfo {author} {\bibfnamefont {A.}~\bibnamefont {Opala}},
  \bibinfo {author} {\bibfnamefont {M.}~\bibnamefont {Matuszewski}}, \bibinfo
  {author} {\bibfnamefont {T.}~\bibnamefont {Paterek}},\ and\ \bibinfo {author}
  {\bibfnamefont {T.~C.}\ \bibnamefont {Liew}},\ }\bibfield  {title} {\bibinfo
  {title} {Quantum reservoir processing},\ }\href@noop {} {\bibfield  {journal}
  {\bibinfo  {journal} {npj Quantum Information}\ }\textbf {\bibinfo {volume}
  {5}},\ \bibinfo {pages} {35} (\bibinfo {year} {2019})}\BibitemShut {NoStop}%
\bibitem [{\citenamefont {Markovi{\'c}}\ and\ \citenamefont
  {Grollier}(2020)}]{markovic2020quantum}%
  \BibitemOpen
  \bibfield  {author} {\bibinfo {author} {\bibfnamefont {D.}~\bibnamefont
  {Markovi{\'c}}}\ and\ \bibinfo {author} {\bibfnamefont {J.}~\bibnamefont
  {Grollier}},\ }\bibfield  {title} {\bibinfo {title} {Quantum neuromorphic
  computing},\ }\href@noop {} {\bibfield  {journal} {\bibinfo  {journal}
  {Applied physics letters}\ }\textbf {\bibinfo {volume} {117}} (\bibinfo
  {year} {2020})}\BibitemShut {NoStop}%
\bibitem [{\citenamefont {Bravo}\ \emph {et~al.}(2022)\citenamefont {Bravo},
  \citenamefont {Najafi}, \citenamefont {Gao},\ and\ \citenamefont
  {Yelin}}]{bravo2022quantum}%
  \BibitemOpen
  \bibfield  {author} {\bibinfo {author} {\bibfnamefont {R.~A.}\ \bibnamefont
  {Bravo}}, \bibinfo {author} {\bibfnamefont {K.}~\bibnamefont {Najafi}},
  \bibinfo {author} {\bibfnamefont {X.}~\bibnamefont {Gao}},\ and\ \bibinfo
  {author} {\bibfnamefont {S.~F.}\ \bibnamefont {Yelin}},\ }\bibfield  {title}
  {\bibinfo {title} {Quantum reservoir computing using arrays of rydberg
  atoms},\ }\href@noop {} {\bibfield  {journal} {\bibinfo  {journal} {PRX
  Quantum}\ }\textbf {\bibinfo {volume} {3}},\ \bibinfo {pages} {030325}
  (\bibinfo {year} {2022})}\BibitemShut {NoStop}%
\bibitem [{\citenamefont {Sakurai}\ \emph {et~al.}(2022)\citenamefont
  {Sakurai}, \citenamefont {Estarellas}, \citenamefont {Munro},\ and\
  \citenamefont {Nemoto}}]{sakurai2022quantum}%
  \BibitemOpen
  \bibfield  {author} {\bibinfo {author} {\bibfnamefont {A.}~\bibnamefont
  {Sakurai}}, \bibinfo {author} {\bibfnamefont {M.~P.}\ \bibnamefont
  {Estarellas}}, \bibinfo {author} {\bibfnamefont {W.~J.}\ \bibnamefont
  {Munro}},\ and\ \bibinfo {author} {\bibfnamefont {K.}~\bibnamefont
  {Nemoto}},\ }\bibfield  {title} {\bibinfo {title} {Quantum extreme reservoir
  computation utilizing scale-free networks},\ }\href@noop {} {\bibfield
  {journal} {\bibinfo  {journal} {Physical Review Applied}\ }\textbf {\bibinfo
  {volume} {17}},\ \bibinfo {pages} {064044} (\bibinfo {year}
  {2022})}\BibitemShut {NoStop}%
\bibitem [{\citenamefont {Nokkala}(2023)}]{nokkala2023online}%
  \BibitemOpen
  \bibfield  {author} {\bibinfo {author} {\bibfnamefont {J.}~\bibnamefont
  {Nokkala}},\ }\bibfield  {title} {\bibinfo {title} {Online quantum time
  series processing with random oscillator networks},\ }\href@noop {}
  {\bibfield  {journal} {\bibinfo  {journal} {Scientific Reports}\ }\textbf
  {\bibinfo {volume} {13}},\ \bibinfo {pages} {7694} (\bibinfo {year}
  {2023})}\BibitemShut {NoStop}%
\bibitem [{\citenamefont {Dudas}\ \emph {et~al.}(2023)\citenamefont {Dudas},
  \citenamefont {Carles}, \citenamefont {Plouet}, \citenamefont {Mizrahi},
  \citenamefont {Grollier},\ and\ \citenamefont
  {Markovi{\'c}}}]{dudas2023quantum}%
  \BibitemOpen
  \bibfield  {author} {\bibinfo {author} {\bibfnamefont {J.}~\bibnamefont
  {Dudas}}, \bibinfo {author} {\bibfnamefont {B.}~\bibnamefont {Carles}},
  \bibinfo {author} {\bibfnamefont {E.}~\bibnamefont {Plouet}}, \bibinfo
  {author} {\bibfnamefont {F.~A.}\ \bibnamefont {Mizrahi}}, \bibinfo {author}
  {\bibfnamefont {J.}~\bibnamefont {Grollier}},\ and\ \bibinfo {author}
  {\bibfnamefont {D.}~\bibnamefont {Markovi{\'c}}},\ }\bibfield  {title}
  {\bibinfo {title} {Quantum reservoir computing implementation on coherently
  coupled quantum oscillators},\ }\href@noop {} {\bibfield  {journal} {\bibinfo
   {journal} {npj Quantum Information}\ }\textbf {\bibinfo {volume} {9}},\
  \bibinfo {pages} {64} (\bibinfo {year} {2023})}\BibitemShut {NoStop}%
\bibitem [{\citenamefont {Hayashi}\ \emph {et~al.}(2023)\citenamefont
  {Hayashi}, \citenamefont {Sakurai}, \citenamefont {Nishio}, \citenamefont
  {Munro},\ and\ \citenamefont {Nemoto}}]{hayashi2023impact}%
  \BibitemOpen
  \bibfield  {author} {\bibinfo {author} {\bibfnamefont {A.}~\bibnamefont
  {Hayashi}}, \bibinfo {author} {\bibfnamefont {A.}~\bibnamefont {Sakurai}},
  \bibinfo {author} {\bibfnamefont {S.}~\bibnamefont {Nishio}}, \bibinfo
  {author} {\bibfnamefont {W.~J.}\ \bibnamefont {Munro}},\ and\ \bibinfo
  {author} {\bibfnamefont {K.}~\bibnamefont {Nemoto}},\ }\bibfield  {title}
  {\bibinfo {title} {Impact of the form of weighted networks on the quantum
  extreme reservoir computation},\ }\href@noop {} {\bibfield  {journal}
  {\bibinfo  {journal} {Physical Review A}\ }\textbf {\bibinfo {volume}
  {108}},\ \bibinfo {pages} {042609} (\bibinfo {year} {2023})}\BibitemShut
  {NoStop}%
\bibitem [{\citenamefont {Boyd}\ and\ \citenamefont
  {Chua}(2003)}]{boyd2003fading}%
  \BibitemOpen
  \bibfield  {author} {\bibinfo {author} {\bibfnamefont {S.}~\bibnamefont
  {Boyd}}\ and\ \bibinfo {author} {\bibfnamefont {L.}~\bibnamefont {Chua}},\
  }\bibfield  {title} {\bibinfo {title} {Fading memory and the problem of
  approximating nonlinear operators with volterra series},\ }\href@noop {}
  {\bibfield  {journal} {\bibinfo  {journal} {IEEE Transactions on circuits and
  systems}\ }\textbf {\bibinfo {volume} {32}},\ \bibinfo {pages} {1150}
  (\bibinfo {year} {2003})}\BibitemShut {NoStop}%
\bibitem [{\citenamefont {Wiener}(1958)}]{wiener1966nonlinear}%
  \BibitemOpen
  \bibfield  {author} {\bibinfo {author} {\bibfnamefont {N.}~\bibnamefont
  {Wiener}},\ }\href@noop {} {\emph {\bibinfo {title} {Nonlinear problems in
  random theory}}}\ (\bibinfo  {publisher} {MIT Press Cambridge, MA},\ \bibinfo
  {year} {1958})\BibitemShut {NoStop}%
\bibitem [{\citenamefont {Garc{\'\i}a-Beni}\ \emph {et~al.}(2023)\citenamefont
  {Garc{\'\i}a-Beni}, \citenamefont {Giorgi}, \citenamefont {Soriano},\ and\
  \citenamefont {Zambrini}}]{garcia2023scalable}%
  \BibitemOpen
  \bibfield  {author} {\bibinfo {author} {\bibfnamefont {J.}~\bibnamefont
  {Garc{\'\i}a-Beni}}, \bibinfo {author} {\bibfnamefont {G.~L.}\ \bibnamefont
  {Giorgi}}, \bibinfo {author} {\bibfnamefont {M.~C.}\ \bibnamefont
  {Soriano}},\ and\ \bibinfo {author} {\bibfnamefont {R.}~\bibnamefont
  {Zambrini}},\ }\bibfield  {title} {\bibinfo {title} {Scalable photonic
  platform for real-time quantum reservoir computing},\ }\href@noop {}
  {\bibfield  {journal} {\bibinfo  {journal} {Physical Review Applied}\
  }\textbf {\bibinfo {volume} {20}},\ \bibinfo {pages} {014051} (\bibinfo
  {year} {2023})}\BibitemShut {NoStop}%
\bibitem [{\citenamefont {Nokkala}\ \emph {et~al.}(2024)\citenamefont
  {Nokkala}, \citenamefont {Giorgi},\ and\ \citenamefont
  {Zambrini}}]{nokkala2024retrieving}%
  \BibitemOpen
  \bibfield  {author} {\bibinfo {author} {\bibfnamefont {J.}~\bibnamefont
  {Nokkala}}, \bibinfo {author} {\bibfnamefont {G.~L.}\ \bibnamefont
  {Giorgi}},\ and\ \bibinfo {author} {\bibfnamefont {R.}~\bibnamefont
  {Zambrini}},\ }\bibfield  {title} {\bibinfo {title} {Retrieving past quantum
  features with deep hybrid classical-quantum reservoir computing},\
  }\href@noop {} {\bibfield  {journal} {\bibinfo  {journal} {Machine Learning:
  Science and Technology}\ }\textbf {\bibinfo {volume} {5}},\ \bibinfo {pages}
  {035022} (\bibinfo {year} {2024})}\BibitemShut {NoStop}%
\bibitem [{\citenamefont {Mujal}\ \emph {et~al.}(2023)\citenamefont {Mujal},
  \citenamefont {Mart{\'\i}nez-Pe{\~n}a}, \citenamefont {Giorgi}, \citenamefont
  {Soriano},\ and\ \citenamefont {Zambrini}}]{mujal2023time}%
  \BibitemOpen
  \bibfield  {author} {\bibinfo {author} {\bibfnamefont {P.}~\bibnamefont
  {Mujal}}, \bibinfo {author} {\bibfnamefont {R.}~\bibnamefont
  {Mart{\'\i}nez-Pe{\~n}a}}, \bibinfo {author} {\bibfnamefont {G.~L.}\
  \bibnamefont {Giorgi}}, \bibinfo {author} {\bibfnamefont {M.~C.}\
  \bibnamefont {Soriano}},\ and\ \bibinfo {author} {\bibfnamefont
  {R.}~\bibnamefont {Zambrini}},\ }\bibfield  {title} {\bibinfo {title}
  {Time-series quantum reservoir computing with weak and projective
  measurements},\ }\href@noop {} {\bibfield  {journal} {\bibinfo  {journal}
  {npj Quantum Information}\ }\textbf {\bibinfo {volume} {9}},\ \bibinfo
  {pages} {16} (\bibinfo {year} {2023})}\BibitemShut {NoStop}%
\bibitem [{\citenamefont {Yasuda}\ \emph {et~al.}(2023)\citenamefont {Yasuda},
  \citenamefont {Suzuki}, \citenamefont {Kubota}, \citenamefont {Nakajima},
  \citenamefont {Gao}, \citenamefont {Zhang}, \citenamefont {Shimono},
  \citenamefont {Nurdin},\ and\ \citenamefont {Yamamoto}}]{yasuda2023quantum}%
  \BibitemOpen
  \bibfield  {author} {\bibinfo {author} {\bibfnamefont {T.}~\bibnamefont
  {Yasuda}}, \bibinfo {author} {\bibfnamefont {Y.}~\bibnamefont {Suzuki}},
  \bibinfo {author} {\bibfnamefont {T.}~\bibnamefont {Kubota}}, \bibinfo
  {author} {\bibfnamefont {K.}~\bibnamefont {Nakajima}}, \bibinfo {author}
  {\bibfnamefont {Q.}~\bibnamefont {Gao}}, \bibinfo {author} {\bibfnamefont
  {W.}~\bibnamefont {Zhang}}, \bibinfo {author} {\bibfnamefont
  {S.}~\bibnamefont {Shimono}}, \bibinfo {author} {\bibfnamefont {H.~I.}\
  \bibnamefont {Nurdin}},\ and\ \bibinfo {author} {\bibfnamefont
  {N.}~\bibnamefont {Yamamoto}},\ }\bibfield  {title} {\bibinfo {title}
  {Quantum reservoir computing with repeated measurements on superconducting
  devices},\ }\href@noop {} {\bibfield  {journal} {\bibinfo  {journal} {arXiv
  preprint arXiv:2310.06706}\ } (\bibinfo {year} {2023})}\BibitemShut {NoStop}%
\bibitem [{\citenamefont {Dambre}\ \emph {et~al.}(2012)\citenamefont {Dambre},
  \citenamefont {Verstraeten}, \citenamefont {Schrauwen},\ and\ \citenamefont
  {Massar}}]{dambre2012information}%
  \BibitemOpen
  \bibfield  {author} {\bibinfo {author} {\bibfnamefont {J.}~\bibnamefont
  {Dambre}}, \bibinfo {author} {\bibfnamefont {D.}~\bibnamefont {Verstraeten}},
  \bibinfo {author} {\bibfnamefont {B.}~\bibnamefont {Schrauwen}},\ and\
  \bibinfo {author} {\bibfnamefont {S.}~\bibnamefont {Massar}},\ }\bibfield
  {title} {\bibinfo {title} {Information processing capacity of dynamical
  systems},\ }\href@noop {} {\bibfield  {journal} {\bibinfo  {journal}
  {Scientific reports}\ }\textbf {\bibinfo {volume} {2}},\ \bibinfo {pages}
  {514} (\bibinfo {year} {2012})}\BibitemShut {NoStop}%
\bibitem [{\citenamefont {Voelker}\ \emph {et~al.}(2019)\citenamefont
  {Voelker}, \citenamefont {Kaji{\'c}},\ and\ \citenamefont
  {Eliasmith}}]{voelker2019legendre}%
  \BibitemOpen
  \bibfield  {author} {\bibinfo {author} {\bibfnamefont {A.}~\bibnamefont
  {Voelker}}, \bibinfo {author} {\bibfnamefont {I.}~\bibnamefont {Kaji{\'c}}},\
  and\ \bibinfo {author} {\bibfnamefont {C.}~\bibnamefont {Eliasmith}},\
  }\bibfield  {title} {\bibinfo {title} {Legendre memory units: Continuous-time
  representation in recurrent neural networks},\ }\href@noop {} {\bibfield
  {journal} {\bibinfo  {journal} {Advances in neural information processing
  systems}\ }\textbf {\bibinfo {volume} {32}} (\bibinfo {year}
  {2019})}\BibitemShut {NoStop}%
\bibitem [{\citenamefont {Rasmussen}\ and\ \citenamefont
  {Nickisch}(2010)}]{rasmussen2010gaussian}%
  \BibitemOpen
  \bibfield  {author} {\bibinfo {author} {\bibfnamefont {C.~E.}\ \bibnamefont
  {Rasmussen}}\ and\ \bibinfo {author} {\bibfnamefont {H.}~\bibnamefont
  {Nickisch}},\ }\bibfield  {title} {\bibinfo {title} {Gaussian processes for
  machine learning (gpml) toolbox},\ }\href@noop {} {\bibfield  {journal}
  {\bibinfo  {journal} {The Journal of Machine Learning Research}\ }\textbf
  {\bibinfo {volume} {11}},\ \bibinfo {pages} {3011} (\bibinfo {year}
  {2010})}\BibitemShut {NoStop}%
\bibitem [{Note1()}]{Note1}%
  \BibitemOpen
  \bibinfo {note} {Here, we considered a deterministic, classical, stable SISO
  LTI system. For a stochastic LTI system with classical Wiener noise, the
  output $y(t)$ becomes stochastic. In this case, the convolution $\DOTSI
  \intop \ilimits@ _{0}^t h(t-\tau )u(\tau ) d\tau $ coincides with the mean of
  $y(t)$.}\BibitemShut {Stop}%
\bibitem [{Note2()}]{Note2}%
  \BibitemOpen
  \bibinfo {note} {We do not consider here the case where $A$ has eigenvalues
  with multiplicity greater than one.}\BibitemShut {Stop}%
\bibitem [{\citenamefont {Zorzi}\ and\ \citenamefont
  {Chiuso}(2018)}]{zorzi2018harmonic}%
  \BibitemOpen
  \bibfield  {author} {\bibinfo {author} {\bibfnamefont {M.}~\bibnamefont
  {Zorzi}}\ and\ \bibinfo {author} {\bibfnamefont {A.}~\bibnamefont {Chiuso}},\
  }\bibfield  {title} {\bibinfo {title} {The harmonic analysis of kernel
  functions},\ }\href@noop {} {\bibfield  {journal} {\bibinfo  {journal}
  {Automatica}\ }\textbf {\bibinfo {volume} {94}},\ \bibinfo {pages} {125}
  (\bibinfo {year} {2018})}\BibitemShut {NoStop}%
\bibitem [{Note3()}]{Note3}%
  \BibitemOpen
  \bibinfo {note} {Lemma \ref {lem:1} relies on a different convergence result
  than that originally used in \protect \citep {zorzi2018harmonic}.
  Specifically, we employ the strong law of large numbers of Monte Carlo
  integration, whereas \protect \citep {zorzi2018harmonic} relies on the
  convergence of Riemann sums.}\BibitemShut {Stop}%
\bibitem [{\citenamefont {Wilson}\ and\ \citenamefont
  {Adams}(2013)}]{wilson2013gaussian}%
  \BibitemOpen
  \bibfield  {author} {\bibinfo {author} {\bibfnamefont {A.}~\bibnamefont
  {Wilson}}\ and\ \bibinfo {author} {\bibfnamefont {R.}~\bibnamefont {Adams}},\
  }\bibfield  {title} {\bibinfo {title} {Gaussian process kernels for pattern
  discovery and extrapolation},\ }in\ \href@noop {} {\emph {\bibinfo
  {booktitle} {International conference on machine learning}}}\ (\bibinfo
  {organization} {PMLR},\ \bibinfo {year} {2013})\ pp.\ \bibinfo {pages}
  {1067--1075}\BibitemShut {NoStop}%
\bibitem [{\citenamefont {Chen}\ \emph {et~al.}(2012)\citenamefont {Chen},
  \citenamefont {Ohlsson},\ and\ \citenamefont {Ljung}}]{chen2012estimation}%
  \BibitemOpen
  \bibfield  {author} {\bibinfo {author} {\bibfnamefont {T.}~\bibnamefont
  {Chen}}, \bibinfo {author} {\bibfnamefont {H.}~\bibnamefont {Ohlsson}},\ and\
  \bibinfo {author} {\bibfnamefont {L.}~\bibnamefont {Ljung}},\ }\bibfield
  {title} {\bibinfo {title} {On the estimation of transfer functions,
  regularizations and gaussian processes—revisited},\ }\href@noop {}
  {\bibfield  {journal} {\bibinfo  {journal} {Automatica}\ }\textbf {\bibinfo
  {volume} {48}},\ \bibinfo {pages} {1525} (\bibinfo {year}
  {2012})}\BibitemShut {NoStop}%
\bibitem [{\citenamefont {Pillonetto}\ \emph {et~al.}(2016)\citenamefont
  {Pillonetto}, \citenamefont {Chen}, \citenamefont {Chiuso}, \citenamefont
  {De~Nicolao},\ and\ \citenamefont {Ljung}}]{pillonetto2016regularized}%
  \BibitemOpen
  \bibfield  {author} {\bibinfo {author} {\bibfnamefont {G.}~\bibnamefont
  {Pillonetto}}, \bibinfo {author} {\bibfnamefont {T.}~\bibnamefont {Chen}},
  \bibinfo {author} {\bibfnamefont {A.}~\bibnamefont {Chiuso}}, \bibinfo
  {author} {\bibfnamefont {G.}~\bibnamefont {De~Nicolao}},\ and\ \bibinfo
  {author} {\bibfnamefont {L.}~\bibnamefont {Ljung}},\ }\bibfield  {title}
  {\bibinfo {title} {Regularized linear system identification using atomic,
  nuclear and kernel-based norms: The role of the stability constraint},\
  }\href@noop {} {\bibfield  {journal} {\bibinfo  {journal} {Automatica}\
  }\textbf {\bibinfo {volume} {69}},\ \bibinfo {pages} {137} (\bibinfo {year}
  {2016})}\BibitemShut {NoStop}%
\bibitem [{\citenamefont {Paulsen}\ and\ \citenamefont
  {Raghupathi}(2016)}]{paulsen2016introduction}%
  \BibitemOpen
  \bibfield  {author} {\bibinfo {author} {\bibfnamefont {V.~I.}\ \bibnamefont
  {Paulsen}}\ and\ \bibinfo {author} {\bibfnamefont {M.}~\bibnamefont
  {Raghupathi}},\ }\href@noop {} {\emph {\bibinfo {title} {An introduction to
  the theory of reproducing kernel Hilbert spaces}}},\ Vol.\ \bibinfo {volume}
  {152}\ (\bibinfo  {publisher} {Cambridge university press},\ \bibinfo {year}
  {2016})\BibitemShut {NoStop}%
\bibitem [{Note4()}]{Note4}%
  \BibitemOpen
  \bibinfo {note} {The finite case is also of interest in machine learning as
  technique in kernel approximation for scalability of kernel methods \protect
  \citep {rahimi2007random}.}\BibitemShut {Stop}%
\bibitem [{\citenamefont {Cho}\ and\ \citenamefont
  {Saul}(2009)}]{cho2009kernel}%
  \BibitemOpen
  \bibfield  {author} {\bibinfo {author} {\bibfnamefont {Y.}~\bibnamefont
  {Cho}}\ and\ \bibinfo {author} {\bibfnamefont {L.}~\bibnamefont {Saul}},\
  }\bibfield  {title} {\bibinfo {title} {Kernel methods for deep learning},\
  }\href@noop {} {\bibfield  {journal} {\bibinfo  {journal} {Advances in neural
  information processing systems}\ }\textbf {\bibinfo {volume} {22}} (\bibinfo
  {year} {2009})}\BibitemShut {NoStop}%
\bibitem [{\citenamefont {Simon}\ \emph {et~al.}(1994)\citenamefont {Simon},
  \citenamefont {Mukunda},\ and\ \citenamefont {Dutta}}]{simon1994quantum}%
  \BibitemOpen
  \bibfield  {author} {\bibinfo {author} {\bibfnamefont {R.}~\bibnamefont
  {Simon}}, \bibinfo {author} {\bibfnamefont {N.}~\bibnamefont {Mukunda}},\
  and\ \bibinfo {author} {\bibfnamefont {B.}~\bibnamefont {Dutta}},\ }\bibfield
   {title} {\bibinfo {title} {Quantum-noise matrix for multimode systems: U (n)
  invariance, squeezing, and normal forms},\ }\href@noop {} {\bibfield
  {journal} {\bibinfo  {journal} {Physical Review A}\ }\textbf {\bibinfo
  {volume} {49}},\ \bibinfo {pages} {1567} (\bibinfo {year}
  {1994})}\BibitemShut {NoStop}%
\bibitem [{\citenamefont {Nurdin}\ and\ \citenamefont
  {Yamamoto}(2017)}]{nurdin2017linear}%
  \BibitemOpen
  \bibfield  {author} {\bibinfo {author} {\bibfnamefont {H.~I.}\ \bibnamefont
  {Nurdin}}\ and\ \bibinfo {author} {\bibfnamefont {N.}~\bibnamefont
  {Yamamoto}},\ }\bibfield  {title} {\bibinfo {title} {Linear dynamical quantum
  systems},\ }\href@noop {} {\bibfield  {journal} {\bibinfo  {journal}
  {Analysis, Synthesis, and Control}\ } (\bibinfo {year} {2017})}\BibitemShut
  {NoStop}%
\bibitem [{\citenamefont {Zhang}\ and\ \citenamefont
  {Dong}(2022)}]{zhang2022linear}%
  \BibitemOpen
  \bibfield  {author} {\bibinfo {author} {\bibfnamefont {G.}~\bibnamefont
  {Zhang}}\ and\ \bibinfo {author} {\bibfnamefont {Z.}~\bibnamefont {Dong}},\
  }\bibfield  {title} {\bibinfo {title} {Linear quantum systems: a tutorial},\
  }\href@noop {} {\bibfield  {journal} {\bibinfo  {journal} {Annual Reviews in
  Control}\ }\textbf {\bibinfo {volume} {54}},\ \bibinfo {pages} {274}
  (\bibinfo {year} {2022})}\BibitemShut {NoStop}%
\bibitem [{\citenamefont {Simon}(2006)}]{simon2006optimal}%
  \BibitemOpen
  \bibfield  {author} {\bibinfo {author} {\bibfnamefont {D.}~\bibnamefont
  {Simon}},\ }\href@noop {} {\emph {\bibinfo {title} {Optimal state estimation:
  Kalman, H infinity, and nonlinear approaches}}}\ (\bibinfo  {publisher} {John
  Wiley \& Sons},\ \bibinfo {year} {2006})\BibitemShut {NoStop}%
\bibitem [{\citenamefont {Gough}\ and\ \citenamefont
  {James}(2009)}]{gough2009series}%
  \BibitemOpen
  \bibfield  {author} {\bibinfo {author} {\bibfnamefont {J.}~\bibnamefont
  {Gough}}\ and\ \bibinfo {author} {\bibfnamefont {M.~R.}\ \bibnamefont
  {James}},\ }\bibfield  {title} {\bibinfo {title} {The series product and its
  application to quantum feedforward and feedback networks},\ }\href@noop {}
  {\bibfield  {journal} {\bibinfo  {journal} {IEEE Transactions on Automatic
  Control}\ }\textbf {\bibinfo {volume} {54}},\ \bibinfo {pages} {2530}
  (\bibinfo {year} {2009})}\BibitemShut {NoStop}%
\bibitem [{\citenamefont {Williams}(1998)}]{williams1998computation}%
  \BibitemOpen
  \bibfield  {author} {\bibinfo {author} {\bibfnamefont {C.~K.}\ \bibnamefont
  {Williams}},\ }\bibfield  {title} {\bibinfo {title} {Computation with
  infinite neural networks},\ }\href@noop {} {\bibfield  {journal} {\bibinfo
  {journal} {Neural Computation}\ }\textbf {\bibinfo {volume} {10}},\ \bibinfo
  {pages} {1203} (\bibinfo {year} {1998})}\BibitemShut {NoStop}%
\bibitem [{\citenamefont {Neal}(2012)}]{neal2012bayesian}%
  \BibitemOpen
  \bibfield  {author} {\bibinfo {author} {\bibfnamefont {R.~M.}\ \bibnamefont
  {Neal}},\ }\href@noop {} {\emph {\bibinfo {title} {Bayesian learning for
  neural networks}}},\ Vol.\ \bibinfo {volume} {118}\ (\bibinfo  {publisher}
  {Springer Science \& Business Media},\ \bibinfo {year} {2012})\BibitemShut
  {NoStop}%
\bibitem [{\citenamefont {Green}\ and\ \citenamefont
  {Limebeer}(2012)}]{green2012linear}%
  \BibitemOpen
  \bibfield  {author} {\bibinfo {author} {\bibfnamefont {M.}~\bibnamefont
  {Green}}\ and\ \bibinfo {author} {\bibfnamefont {D.~J.}\ \bibnamefont
  {Limebeer}},\ }\href@noop {} {\emph {\bibinfo {title} {Linear robust
  control}}}\ (\bibinfo  {publisher} {Courier Corporation},\ \bibinfo {year}
  {2012})\BibitemShut {NoStop}%
\bibitem [{\citenamefont {Kingma}\ and\ \citenamefont
  {Ba}(2015)}]{kingma2015adam}%
  \BibitemOpen
  \bibfield  {author} {\bibinfo {author} {\bibfnamefont {D.~P.}\ \bibnamefont
  {Kingma}}\ and\ \bibinfo {author} {\bibfnamefont {J.}~\bibnamefont {Ba}},\
  }\bibfield  {title} {\bibinfo {title} {Adam: A method for stochastic
  optimization},\ }in\ \href@noop {} {\emph {\bibinfo {booktitle} {Proceedings
  of the 3rd International Conference on Learning Representations (ICLR)}}}\
  (\bibinfo {year} {2015})\BibitemShut {NoStop}%
\bibitem [{\citenamefont {Wringe}\ \emph {et~al.}(2025)\citenamefont {Wringe},
  \citenamefont {Trefzer},\ and\ \citenamefont
  {Stepney}}]{wringe2025reservoir}%
  \BibitemOpen
  \bibfield  {author} {\bibinfo {author} {\bibfnamefont {C.}~\bibnamefont
  {Wringe}}, \bibinfo {author} {\bibfnamefont {M.}~\bibnamefont {Trefzer}},\
  and\ \bibinfo {author} {\bibfnamefont {S.}~\bibnamefont {Stepney}},\
  }\bibfield  {title} {\bibinfo {title} {Reservoir computing benchmarks: a
  tutorial review and critique},\ }\href@noop {} {\bibfield  {journal}
  {\bibinfo  {journal} {International Journal of Parallel, Emergent and
  Distributed Systems}\ ,\ \bibinfo {pages} {1}} (\bibinfo {year}
  {2025})}\BibitemShut {NoStop}%
\bibitem [{\citenamefont {Barbosa}\ \emph {et~al.}(2021)\citenamefont
  {Barbosa}, \citenamefont {Griffith}, \citenamefont {Rowlands}, \citenamefont
  {Govia}, \citenamefont {Ribeill}, \citenamefont {Nguyen}, \citenamefont
  {Ohki},\ and\ \citenamefont {Gauthier}}]{barbosa2021symmetry}%
  \BibitemOpen
  \bibfield  {author} {\bibinfo {author} {\bibfnamefont {W.~A.}\ \bibnamefont
  {Barbosa}}, \bibinfo {author} {\bibfnamefont {A.}~\bibnamefont {Griffith}},
  \bibinfo {author} {\bibfnamefont {G.~E.}\ \bibnamefont {Rowlands}}, \bibinfo
  {author} {\bibfnamefont {L.~C.}\ \bibnamefont {Govia}}, \bibinfo {author}
  {\bibfnamefont {G.~J.}\ \bibnamefont {Ribeill}}, \bibinfo {author}
  {\bibfnamefont {M.-H.}\ \bibnamefont {Nguyen}}, \bibinfo {author}
  {\bibfnamefont {T.~A.}\ \bibnamefont {Ohki}},\ and\ \bibinfo {author}
  {\bibfnamefont {D.~J.}\ \bibnamefont {Gauthier}},\ }\bibfield  {title}
  {\bibinfo {title} {Symmetry-aware reservoir computing},\ }\href@noop {}
  {\bibfield  {journal} {\bibinfo  {journal} {Physical Review E}\ }\textbf
  {\bibinfo {volume} {104}},\ \bibinfo {pages} {045307} (\bibinfo {year}
  {2021})}\BibitemShut {NoStop}%
\bibitem [{\citenamefont {Forgione}\ and\ \citenamefont
  {Piga}(2021)}]{forgione2021dynonet}%
  \BibitemOpen
  \bibfield  {author} {\bibinfo {author} {\bibfnamefont {M.}~\bibnamefont
  {Forgione}}\ and\ \bibinfo {author} {\bibfnamefont {D.}~\bibnamefont
  {Piga}},\ }\bibfield  {title} {\bibinfo {title} {dynonet: A neural network
  architecture for learning dynamical systems},\ }\href@noop {} {\bibfield
  {journal} {\bibinfo  {journal} {International Journal of Adaptive Control and
  Signal Processing}\ }\textbf {\bibinfo {volume} {35}},\ \bibinfo {pages}
  {612} (\bibinfo {year} {2021})}\BibitemShut {NoStop}%
\bibitem [{\citenamefont {de~Oliveira~Teloli}\ \emph
  {et~al.}(2021)\citenamefont {de~Oliveira~Teloli}, \citenamefont {Villani},
  \citenamefont {da~Silva},\ and\ \citenamefont {Todd}}]{de2021use}%
  \BibitemOpen
  \bibfield  {author} {\bibinfo {author} {\bibfnamefont {R.}~\bibnamefont
  {de~Oliveira~Teloli}}, \bibinfo {author} {\bibfnamefont {L.~G.}\ \bibnamefont
  {Villani}}, \bibinfo {author} {\bibfnamefont {S.}~\bibnamefont {da~Silva}},\
  and\ \bibinfo {author} {\bibfnamefont {M.~D.}\ \bibnamefont {Todd}},\
  }\bibfield  {title} {\bibinfo {title} {On the use of the gp-narx model for
  predicting hysteresis effects of bolted joint structures},\ }\href@noop {}
  {\bibfield  {journal} {\bibinfo  {journal} {Mechanical Systems and Signal
  Processing}\ }\textbf {\bibinfo {volume} {159}},\ \bibinfo {pages} {107751}
  (\bibinfo {year} {2021})}\BibitemShut {NoStop}%
\bibitem [{\citenamefont {Benavoli}\ \emph {et~al.}(2025)\citenamefont
  {Benavoli}, \citenamefont {Piga}, \citenamefont {Forgione},\ and\
  \citenamefont {Zaffalon}}]{benavoli2025dynogp}%
  \BibitemOpen
  \bibfield  {author} {\bibinfo {author} {\bibfnamefont {A.}~\bibnamefont
  {Benavoli}}, \bibinfo {author} {\bibfnamefont {D.}~\bibnamefont {Piga}},
  \bibinfo {author} {\bibfnamefont {M.}~\bibnamefont {Forgione}},\ and\
  \bibinfo {author} {\bibfnamefont {M.}~\bibnamefont {Zaffalon}},\ }\bibfield
  {title} {\bibinfo {title} {dynogp: Deep gaussian processes for dynamic system
  identification},\ }\href@noop {} {\bibfield  {journal} {\bibinfo  {journal}
  {arXiv preprint arXiv:2502.05620}\ } (\bibinfo {year} {2025})}\BibitemShut
  {NoStop}%
\bibitem [{\citenamefont {Koga}\ and\ \citenamefont
  {Yamamoto}(2012)}]{koga2012dissipation}%
  \BibitemOpen
  \bibfield  {author} {\bibinfo {author} {\bibfnamefont {K.}~\bibnamefont
  {Koga}}\ and\ \bibinfo {author} {\bibfnamefont {N.}~\bibnamefont
  {Yamamoto}},\ }\bibfield  {title} {\bibinfo {title} {Dissipation-induced pure
  gaussian state},\ }\href@noop {} {\bibfield  {journal} {\bibinfo  {journal}
  {Physical Review A—Atomic, Molecular, and Optical Physics}\ }\textbf
  {\bibinfo {volume} {85}},\ \bibinfo {pages} {022103} (\bibinfo {year}
  {2012})}\BibitemShut {NoStop}%
\bibitem [{\citenamefont {{\AA}str{\"o}m}\ and\ \citenamefont
  {Murray}(2004)}]{Astrom2004}%
  \BibitemOpen
  \bibfield  {author} {\bibinfo {author} {\bibfnamefont {K.~J.}\ \bibnamefont
  {{\AA}str{\"o}m}}\ and\ \bibinfo {author} {\bibfnamefont {R.~M.}\
  \bibnamefont {Murray}},\ }\href@noop {} {\emph {\bibinfo {title} {Analysis
  and Design of Feedback Systems}}}\ (\bibinfo  {publisher} {Princeton
  University Press},\ \bibinfo {address} {Princeton, NJ},\ \bibinfo {year}
  {2004})\BibitemShut {NoStop}%
\bibitem [{\citenamefont {Rahimi}\ and\ \citenamefont
  {Recht}(2007)}]{rahimi2007random}%
  \BibitemOpen
  \bibfield  {author} {\bibinfo {author} {\bibfnamefont {A.}~\bibnamefont
  {Rahimi}}\ and\ \bibinfo {author} {\bibfnamefont {B.}~\bibnamefont {Recht}},\
  }\bibfield  {title} {\bibinfo {title} {Random features for large-scale kernel
  machines},\ }\href@noop {} {\bibfield  {journal} {\bibinfo  {journal}
  {Advances in neural information processing systems}\ }\textbf {\bibinfo
  {volume} {20}} (\bibinfo {year} {2007})}\BibitemShut {NoStop}%
\end{thebibliography}%

\end{document}